\documentclass[journal,11pt,onecolumn,draftclsnofoot,]{IEEEtran}

\usepackage{caption}
\usepackage{graphicx}
\usepackage{cite}
\usepackage{amsmath, amsthm, amssymb, bm}
\usepackage{amsthm}
\usepackage{algpseudocode}
\usepackage{epstopdf}
\usepackage{color}
\usepackage{subfigure}
\usepackage{epsfig}
\usepackage{flushend}
\usepackage{cleveref}
\usepackage{color}
\newtheorem{theorem}{Theorem}

\newtheorem{proposition}{Proposition}
\theoremstyle{remark}
\newtheorem{remark}{Remark}
\newtheorem{lemma}{Lemma}
\newtheorem{assumption}{Assumption}
\newtheorem{corollary}{Corollary}

\newcommand{\skipval}{0.901mm} 
\usepackage{amsmath}
\usepackage{epstopdf}
\usepackage{amssymb}
\usepackage{systeme}
\usepackage{color}
\usepackage{pgfplots}
\usepackage{bbm}

\long\def\comment#1{}

\DeclareMathOperator*{\argmax}{arg\,max}
\DeclareMathOperator*{\argmin}{arg\,min}

\newfont{\bbb}{msbm10 scaled 700}

\newfont{\bb}{msbm10 scaled 1100}

\newcommand{\gv}{{\bf g}}
\newcommand{\hv}{{\bf h}}

\newcommand{\qv}{{\bf q}}

\newcommand{\sv}{{\bf s}}

\newcommand{\uv}{{\bf u}}
\newcommand{\wv}{{\bf w}}
\newcommand{\vv}{{\bf v}}
\newcommand{\xv}{{\bf x}}
\newcommand{\yv}{{\bf y}}
\newcommand{\zv}{{\bf z}}

\newcommand{\Am}{{\bf A}}

\newcommand{\Cm}{{\bf C}}

\newcommand{\Gm}{{\bf G}}
\newcommand{\Hm}{{\bf H}}
\newcommand{\Id}{{\bf I}}

\newcommand{\Xm}{{\bf X}}
\newcommand{\Ym}{{\bf Y}}
\newcommand{\Zm}{{\bf Z}}

\newcommand{\Deltam}{\hbox{\boldmath$\Delta$}}

\ifCLASSINFOpdf
\else
\fi
\begin{document}
%
\title{Asymptotic Performance of Box-RLS Decoders under Imperfect CSI with Optimized Resource Allocation}
%
%
\author{
Ayed M. Alrashdi, Abla Kammoun, Ali H. Muqaibel and Tareq Y. Al-Naffouri
\vspace{-12pt}
\thanks{
Part of this work was presented in \cite{alrashdi2018optimum}.

A. M. Alrashdi is with the Department of Electrical Engineering, College of Engineering, University of Ha'il, P.O. Box 2440, Ha'il, 81441, Saudi Arabia (e-mail: am.alrashdi@uoh.edu.sa).

A. H. Muqaibel is with the Department of Electrical Engineering, King Fahd University of Petroleum and Minerals, Dhahran 3126, Saudi Arabia (e-mail: muqaibel@kfupm.edu.sa).

A. Kammoun and T. Y. Al-Naffouri are with the Computer, Electrical, and Mathematical Sciences and Engineering Division, King Abdullah University of Science and Technology, Thuwal 23955, Saudi Arabia (e-mail: abla.kammoun@kaust.edu.sa; tareq.alnaffouri@kaust.edu.sa).
}
}
\maketitle
\begin{abstract}
This paper considers the problem of symbol detection in massive multiple-input multiple-output (MIMO) wireless communication systems. We consider hard-thresholding preceded by two variants of the regularized least squares (RLS) decoder; namely the unconstrained RLS and the RLS with box constraint. For all schemes, we focus on the evaluation of the mean squared error (MSE) and the symbol error probability (SEP) for $M$-ary pulse amplitude modulation ($M$-PAM) symbols transmitted over a massive MIMO system when the channel is estimated using linear minimum mean squared error (LMMSE) estimator. Under such circumstances, the channel estimation error is Gaussian which allows for the use of the convex Gaussian min-max theorem (CGMT) to derive asymptotic approximations for the MSE and SEP when the system dimensions and the coherence duration grow large with the same pace. The obtained  expressions are then leveraged to derive the optimal power distribution between pilot and data under a total transmit energy constraint. In addition, we derive an asymptotic approximation  of the goodput for all schemes which is then used to jointly optimize the number of training symbols and their associated power. Numerical results are presented to support the accuracy of the theoretical results.  
\end{abstract}
\begin{IEEEkeywords}
Power allocation, mean squared error, symbol error probability, goodput, regularized least squares, massive MIMO, channel estimation, box-constraint.
\end{IEEEkeywords}
\IEEEpeerreviewmaketitle
\section{Introduction}
\label{sec:intro}
The use of multiple-input multiple-output (MIMO) systems has been recognized as an efficient technology to meet the ever-increasing demand in spectral efficiency. It is indeed known since the early works of Telatar \cite{Tela99} and Foshini \cite{Fosc98} that the mutual information scales with the minimum of the number of transmit and receive antennas. In practice, however, the spectral efficiency of a wireless link depends not only on how many antennas are deployed at the transmit and receive sides but also on the channel estimation accuracy, the detection procedure as well as the distribution of the power resources, all of which have a direct bearing on the end-to-end signal-to-noise-ratio.       
At the receiver side, accurate channel estimation and symbol detection are crucial to reap the gains promised by the additional degrees of freedom offered by MIMO systems. Channel estimation is performed by allocating a training period during which the transmitter sends known \emph{pilot} symbols to the receiver for it to acquire an estimate of the channel state information (CSI). During the data transmission phase, this estimate is leveraged by the receiver to equalize the channel and recover the data symbols. It is worth mentioning that it is important for the success of the data recovery step to acquire accurate channel estimates because otherwise the error in the channel estimation would propagate to the symbol recovery, crippling the overall performance even under state-of-the-art detection strategies.  Resource allocation in terms of power and time  is also an essential part of the design of wireless systems. Increasing the duration of the transmitted pilot sequence would lead to enhanced channel estimation quality but at the cost of spectral efficiency losses since less time is spent to transmit useful data. Moreover, as the total power allocated to data and training transmission is fixed, we cannot increase the power allocated to data without affecting the channel accuracy estimation and vice versa.        

The problem of finding optimal power allocation between pilot and data has received a lot of attention over the last decades. It has been applied to different contexts  with the aim of realizing various communication objectives. In \cite{hassibi2003much}, \cite{kannu2005capacity} and \cite{gottumukkala2009}, the authors proposed  power designs that optimizes bounds on the average channel capacity. A different line of research works in \cite{simko2011, vsimko2012optimal, wang2011power, zhang2016energy} considered the post-equalization signal to interference and noise ratio  as a target metric to determine the optimal power allocation. Depending on the application at hand, different other metrics such as  the bit error rate (BER) \cite{wang2014ber,wang2009super}, the symbol error rate (SER) \cite{cai2004error}, mean squared error (MSE)-related indexes \cite{liu2016pilot,zhao2017game}, max-min fairness utilities \cite{van2018joint} or bounds on the received signal-to-noise ratio \cite{zhou2010two} have been used in several existing papers. Of interest in these works are a wide range of contexts including classical MIMO systems \cite{hassibi2003much}, mutli-carrier systems \cite{rottenberg2016generalized, xu2012energy, ma2003optimal, khosravi2014joint, chen2003pilot, montalban2013power}, amplify-and-forward relaying \cite{wang2014ber}, cognitive radio systems  \cite{yu2016optimal} and very recently massive MIMO systems in both single-cell \cite{cheng2016optimal} and multi-user multi-cell settings \cite{van2018joint, dao2018pilot, ngo2014massive, liu2016pilot, guo2015energy}.

Most of the aforementioned works, despite looking at the power allocation problem from different angles, present the  common denominator of relying on  linear receivers and optimizing some related performance metrics. To achieve optimal performance, it is well known that  non-linear detectors are required, but they are often not implemented as they require a prohibitively high computation complexity. 
In this work, we consider optimization of the power allocation of a non-linear decoder coined Box-RLS decoder, the idea behind which has been proposed in \cite{atitallah2017ber,thrampoulidis2016ber} and earlier in \cite{tan2001constrained}.  As explained in \cite{tan2001constrained}, this decoder is rooted in the formulation of the maximum-likelihood (ML) problem. At the one hand,  the ML decoder  is known to be NP-hard since its solution is constrained to belong to a finite discrete set, at the other hand, the Least Square (LS) decoder or its regularized version referred to as   Regularized Least Squares (RLS) decoder \cite{tikhonov1963solution} are linear detectors, obtained by solving the unconstrained ML problem. The Box-RLS decoder falls between the ML  and the RLS decoders in that the solution is constrained to lie within a closed convex-set. As a consequence, unlike the RLS decoder, the solution of the Box-RLS cannot be expressed in closed-form but is numerically computed using standard convex-optimization tools. Both the RLS and Box-RLS decoders are more computationally efficient than the ML decoder. Besides they are also shown to outperform many heuristic algorithms such as zero-forcing (ZF), successive interference cancellation and decision-feedback, \cite{atitallah2017ber, thrampoulidis2016ber, tan2001constrained}. Performance analysis of the Box-RLS has been carried out in \cite{atitallah2017ber}, \cite{thrampoulidis2016ber} but these studies assume unrealistic scenarios in which the channel is perfectly known and do not address the problem of power allocation. To fill this gap, this work addresses the problem of finding optimal power allocation between pilots and data for RLS and Box-RLS under a fixed total power constraint and when the channel is estimated using the linear minimum mean square error estimator (LMMSE). Particularly, we derive sharp characterization of the MSE and symbol error probability (SEP) for the RLS and Box-RLS decoders under $M$-ary pulse amplitude modulation ($M$-PAM) when assuming that the symbol is recovered by hard-thresholding the output of RLS and Box-RLS decoders. Our analysis, based on the assumption that the MIMO channel and the noise are independent and follow  standard  Gaussian distributions, builds upon the CGMT framework put forth in \cite{thrampoulidis2018precise, alrashdi2017precise,alrashdi2019precise}. As compared to previous works dealing with the use of CGMT in high-dimensional regression problems, our consideration of imperfect CSI poses technical challenges towards assessing the performance of the Box-RLS. Particularly, contrary to previous studies, the application of the CGMT in the asymptotic regime leads to  a non-convex optimization problem; the uniqueness of its solution which is an important step in the analysis becomes thus extremely challenging. This required us to develop new techniques to break these difficulties, which while being accommodated 
to this specific scenario, may be of independent interest. We refer the interested reader to Appendix B devoted to the detailed proofs of our main results.       
To summarize, the main contributions of this work can be listed as follows:
\begin{enumerate}
	\item We derive sharp characterizations of the MSE, SEP and goodput expressions for the RLS, LS and Box-RLS under imperfect CSI. Our expressions shed light on interesting relationships between MSE and SEP for $M$-PAM modulation and under imperfect CSI.  


	\item We determine the optimal power allocation between training and data symbols when SEP or MSE are used as target criteria. 
	\item We optimize the power and the number of pilot symbols to maximize the goodput for all studied decoders.
	 

\end{enumerate}

To the best of our knowledge, none of the above was previously derived for the $M$-PAM case in the presence of imperfect CSI.
\subsection{Paper Organization}
This paper is organized as follows. The system model is presented in Section II. In Section III, we discuss channel estimation, the properties of the estimator and the channel estimation error, as well as symbol estimation. The MSE/SEP of symbol estimation is derived in Section IV by applying the CGMT. These expressions are then validated through an assortment of numerical results and  leveraged to find optimal power strategies in Section V. The key ingredient of the analysis which is the CGMT is reviewed in Appendix A. The proofs for the derived MSE and SEP are presented in Appendix B and Appendix C for the Boxed RLS and un-boxed RLS, respectively.
\subsection{Notations}
Scalars are denoted by lower-case letters (e.g., $\alpha$), column vectors are represented by boldface lowercase
letters (e.g., $\xv$), whereas matrices are denoted by boldface upper-case letters (e.g., $\Xm$). The notations $(\cdot)^T$ and $(\cdot)^{-1}$ denote the transpose and inversion operators, respectively. The $j$-th element of vector ${\bf x}$ will be denoted by ${x}_j$.  The symbol $\Id_N$ is used to represent the identity matrix of dimension $N \times N$. We use the standard notation $\mathbb{P}[\cdot]$ and $\mathbb{E}[\cdot]$ to denote probability and expectation. We write $X \thicksim p_X$ to denote that a random variable $X$ has a probability density function (pdf) $p_X$. In particular, $G \thicksim  \mathcal{N}(\mu, \sigma^2)$ implies that $G$ has a Gaussian (normal) distribution of mean $\mu$ and variance $\sigma^2$. $p(x)= \frac{1}{\sqrt{2 \pi}} e^{\frac{-x^2}{2}}$ and $Q(x) = \frac{1}{\sqrt{2\pi}}\int_{x}^{\infty} e^{-t^2/2} {\rm{d}} t$ denote the pdf of a standard normal distribution and its associated $Q$-function respectively.
Finally, $\| \cdot \|$ indicates the Euclidean norm (i.e., the $\ell_2$-norm) of a vector and $\| \cdot \|_\infty$ represents its $\ell_\infty$-norm.\footnote{For a vector $\xv$, $\| \xv \|_\infty = \max_{j} | x_j|$.}
\begin{table}[ht]
\centering
\caption{Summary of Main Variables}
\begin{tabular}{|l|l|}
\hline
\textbf{Symbol} &  \quad \quad \quad \quad \textbf{Meaning}  \\
\hline
$T_{\tiny{p}}$ & number of pilot symbols \\
\hline
$T_ {d}$ & number of data symbols \\ 
\hline
$T$ & total number of symbols \\
\hline
$M$ & size of the PAM signal \\ 
\hline
$ \mathcal{E}$ & average power of non-normalized $M$-PAM signal \\
\hline
$\rho_ {p}$ & average training power \\
\hline
$\rho_ {d}$ & average data power \\
\hline
$\rho$ & average received power per receive antenna \\
\hline
$\rho_\text{eff}$ & effective SNR \\
\hline
$\alpha$ & data power ratio \\
\hline
$K$ & Number of transmit (Tx) antennas \\
\hline
$N$ & Number of receive (Rx) antennas\\
\hline
$\delta$ & ratio of Rx to Tx antennas $N/K$ \\
\hline
$ \lambda$ & regularization coefficient \\
\hline
$t$ & box-constraint threshold  \\
\hline
$ \Hm$ & channel matrix \\
\hline
{{$ \widehat{\Hm}$}} & estimated channel matrix \\
\hline
$ \Deltam$ & channel estimation error matrix \\
\hline
$\sigma_{ {\Delta}}^2$ & variance of error matrix  \\
\hline
$\sigma_{ {\hat{H}}}^2$ & variance of estimated channel matrix  \\
\hline
$\Xm_ {p}$ & pilot symbols matrix\\
\hline
$\xv_{0}$ & transmitted data symbol vector \\
\hline
$\widehat{\xv}$ & estimated data symbol vector \\
\hline
$\yv$ & received data symbol vector \\
\hline
$\zv$ & noise vector \\
\hline
\end{tabular}
\label{tab:1}
\end{table}
\section{System Model}
\label{sec:model}
We consider a flat block-fading massive MIMO system with $K$ transmitter antennas and $N$ receiver antennas. The transmission consists of $T$ symbols that occur in a time interval within which the channel is assumed to be static. A number $T_ {p}$ pilot symbols (for channel estimation) occupy the first part of the transmission interval with power, $\rho_ {p}$. The remaining part is devoted for transmitting $T_ {d}= T - T_ {p}$ data symbols with power, $\rho_{ {d}}$. Figure \ref{fig:channels} illustrates the  system model. It implies from conservation of time and energy that:
\begin{equation}
\label{eq:energy conseve}
\rho_ {p} T_ {p} + \rho_{ {d}} T_ {d}= \rho T,
\end{equation}
where $\rho$ is the expected average power.
Alternatively, we have $ \rho_ {d} T_ {d} = \alpha \rho T$, where $\alpha \in (0,1)$ is the ratio of the power allocated to the data so that 
\begin{equation}
\label{eq:data power}
\rho_ {p} T_ {p} =  (1- \alpha) \rho T.
\end{equation}
The received signal model for the \emph{data} transmission phase is given by
\begin{equation}
\label{eq:data model}
\yv = \sqrt{\frac{\rho_ {d}}{K}}  \Hm \xv_{0}+ \zv,
\end{equation}
where $\yv \in \mathbb{R}^{N}$ is the received data symbol vector, $\xv_0 \in \mathbb{R}^{K}$ is the transmitted data symbol vector, $\Hm\in \mathbb{R}^{N\times K}$ is a channel matrix with i.i.d. Gaussian elements $h_{ij}\thicksim \mathcal{N}\left(0, 1\right)$, and $\zv \in \mathbb{R}^{N}$ stands for the additive Gaussian noise at the receiver with i.i.d. elements of mean $0$ and variance 1. It is assumed that $\xv_0$ has i.i.d. $M$-PAM symbols normalized to have unit variance ($\mathbb{E}[\xv_0 \xv_0^T] = \Id_K$), such that each transmit antenna sends a data symbol $x_{0,j}$ that takes values (with equal probability $1/M$) in the set:
\begin{equation}
x_{0,j} \in \mathcal{C}:= \bigg\{\pm\frac{1}{\sqrt{\mathcal{E}}}, \pm\frac{3}{\sqrt{\mathcal{E}}}, \cdots, \pm\frac{(M-1)}{\sqrt{\mathcal{E}}}  \bigg\}, j=1,2,\cdots,K,
\end{equation}
{{where $\mathcal{E} = \frac{M^2-1}{3}$ is the average power of the non-normalized $M$-PAM signal}}, $M=2^{b}$ being the modulation order and $b$ the number of bits carried by each symbol.

As the channel matrix $\Hm$ is unknown to the receiver, a training phase during which the transmitter sends $T_{ {p}} \geq K$ \emph{pilot} symbols is dedicated. 
The received signals corresponding to this phase can be modeled as
\begin{equation}
\label{eq:pilot model}
\Ym_ {p} = \sqrt{\frac{\rho_ {p}}{K}}  \Hm \Xm_ {p} + \Zm_ {p},
\end{equation}
where $\Ym_ {p}\in \mathbb{R}^{N \times T_ {p}}$ is the received signal matrix, $\Xm_ {p}\in \mathbb{R}^{K \times T_ {p}}$ is the matrix of transmitted pilot symbols, and $\Zm_ {p}\in \mathbb{R}^{N \times T_ {p}}$  stands for the additive Gaussian noise with $\mathbb{E}[\Zm_{ {p}} \Zm_{ {p}}^T] =T_{ {p}} \Id_N$.

For the reader convenience, we summarize in Table \ref{tab:1} the notation symbols of the parameters used in this paper.
	\begin{figure}
		\begin{center}
		\includegraphics[width=9.0cm]{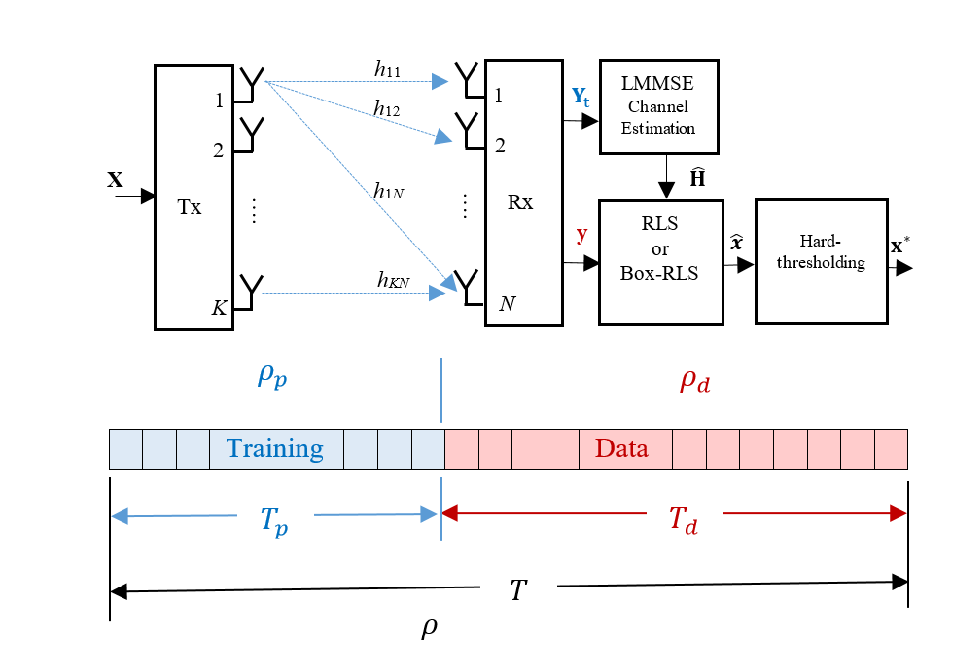}
		\caption{\scriptsize{Massive MIMO system with training-based transmission.}}
		\label{fig:channels}
		\end{center}
	\end{figure}
\section{MIMO Symbols Detection under LMMSE Channel Estimation}
\label{sec:ch estimate}
\subsection{LMMSE Channel Estimation}
%
Based on the knowledge of $\Ym_ {p}$ from \eqref{eq:pilot model}, the LMMSE channel estimate  is given by \cite{kay1993fundamentals}
\begin{align}
\label{eq:LMMSE}
\widehat{\Hm} &= \sqrt{\frac{K}{\rho_ {p}}}\Ym_ {p} \Xm_{ {p}}^T \left( \frac{K}{\rho_ {p}} \Id_ {K}+ \Xm_{ {p}} \Xm_{ {p}}^T \right)^{-1},\nonumber \\
& = \Hm - \Deltam,
\end{align}
where $\Deltam$ is the zero-mean channel estimation error matrix, which is independent of $\widehat{\Hm}$, as per the orthogonality principle of the LMMSE estimation \cite{hassibi2003much}, \cite{kay1993fundamentals}. For MIMO channels with i.i.d. entries, it has been proved that the optimal $\Xm_ {p}$ that minimizes the estimation mean squared error under a total power constraint satifies \cite{hassibi2003much}
\begin{equation}
\label{eq:pilot orthogonality}
 \Xm_ {p} \Xm_{ {p}}^{T}= T_ {p}\Id_ {K}.
\end{equation}
For the above condition to hold, the number of training symbols should be greater than or equal to $K$. Moreover, under \eqref{eq:pilot orthogonality},  the channel estimate $\widehat{\Hm}$ has i.i.d. zero-mean Gaussian entries with mean $0$ and variance  $\sigma_{ {\hat{H}}}^2 = 1-\sigma_{ {\Delta}}^2$ \cite{hassibi2003much}, with
\begin{equation}
\label{eq:error matrix variance}
\sigma_{ {\Delta}}^2 = \frac{1}{ 1+ \frac{\rho_ {p}}{K} T_ {p} },
\end{equation}
being the variance of each element in $\Deltam$.
It appears from \eqref{eq:error matrix variance} that channel estimation error decreases with the pilot energy given by $ \rho_ {p} T_ {p}$. 
\subsection{Symbol Detection under Imperfect Channel Estimation}


With the channel estimate  $\widehat{\Hm}$ at hand, the receiver can proceed to the recovery of the transmitted symbols. The optimal decoder which minimizes the probability of error is the maximum likelihood (ML)  decoder under perfect channel knowledge which is given by:
\begin{equation}
	\widehat{\xv}_{\text{\tiny{ML}}} = \underset{{\xv \in \mathcal{C}^{ {K}}}}{\argmin} \big\| \yv -\sqrt{\frac{\rho_ {d}}{K}} \Hm \xv \big\|^2.
	\label{eq:ML}
\end{equation} 
As can be seen, the ML decoder involves a combinatorial optimization problem. It presents thus a prohibitively high computational complexity, especially when the system dimensions become large as envisioned by current communication systems. 
 To overcome this issue, suboptimal strategies that require less computational complexity are in general used. They often proceed in two steps. First, a real-valued approximation of the transmitted symbol is obtained. This estimate is then hard-thresholded in a second step to produce the final estimate. In this work, the focus is on the regularized least squares (RLS) and the Box-regularized least squares (Box-RLS) decoders. 
    

The RLS decoder is based on regularizing the cost in \eqref{eq:ML} and relaxing the finite-alphabet constraint, thus leading to:
\begin{subequations}
\begin{align}\label{eq:LS matrix}
\widehat{\xv} &= \argmin_{\xv \in \mathbb{R}^{ {K}}}\| \yv -  \sqrt{\frac{\rho_ {d}}{K}}  \widehat{\Hm}  \xv \|^2 + \lambda \rho_ {d}\| \xv \|^2,  \\
&= (\Am^T \Am + {\lambda} \rho_ {d} \Id_ {K})^{-1} \Am^T \yv,\label{eq:LS matrixb}\\
	{x}_j^* &= \argmin_{s \in \mathcal{C}} \bigg| \frac{\widehat{x}_j}{B} -s \bigg|, j=1,2,\cdots,K,\label{eq:LS matrixc}
\end{align}
\end{subequations}
where $\lambda \geq 0$ is the regularization coefficient, $\Am=  \sqrt{\frac{\rho_ {d}}{K}} \widehat{\Hm}$ and $B$ is a normalization constant, the value of which will be suggested from our analysis so as  to remove the bias of the decoder \cite{practical_guide} \footnote{It follows from our analysis that $\widehat{x}_j$ behaves as $B {x}_{0,j}$ plus some independent Gaussian noise where $B$ is some constant depending on the regularization factor, the data and channel estimation powers. It is thus sensible to divide $\widehat{x}_{j}$ by $B$ to remove the induced bias.}. Note that the optimization in \eqref{eq:LS matrixc} simply selects the symbol value $s$ that is closest to the solution $\frac{\widehat{ x}_j}{B}$ among a total of $M$ possible choices. As can be seen from \eqref{eq:LS matrixb}, the elements of the solution $\widehat{\xv}$ may take large values in high-noise conditions or poor channel estimation scenarios. This motivates the Box-RLS decoder  \cite{yener2002cdma, tan2001constrained, ma2002quasi}  given by
\begin{subequations}
\begin{align}\label{eq:Box-RLS matrix}
\widehat{\xv} &= \argmin_{\|\xv\|_{\infty}\leq t} \|\yv -  \sqrt{\frac{\rho_ {d}}{K}} \widehat{\Hm}   \xv \|^2 + \lambda \rho_ {d}\| \xv\|^2, \\
	&  { x}_j^* = \argmin_{s \in \mathcal{C}} \bigg|\frac{\widehat{x}_j}{B} -s \bigg|, j=1,2,\cdots,K,\label{eq:BLS matrixb}
\end{align}
\end{subequations}
 which is based on relaxing the finite-alphabet constraint to the convex constraint $\xv \in [-t,t]^K$, where $t$ is a fixed threshold that can be optimally tailored according to the propagation scenario. 
\subsection{Performance Metrics}
This work considers the performance evaluation of the RLS and the Box-RLS decoders in terms of three different performance metrics, which are:

\noindent\textbf{Mean Squared Error}: A natural and heavily used measure of performance is the reconstruction \textit{mean squared error} (MSE), which measures the deviation of $\widehat{\xv}$ from the true signal $\xv_0$. This assesses the performance of the \textit{first} step of the decoding algorithms. Formally, the MSE is defined as
\begin{equation}
{\rm{MSE}} := \frac{1}{K}\| \widehat{\xv}  - \xv_0\|^2.
\end{equation}
\textbf{Symbol Error Probability}: The symbol error rate (SER) characterizes the performance of the detection process and is defined as:
\begin{equation}
{\rm{SER}} := \frac{1}{K} \sum_{j=1}^{K} {\bf 1}_{\{{x}^*_j \neq x_{0,j} \}},
\end{equation}
where ${\bf 1}_{\{\cdot\}}$ indicates the indicator function.\\
In relation to the SER is the symbol error probability (SEP) which is defined as the expectation of the SER averaged over the noise, the channel and the constellation. Formally, the symbol error probability denoted by ${\rm SEP}$ is given by:
\begin{equation}
	{\rm SEP}:=\mathbb{E}[{\rm SER}] =\frac{1}{K}\sum_{j=1}^K \mathbb{P}\left[{x}^*_j \neq x_{0,j} \right].
	\label{eq:SEP}
\end{equation}
\textbf{Goodput}: The goodput is a performance measure that accounts for the amount of useful data transmitted, divided by the time it takes to \textit{successfully} transmit it. The amount of data considered excludes protocol overhead bits as well as retransmitted data packets \cite{grote2008ieee}. In our context, it can be defined as
\begin{equation}\label{good_def}
	G := \left( \frac{T -T_ {p}}{T}\right) \left(1- {\rm SEP} \right).
\end{equation}
Goodput and throughput are connected performance parameters in that the throughput can be obtained by simply dividing the goodput by the data transmission rate.  
\section{Analysis of the Mean Squared Error (MSE) and Symbol-Error Probability (SEP)}
\label{sec:BER}
In this section, we derive asymptotic expressions of the MSE and SEP for the RLS and Box-RLS decoders. Particularly, we show that these metrics can be approximated by deterministic quantities that involve the power and time devoted for data and training transmissions. 
Our analysis builds upon the CGMT framework. For the RLS decoder, the same results could have been obtained using tools from random matrix theory as the decoder possesses a closed-form expression. However, since the use of the CGMT framework is more adapted to the Box-RLS decoder that cannot be expressed in closed-form, we rely in this work on the CGMT framework for both decoders for the sake of a unified presentation. 

Prior to stating our main results, we shall introduce the following assumptions which describe the considered growth rate regime:
\subsection{Technical Assumptions} 
\begin{assumption}
 We consider the asymptotic regime in which the system dimensions $K$ and $N$ grow simultaneously to infinity at a fixed ratio $$\delta :=\frac{N}{K} \in (0,\infty).$$ 
\end{assumption}
\begin{assumption}
 We assume a fixed normalized coherence interval $$\tau:= \frac{T}{K}\in (1,\infty).$$ 
and that the pilot and data symbols grow proportionally with $K$, where:
	$$\tau_{ {p}} :=\frac{T_{ {p}}}{K} \in [1,\infty),$$ and $$\tau_{ {d}} :=\frac{T_{ {d}}}{K},$$ are fixed and denote the normalized number of pilot and data symbols, respectively.  
\end{assumption}
In the sequel, we leverage the statistical distribution of the channel and the channel estimate as well as the asymptotic regime specified in Assumption 1 and 2 to provide asymptotic approximations of  the MSE and SEP for RLS and Box-RLS.  We use the standard notation ${\rm{plim}}_{n\to\infty}\ X_n = X$ to denote that a sequence of random variables $X_n$ converges in probability towards a constant $X$.
\subsection{MSE and SEP Analysis for RLS}
We provide herein asymptotic approximations of the MSE and SEP for RLS under imperfect channel state information.  
The derived closed form expressions are given in Theorem \ref{MSE-RLS}, and Theorem \ref{SER-RLS} while the proof is given in Appendix~\ref{sec:unbox}.
\begin{theorem}(MSE for RLS): \label{MSE-RLS} \normalfont
Fix $\lambda>0$, $\delta>0$, and let $\widehat{\xv}$ be a minimizer of the RLS problem in (\ref{eq:LS matrix}). 
Define
\begin{equation*}
\Upsilon(\lambda,\delta)=\frac{-(\delta-\frac{\lambda}{ \sigma_{ {\hat{H}}}^2}-1)+\sqrt{(\delta-\frac{\lambda}{ \sigma_{ {\hat{H}}}^2}-1)^2+4\frac{\lambda}{\sigma_{ {\hat{H}}}^2}\delta}}{2\delta},
\end{equation*}
\label{th:MSE_RLS}
and
	\begin{equation*}
		\theta_\star = \sqrt{\frac{\rho_ {d} \sigma_{ {\hat{H}}}^2 \left(\frac{\Upsilon(\lambda,\delta)}{1+\Upsilon(\lambda,\delta)}\right)^2+  \rho_ {d} \sigma_{ {\Delta}}^2 +1 }{\delta-\frac{1}{(1+\Upsilon(\lambda,\delta))^2}}}.
		\end{equation*}
	Then, under Assumption 1 and Assumption 2, it holds:
\begin{equation}\label{RLS-MSE}
\underset{K\to\infty}{{\rm{plim}}} {\rm{ MSE}} = \frac{1}{\rho_ {d} \sigma_{ {\hat{H}}}^2} \biggl(\delta \theta_{\star}^2  - \rho_ {d}   \sigma_{ {\Delta}}^2 -1 \biggr),
\end{equation}
\label{th:MSE_RLS}
\end{theorem}
\begin{proof}
The proof of  Theorem \ref{MSE-RLS} is given in Appendix~\ref{sec:unbox}.
\end{proof}
It is worth mentioning that the above formula is not restricted to $\xv_0$ belonging to $M$-PAM constellation and is valid for $\xv_0$ from any distribution provided that $\xv_0$ is normalized to have unit-variance. However, assuming that $\xv_0$ is drawn from $M$-PAM constellations, the SEP can be approximated as:
\begin{theorem}(SEP of RLS):\label{SER-RLS} \normalfont
Under the same setting of Theorem~\ref{MSE-RLS}, the SEP under $M$-PAM modulation of the RLS decoder employing the normalization constant
	$$
	B=\frac{\sigma_{ {\hat{H}}}^2\frac{\beta_\star}{\theta_\star}}{\sigma_{ {\hat{H}}}^2\frac{\beta_\star}{\theta_\star}+2\lambda},
	$$
	with $\beta_\star=\frac{2\lambda\theta_\star}{\sigma_{ {\hat{H}}}^2\Upsilon(\lambda,\delta)}$,
	converges to:
\begin{align}\label{SER-RLS1}
	&\underset{K\to\infty}{{\rm{plim}}} {\rm{SEP}}  = \widetilde{\rm {SEP}}_{\rm RLS},
\end{align}
where
	\begin{align}
		\widetilde{\rm {SEP}}_{\rm RLS}=2\bigg(1- \frac{1}{M}\bigg) Q \biggl( \frac{\sqrt{\rho_{ {d}} \sigma_{ {\hat{H}}}^2}}{\sqrt{\mathcal{E}}\theta_\star}  \biggr).
	\end{align}
	where $\theta_\star$ is defined in Theorem \ref{MSE-RLS}.
\label{th:SER_RLS}
\end{theorem}
\begin{proof}
	A sketch of the proof is provided in Appendix \ref{sec:unbox}. 
\end{proof}
Before proceeding further, we validate the approximations provided in Theorem \ref{th:MSE_RLS} and Theorem \ref{th:SER_RLS}. To this end, we report in Figure \ref{xz2} and \ref{fig221} the MSE and SEP for the RLS decoder when $K=400, \delta=1.2,  T_ {p} = 456, T=1000$,  $\alpha = 0.5$ and $M=2$ (corresponding to  Binary Phase Shift Keying (BPSK) modulation), as a function of the average power $\rho$. As seen, the simulation results, averaged over 500 realizations of the channel,  show a perfect agreement with the theoretical results. 

\begin{figure}
\begin{center}
%
%
\definecolor{mycolor1}{rgb}{0.00000,0.45098,0.74118}%
\definecolor{mycolor2}{rgb}{0.85098,0.32941,0.10196}%
\definecolor{mycolor3}{rgb}{0.63529,0.07843,0.18431}%
\definecolor{mycolor4}{rgb}{1.00000,0.00000,1.00000}%
\definecolor{OliveGreen}{rgb}{0,0.5,0}%
\begin{tikzpicture}[scale=1,font=\small]
    \renewcommand{\axisdefaulttryminticks}{4}
    \tikzstyle{every major grid}+=[style=densely dashed]
    \tikzstyle{every axis y label}+=[yshift=-10pt]
    \tikzstyle{every axis x label}+=[yshift=5pt]
    \tikzstyle{every axis legend}+=[cells={anchor=west},fill=white,
        at={(0.98,0.98)}, anchor=north east, font=\tiny ]
        
       \begin{axis}				[
      xmin=0,
       xmax=35,
      ymin=0,
      ymax=0.9,
      grid=major,
      scaled ticks=true,
       xlabel={$\rho$ (dB)},
       ylabel={$\rm MSE$},
   			]
\addplot [color=red, line width=1.0pt]
  table[row sep=crcr]{%
0	0.871212202685771\\
1	0.833246210722417\\
2	0.789615492297682\\
3	0.741123719713622\\
4	0.68883449134569\\
5	0.633901262880953\\
6	0.577435940482042\\
7	0.520444321203\\
8	0.463820843797481\\
9	0.408379789072107\\
10	0.354897176634555\\
11	0.304138570539753\\
12	0.256851554403679\\
13	0.213713366654539\\
14	0.175247668997248\\
15	0.141747458444321\\
16	0.113240350519047\\
17	0.0895059891235959\\
18	0.0701311101522994\\
19	0.0545819874303803\\
20	0.0422767820546166\\
21	0.0326447148617247\\
22	0.0251656421567238\\
23	0.0193903487730797\\
24	0.0149459134917745\\
25	0.0115315743646757\\
26	0.00890975679755877\\
27	0.00689550996954936\\
28	0.00534625416029225\\
29	0.00415275368428929\\
30	0.0032316136912938\\
31	0.00251926670567999\\
32	0.00196726580944954\\
33	0.0015386567089992\\
34	0.00120520846132127\\
35	0.000945311881560552\\
};
\addlegendentry{Theory: Box-RLS; BPSK}
\addplot [only marks, line width =1pt, mark size=1pt, mark=o, mark options={solid, OliveGreen}]
  table[row sep=crcr]{%
0	0.868874658597943\\
1	0.830823375490779\\
2	0.788800105583126\\
3	0.742100074972399\\
4	0.686910585456143\\
5	0.632843671406033\\
6	0.575433153897887\\
7	0.51757579345007\\
8	0.463939521651431\\
9	0.411084127482249\\
10	0.358748599110781\\
11	0.303702917119962\\
12	0.259701889135899\\
13	0.213899694578417\\
14	0.174595502345566\\
15	0.142213476522434\\
16	0.114720578064387\\
17	0.08995263553743\\
18	0.0706426757515566\\
19	0.054760184006676\\
20	0.0423776666427946\\
21	0.0328021583758602\\
22	0.025442744204253\\
23	0.0194510777794442\\
24	0.0151842752329766\\
25	0.0117111088910959\\
26	0.00907098093865294\\
27	0.00688452214640817\\
28	0.00534947745245011\\
29	0.00420925773458161\\
30	0.00323407818600913\\
31	0.00254092243032044\\
32	0.00196133524041875\\
33	0.00154084096065981\\
34	0.00120664726114794\\
35	0.000964901473619992\\
};
\addlegendentry{Simulation: Box-RLS; BPSK}

\addplot [color=black, dashdotted, line width=1.0pt]
  table[row sep=crcr]{%
0	0.871437073370017\\
1	0.833967209385253\\
2	0.791324365730621\\
3	0.744429689589993\\
4	0.694413620683675\\
5	0.642474118340835\\
6	0.589752229154876\\
7	0.537249879055271\\
8	0.485794863988491\\
9	0.436041351695187\\
10	0.388489516238648\\
11	0.343511177752101\\
12	0.301373702979184\\
13	0.262258973939979\\
14	0.226277384459951\\
15	0.193478823812637\\
16	0.163863349867904\\
17	0.137393068772852\\
18	0.114003376487199\\
19	0.0936079587950949\\
20	0.076092003140543\\
21	0.0612955263975044\\
22	0.0489998231718719\\
23	0.0389319774671206\\
24	0.0307886431960565\\
25	0.024265468161524\\
26	0.0190790920981082\\
27	0.0149787038694583\\
28	0.0117500141105737\\
29	0.00921468081299643\\
30	0.00722719150840245\\
31	0.00567055942687679\\
32	0.00445173672275899\\
33	0.0034972766699961\\
34	0.00274950461211057\\
35	0.00216328277075067\\
};
\addlegendentry{Theory: Box-RLS; 4-PAM}

\addplot [only marks, line width =1pt, mark size=1.2pt, mark=diamond, mark options={solid, blue}]
  table[row sep=crcr]{%
0	0.872087483056199\\
1	0.838215003021759\\
2	0.792718292367387\\
3	0.742365916414344\\
4	0.693533667386333\\
5	0.642987633367621\\
6	0.588098737554053\\
7	0.533954786656438\\
8	0.485991399272924\\
9	0.435683897949337\\
10	0.38905423004054\\
11	0.341879337809327\\
12	0.301694279563105\\
13	0.262892205538514\\
14	0.227963296443368\\
15	0.193781234908991\\
16	0.16441514451878\\
17	0.138804512823268\\
18	0.115220105856301\\
19	0.0931668644096088\\
20	0.0764905597853412\\
21	0.0614979175324972\\
22	0.049737760323294\\
23	0.0389979767669223\\
24	0.0312859743231903\\
25	0.0240899080654295\\
26	0.0191815972543288\\
27	0.0150082712994617\\
28	0.0119753472738706\\
29	0.00927168403785697\\
30	0.00730538451074182\\
31	0.00572645503979219\\
32	0.00453187561298582\\
33	0.0034639541433993\\
34	0.00274250351161656\\
35	0.00220342030542312\\
};
\addlegendentry{Simulation: Box-RLS; 4-PAM}

\addplot [color=mycolor4, dashed, line width=1.0pt]
  table[row sep=crcr]{%
0	0.871446072678727\\
1	0.83402783533993\\
2	0.791568725279436\\
3	0.745122063159603\\
4	0.695948460707339\\
5	0.645332509678497\\
6	0.594442537355299\\
7	0.54425266239889\\
8	0.495519859221039\\
9	0.448797037747719\\
10	0.404463487622635\\
11	0.36275967912656\\
12	0.323819403046235\\
13	0.287696519183532\\
14	0.254385987629118\\
15	0.223839905167041\\
16	0.195979545129218\\
17	0.170704321165657\\
18	0.147898404716632\\
19	0.127435538501213\\
20	0.10918244834212\\
21	0.0930011715704469\\
22	0.0787505839609838\\
23	0.0662874020557637\\
24	0.0554669417298609\\
25	0.0461439021455937\\
26	0.0381733954740033\\
27	0.0314123483189251\\
28	0.0257212726167438\\
29	0.0209662738459123\\
30	0.0170210722015505\\
31	0.0137687861471301\\
32	0.0111032703921362\\
33	0.00892988974130204\\
34	0.00716571221438941\\
35	0.00573918927522137\\
};
\addlegendentry{Theory: RLS}

\addplot [only marks, line width =1pt, mark size=0.95pt, mark=square, mark options={solid, mycolor1}]
  table[row sep=crcr]{%
0	0.868232310127075\\
1	0.836390729922559\\
2	0.790747613780211\\
3	0.746499240999121\\
4	0.692110868099238\\
5	0.647234808251292\\
6	0.595623772973538\\
7	0.550840747476073\\
8	0.498225294335201\\
9	0.456592322635981\\
10	0.405656239498327\\
11	0.365305865133001\\
12	0.327661662260216\\
13	0.288200649355961\\
14	0.252661904267047\\
15	0.224168029295299\\
16	0.195426783724478\\
17	0.174736919320239\\
18	0.146928422835637\\
19	0.128935530178427\\
20	0.110148003981192\\
21	0.0939440549239677\\
22	0.079100403868471\\
23	0.0668051532280007\\
24	0.0577706349725925\\
25	0.0476144916372671\\
26	0.0384035222455941\\
27	0.0313834514001114\\
28	0.0256315737936198\\
29	0.0208632212278472\\
30	0.0172028722835446\\
31	0.0135702626693177\\
32	0.0112534629477695\\
33	0.00915163654733009\\
34	0.00727465281478669\\
35	0.0057292443800899\\
};
\addlegendentry{Simulation: RLS}
\end{axis}

\end{tikzpicture} \vskip-4mm
\caption{\scriptsize {MSE performance of both RLS and Box-RLS decoders.}}%
\label{xz2}
\end{center}
\end{figure}
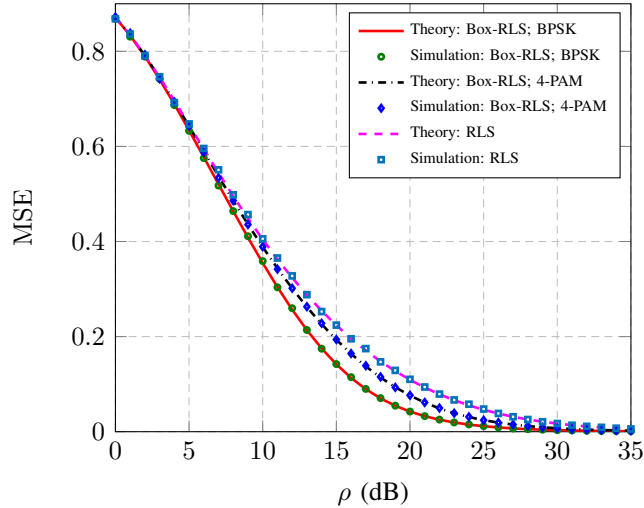

\begin{figure}
\begin{center}
\definecolor{OliveGreen}{rgb}{0.,0.5,0}%
\begin{tikzpicture}
    \renewcommand{\axisdefaulttryminticks}{4}
    \tikzstyle{every major grid}+=[style=densely dashed]
    \tikzstyle{every axis y label}+=[yshift=-10pt]
    \tikzstyle{every axis x label}+=[yshift=5pt]
    \tikzstyle{every axis legend}+=[cells={anchor=west},fill=white,
        at={(0.01,0.01)}, anchor=south west, font=\tiny ]
\begin{semilogyaxis}[
xmin=0,
xmax=35,
xlabel={$\rho$ (dB)},
ymin=1e-08,
ymax=1,
 grid=major,
 scaled ticks=true,
ylabel={$\rm SEP$},
]

\addplot [color=blue, dashed, line width=1.0pt]
  table[row sep=crcr]{%
0	0.350459340292567\\
1	0.327764311848152\\
2	0.303926176211485\\
3	0.279320706337796\\
4	0.254314158316087\\
5	0.229242888725249\\
6	0.204406786760934\\
7	0.180073309514013\\
8	0.15648712837236\\
9	0.133880765103076\\
10	0.112482833893392\\
11	0.0925217103150944\\
12	0.0742233021576124\\
13	0.0578021802676642\\
14	0.0434459087685994\\
15	0.0312933201890589\\
16	0.0214089981539884\\
17	0.0137584510384389\\
18	0.00819101549419987\\
19	0.00443925225095467\\
20	0.00214242164474999\\
21	0.000895369378812596\\
22	0.000312783009952526\\
23	8.73276360949533e-05\\
24	1.84089840021163e-05\\
25	2.72630922944329e-06\\
26	2.58913272686158e-07\\
27	1.4049048537798e-08\\
28	3.76485824561519e-10\\
29	4.14564850454516e-12\\
30	1.48766661366215e-14\\
31	1.29918691063005e-17\\
32	1.91165053296695e-21\\
33	2.98331362174482e-26\\
34	2.75751521374633e-32\\
35	7.2508171277502e-40\\
};
\addlegendentry{Theory: RLS}
\addplot [color=red, line width=1.0pt]
  table[row sep=crcr]{%
0	0.350457294241988\\
1	0.327753869309766\\
2	0.303887706964755\\
3	0.279208559407752\\
4	0.254039230640741\\
5	0.228652819430034\\
6	0.203267789562573\\
7	0.178060056621547\\
8	0.153189736493996\\
9	0.128840221301295\\
10	0.105266338996575\\
11	0.08284346209772\\
12	0.0620987274007333\\
13	0.0436925267371205\\
14	0.0283163162140137\\
15	0.0164998657742636\\
16	0.00838245060260327\\
17	0.00357064911454958\\
18	0.00121384188672377\\
19	0.000309424762544297\\
20	5.46775883898135e-05\\
21	6.06616166277274e-06\\
22	3.72931776763772e-07\\
23	1.08526141113352e-08\\
24	1.22546152720128e-10\\
25	4.17823828367716e-13\\
26	3.13450117928608e-16\\
27	3.47074559128233e-20\\
28	3.427905354052e-25\\
29	1.6000686196605e-31\\
30	1.58464101983168e-39\\
31	1.21313235359248e-49\\
32	2.01068942655409e-62\\
33	1.45087587121845e-78\\
34	6.03839370290781e-99\\
35	1.13531914618792e-124\\
};
\addlegendentry{Theory: Box-RLS}
\addplot [only marks, line width =1pt, mark size=1.5pt, mark=o, mark options={solid, OliveGreen}]
  table[row sep=crcr]{%
0	0.35029\\
1	0.329365\\
2	0.304495\\
3	0.279165\\
4	0.254035\\
5	0.228675\\
6	0.204945\\
7	0.1804\\
8	0.156435\\
9	0.133495\\
10	0.11249\\
11	0.092495\\
12	0.07419\\
13	0.0587350000000001\\
14	0.044275\\
15	0.03324\\
16	0.022035\\
17	0.01348\\
18	0.00859999999999996\\
19	0.00450999999999997\\
20	0.00217999999999999\\
21	0.001145\\
22	0.000415\\
23	0.000115\\
24	1e-05\\
25	5e-06\\
26	0\\
27	0\\
28	0\\
29	0\\
30	0\\
31	0\\
32	0\\
33	0\\
34	0\\
35	0\\
};
\addlegendentry{Simulation: RLS}

\addplot [only marks, line width =1pt,mark size=1.2pt,mark=square, mark options={solid, black}]
  table[row sep=crcr]{%
0	0.34829\\
1	0.32627\\
2	0.303635\\
3	0.27983\\
4	0.25391\\
5	0.228285\\
6	0.20187\\
7	0.176435\\
8	0.153285\\
9	0.130705\\
10	0.10746\\
11	0.08269\\
12	0.062785\\
13	0.043285\\
14	0.028045\\
15	0.016865\\
16	0.00917499999999997\\
17	0.00385499999999997\\
18	0.0014\\
19	0.000315\\
20	0.0001\\
21	5e-06\\
22	0\\
23	0\\
24	0\\
25	0\\
26	0\\
27	0\\
28	0\\
29	0\\
30	0\\
31	0\\
32	0\\
33	0\\
34	0\\
35	0\\
};
\addlegendentry{Simulation: Box-RLS}

\addplot [color=red, line width=1.0pt, forget plot]
  table[row sep=crcr]{%
0	0.711253484048703\\
1	0.699586152608233\\
2	0.684194241923735\\
3	0.665403607486198\\
4	0.64380420939113\\
5	0.619989082397616\\
6	0.594431502722788\\
7	0.567456173242119\\
8	0.539253203928065\\
9	0.509905654983553\\
10	0.479418518895088\\
11	0.447745860428822\\
12	0.414816677914106\\
13	0.380561547333089\\
14	0.344942741954312\\
15	0.307991150743048\\
16	0.269854274186925\\
17	0.230859692545542\\
18	0.191593448717738\\
19	0.152975987179971\\
20	0.116288181974942\\
21	0.0830759900420946\\
22	0.0548855129461526\\
23	0.0328681141245825\\
24	0.0173970488830062\\
25	0.00788370017741921\\
26	0.00293811162958705\\
27	0.000855912780302872\\
28	0.000182791493980562\\
29	2.63918324912536e-05\\
30	2.32571139485218e-06\\
31	1.09939941030772e-07\\
32	2.36883399824932e-09\\
33	1.89440157570858e-11\\
34	4.33970489021406e-14\\
35	2.05429471551232e-17\\
};
\addplot [only marks, line width =1pt,mark size=1.2pt,mark=square, mark options={solid, black}, forget plot]
  table[row sep=crcr]{%
0	0.713095\\
1	0.700564999999999\\
2	0.685090000000001\\
3	0.663655\\
4	0.64351\\
5	0.61872\\
6	0.59371\\
7	0.566244999999999\\
8	0.539895\\
9	0.50927\\
10	0.478585\\
11	0.44592\\
12	0.414390000000001\\
13	0.38022\\
14	0.34742\\
15	0.306595\\
16	0.27048\\
17	0.233195\\
18	0.193285\\
19	0.15074\\
20	0.11712\\
21	0.08325\\
22	0.05689\\
23	0.033485\\
24	0.018745\\
25	0.00813499999999997\\
26	0.00341499999999998\\
27	0.00107\\
28	0.00026\\
29	6.5e-05\\
30	5e-06\\
31	0\\
32	0\\
33	0\\
34	0\\
35	0\\
};
\addplot [color=blue, dashed, line width=1.0pt, forget plot]
  table[row sep=crcr]{%
0	0.711253924048519\\
1	0.699591479103027\\
2	0.684226136323834\\
3	0.665522151352639\\
4	0.644122242583016\\
5	0.620670282741894\\
6	0.59567845637359\\
7	0.569500002679421\\
8	0.542351369380446\\
9	0.51434830732354\\
10	0.485540284863718\\
11	0.455938763226462\\
12	0.425539662428557\\
13	0.394341881735279\\
14	0.362363919102911\\
15	0.329660332960631\\
16	0.296339302277958\\
17	0.262581873203749\\
18	0.228662489830488\\
19	0.194968876508967\\
20	0.162017069633976\\
21	0.130454388200832\\
22	0.101039936472539\\
23	0.0745904697728985\\
24	0.0518822443133251\\
25	0.0335109857886377\\
26	0.0197348904292777\\
27	0.0103547023768895\\
28	0.00470180553535449\\
29	0.0017811494572203\\
30	0.000537518182497747\\
31	0.000121922198230349\\
32	1.93174416766044e-05\\
33	1.94945927541291e-06\\
34	1.11556932947844e-07\\
35	3.12695901256859e-09\\
};
\addplot [only marks, line width =1pt, mark size=1.5pt, mark=o, mark options={solid, OliveGreen}, forget plot]
  table[row sep=crcr]{%
0	0.710835\\
1	0.699545\\
2	0.682605\\
3	0.6664\\
4	0.64303\\
5	0.62046\\
6	0.59462\\
7	0.570305\\
8	0.53977\\
9	0.51497\\
10	0.48589\\
11	0.45662\\
12	0.42586\\
13	0.39182\\
14	0.363355\\
15	0.32811\\
16	0.296835\\
17	0.26537\\
18	0.229275\\
19	0.195645\\
20	0.16357\\
21	0.13141\\
22	0.10251\\
23	0.075945\\
24	0.05345\\
25	0.034135\\
26	0.021585\\
27	0.011185\\
28	0.00540499999999997\\
29	0.00211499999999999\\
30	0.000715\\
31	0.000175\\
32	4e-05\\
33	1e-05\\
34	0\\
35	0\\
};
\addplot [color=red, line width=1.0pt, forget plot]
  table[row sep=crcr]{%
0	0.856560776522316\\
1	0.850133914691904\\
2	0.842154681233208\\
3	0.832552638443874\\
4	0.821367596201006\\
5	0.808719061693533\\
6	0.794752906007021\\
7	0.77959728627331\\
8	0.763338251523865\\
9	0.746011163400331\\
10	0.727601046573941\\
11	0.708046785737042\\
12	0.687246486827392\\
13	0.665063042317415\\
14	0.641329954218627\\
15	0.615858022031592\\
16	0.588443805120912\\
17	0.55888093272155\\
18	0.526975388185607\\
19	0.492565801741042\\
20	0.455549495149654\\
21	0.415914645772482\\
22	0.373779166115683\\
23	0.329438883102927\\
24	0.283430315129813\\
25	0.236609119050697\\
26	0.190222537502513\\
27	0.145915433224433\\
28	0.105590134886961\\
29	0.0710877716827082\\
30	0.0437630656502699\\
31	0.0241037388312149\\
32	0.0115533331425682\\
33	0.00465331448729876\\
34	0.00150670595908694\\
35	0.000370887903393403\\
};
\addplot [only marks, line width =1pt,mark size=1.2pt,mark=square, mark options={solid, black}, forget plot]
  table[row sep=crcr]{%
0	0.856510000000002\\
1	0.849320000000001\\
2	0.844285\\
3	0.832799999999999\\
4	0.820784999999999\\
5	0.808029999999999\\
6	0.794555000000001\\
7	0.77926\\
8	0.763515\\
9	0.746525\\
10	0.727165000000001\\
11	0.7068\\
12	0.68862\\
13	0.66251\\
14	0.641995\\
15	0.617685\\
16	0.58951\\
17	0.557095\\
18	0.52677\\
19	0.491905\\
20	0.454865\\
21	0.41631\\
22	0.3762\\
23	0.329715\\
24	0.284344999999999\\
25	0.2367\\
26	0.190945\\
27	0.146995\\
28	0.10251\\
29	0.072645\\
30	0.045095\\
31	0.02587\\
32	0.012415\\
33	0.00526499999999997\\
34	0.00207999999999999\\
35	0.00046\\
};
\addplot [color=blue, dashed, line width=1.0pt, forget plot]
  table[row sep=crcr]{%
0	0.856560809832214\\
1	0.850134419878044\\
2	0.842158542692523\\
3	0.832570741289359\\
4	0.821427119263506\\
5	0.808870186196997\\
6	0.795070694238331\\
7	0.780179755462676\\
8	0.764304000692831\\
9	0.747499246190571\\
10	0.729774049819035\\
11	0.711096728316407\\
12	0.691402468491484\\
13	0.670599295889109\\
14	0.648572814872365\\
15	0.625190176892908\\
16	0.600303990859824\\
17	0.573757053962806\\
18	0.545388946448561\\
19	0.515045726852003\\
20	0.48259416031167\\
21	0.447942031240205\\
22	0.411065980251948\\
23	0.372047720791533\\
24	0.331118094210097\\
25	0.288705812888075\\
26	0.245483580282788\\
27	0.202398568783768\\
28	0.160667859443718\\
29	0.1217149524468\\
30	0.0870258298879691\\
31	0.0579193827920901\\
32	0.0352629710293805\\
33	0.019215009279918\\
34	0.0091168365245574\\
35	0.00363820395796678\\
};
\addplot [only marks, line width =1pt, mark size=1.5pt, mark=o, mark options={solid, OliveGreen}, forget plot]
  table[row sep=crcr]{%
0	0.857660000000002\\
1	0.850910000000001\\
2	0.841335\\
3	0.832469999999999\\
4	0.823169999999998\\
5	0.80854\\
6	0.79551\\
7	0.779579999999999\\
8	0.76415\\
9	0.75004\\
10	0.72983\\
11	0.71202\\
12	0.689675\\
13	0.66997\\
14	0.649554999999999\\
15	0.62468\\
16	0.600745\\
17	0.57293\\
18	0.546655\\
19	0.51363\\
20	0.480245\\
21	0.44767\\
22	0.41179\\
23	0.370195\\
24	0.33083\\
25	0.28816\\
26	0.2458\\
27	0.20456\\
28	0.162835\\
29	0.12168\\
30	0.0883350000000001\\
31	0.060655\\
32	0.0356\\
33	0.019625\\
34	0.010535\\
35	0.00469499999999997\\
};

\draw [black, line width=0.60pt] (axis cs: 19.20,0.001) ellipse [x radius=10, y radius=2];
\node at (axis cs: 15,0.0001) (nodeA) {\scriptsize {BPSK}};
  \draw[black,thin,->] (axis cs: 18,0.001) - - (nodeA) ;
  
  \draw [black, line width=0.60pt] (axis cs: 28.20,0.001) ellipse [x radius=10, y radius=2];
\node at (axis cs: 26,0.0001) (nodeA) {\scriptsize {4-PAM}};
  \draw[black,thin,->] (axis cs: 27.20,0.001) - - (nodeA) ;
  
  \draw [black, line width=0.60pt] (axis cs: 31.20,0.05) ellipse [x radius=8, y radius=1.2];
\node at (axis cs: 30,0.01) (nodeA) {\scriptsize {8-PAM}};
  \draw[black,thin,->] (axis cs: 30.6,0.08) - - (nodeA) ;
  \end{semilogyaxis}
\end{tikzpicture}%
\caption{\scriptsize {Symbol error probability of both RLS and Box-RLS decoders.}}%
\label{fig221}
\end{center}
\end{figure}
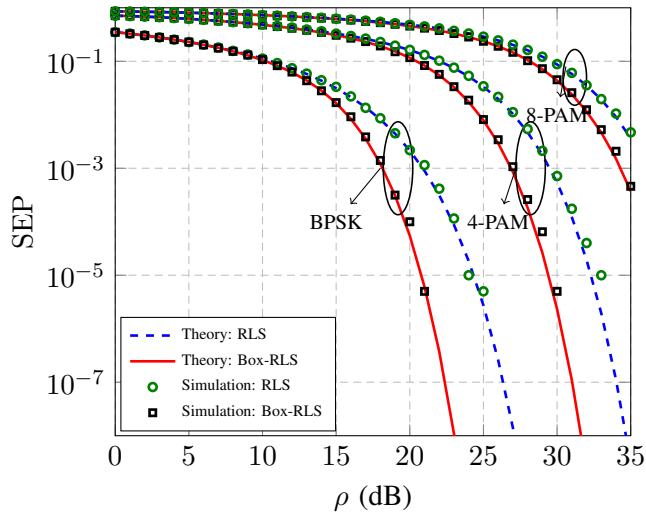
\begin{corollary}{(Optimal regularization coefficient for RLS in  MSE and SEP senses):}\label{C:opt_Lmda} \normalfont
Let $\lambda_\star$ denote the optimal regularization coefficient that minimizes the limit in (\ref{RLS-MSE}) or in \eqref{SER-RLS1}. Then,
\begin{equation}\label{opt_Lmda}
\lambda_\star = \frac{1}{\rho_ {d}}+ \sigma_{ {\Delta}}^2.
\end{equation}
\end{corollary}
\begin{proof}
Note that in both \eqref{RLS-MSE} and \eqref{SER-RLS1}, the regularization coefficient $\lambda$ appears through $\theta_\star$ only. Then, $\lambda_\star = {\argmin}_{\lambda\geq 0} \theta_\star$. Taking the derivative of $\theta_\star$ with respect to $\lambda$, setting it to zero and solving completes the proof of the corollary.
\end{proof}
\begin{remark}
	It is worth mentioning that the optimal regularization coefficient in \eqref{opt_Lmda} minimizes both the MSE and SEP.
Moreover, it can be written in terms of the so called effective SNR of the system \cite{hassibi2003much} as $\lambda_\star = \frac{\sigma_{ {\hat{H}}}^2}{\rho_{\text{eff}}}$, where
	\begin{equation}\rho_{\text{eff}} := \frac{\rho_ {d} \sigma_{ {\hat{H}}}^2}{1 + \rho_ {d} \sigma_{ {\Delta}}^2}. \label{rho_eff} \end{equation} 
\end{remark}
\begin{remark}
{In Appendix \ref{LMMSE_Append}, we show that the RLS detector with optimal regularization coefficient is equivalent to the LMMSE detector. The later is known by definition to minimize the MSE, but it turns out according to Corollary \ref{C:opt_Lmda} that it also minimizes the asymptotic SEP among all other choices of $\lambda$. In the perfect CSI case, $\sigma_{ {\Delta}}^2 = 0$, hence the optimal regularization coefficient becomes ${\lambda}_\star = \frac{1}{\rho_ {d}}$ which is clearly equivalent to the LMMSE decoder. This shows that in both perfect and imperfect settings, the RLS with optimal regularization coefficient turns out to be the LMMSE detector. Such a finding  is appealing due to the fundamental importance of the LMMSE decoder in many applications.}
\end{remark}
\subsection{MSE and SEP Analysis for Box-RLS}
In this subsection, we study the asymptotic performance of the Box-RLS decoder in terms of the MSE and SEP. We first present the MSE results in the following theorem.
\begin{theorem}(MSE of Box-RLS):\label{Box-RLS MSE} \normalfont
Fix $\lambda>0$, $\delta>0$, and let $\widehat{\xv}$ be a minimizer of the Box-RLS problem in (\ref{eq:Box-RLS matrix}). 
Let $\beta_\star$ and $\theta_\star$ be the unique solutions in $\beta$ and $\theta$ to the following max-min problem:
\begin{align} \label{box-minimax}
		\sup_{\beta >0}\min_{ \theta \geq 0} & \ D(\theta,\beta):= \frac{\beta\delta \theta}{2} +\frac{\beta}{2\theta}(1+\rho_{ {d}})- \frac{\beta^2}{4} 
+  \frac{1}{M} \sum_{i =\pm1,\pm3,\cdots,\pm(M-1)} \Psi (\theta,\beta;i),
\end{align}
with
\begin{align}
	\Psi (\theta,\beta; i) :=& t \big(c_i Q(-\ell_i) + d_i Q(\mu_i) \big) - \beta \sqrt{\rho_{ {d}}\sigma_{\hat{ {H}}}^2} t \big( p(\ell_i) + p(\mu_i) \big) \nonumber \\
	-&\frac{\beta^2}{{2 \rho_{ {d}}\sigma_{\hat{ {H}}}^2\frac{\beta}{\theta}} + 4 \lambda\rho_{ {d}}} \int_{\ell_{i}}^{\mu_{i}}  \bigg( \frac{\rho_{ {d}}\sigma_{\hat{ {H}}}^2 i}{\sqrt{\mathcal{E}\theta} } +  \sqrt{\rho_{ {d}}\sigma_{\hat{ {H}}}^2} h \bigg)^2 p(h) {\rm{d}}h,
\end{align}
	where $\ell_i =-t(\frac{\sqrt{\rho_{ {d}}\sigma_{\hat{ {H}}}^2}}{\theta}+\frac{2\lambda \rho_d}{\sqrt{\rho_{ {d}}\sigma_{\hat{ {H}}}^2} \beta})-\frac{\sqrt{\rho_{ {d}}\sigma_{\hat{ {H}}}^2}i}{\theta \sqrt{\mathcal{E}}}$, $\mu_i =t(\frac{\sqrt{\rho_{ {d}}\sigma_{\hat{ {H}}}^2}}{\theta}+\frac{2\lambda\rho_d}{\sqrt{\rho_{ {d}}\sigma_{\hat{ {H}}}^2} \beta})-\frac{\sqrt{\rho_{ {d}}\sigma_{\hat{ {H}}}^2} i}{\theta \sqrt{\mathcal{E}}}$, $c_i = \frac{-\beta \sqrt{\rho_{ {d}}\sigma_{\hat{ {H}}}^2}}{2} \ell_i + \frac{\beta {\rho_{ {d}}\sigma_{\hat{ {H}}}^2} i}{2 \theta\sqrt{\mathcal{E}}}$, and $d_i = \frac{\beta \sqrt{\rho_{ {d}}\sigma_{\hat{ {H}}}^2} }{2} \mu_i - \frac{\beta {\rho_{ {d}}\sigma_{\hat{ {H}}}^2} i}{2 \theta\sqrt{\mathcal{E}} }$.
Then, under Assumption 1 and Assumption 2, it holds:
\begin{equation}
\underset{K\to\infty}{{\rm{plim}}} {\rm{MSE}} = \frac{1}{\rho_ {d} \sigma_{ {\hat{H}}}^2} \biggl(\delta \theta_{\star}^2  - \rho_ {d}   \sigma_{ {\Delta}}^2 -1 \biggr),
\end{equation}
\end{theorem}
\begin{proof}
The proof of this theorem is deferred to Appendix \ref{sec:box}.
\end{proof}
\begin{remark}
	It is worth mentioning that contrary to previous works based on  the framework of the CGMT, the optimization problem in \eqref{box-minimax} which  resulted from the asymptotic analysis  is not convex-concave in the variables $\theta$ and $\beta$. Indeed, it is concave in $\beta$ but not convex in $\theta$. This poses a major challenge to prove the uniqueness of the solutions in $\theta$ of \eqref{box-minimax}, which is a crucial step that is required to ensure  convergence results. The reader can refer to Appendix \ref{sec:box} for more details of the technical arguments developed to show the uniqueness of the solutions of \eqref{box-minimax}.  
\end{remark}
\begin{remark}
If the optimal values $\theta_\star$ and $\beta_\star$ are strictly positive, then they satisfy the following first-order stationarity conditions: 
\begin{align*}
  & \frac{\partial D(\theta,\beta)}{\partial \theta} =0,\nonumber \\
  &\frac{\partial D(\theta,\beta)}{\partial \beta} =0.
\end{align*}
 which can be exploited in practice to facilitate their numerical evaluation. 
\end{remark}

The following theorem provides the asymptotic expression of the SEP for the Box-RLS decoder.
\begin{theorem}(SEP for Box-RLS):\label{Box_RLS SER} \normalfont
	Under the same settings of Theorem \ref{Box-RLS MSE}, assuming that the Box-RLS decoder uses a normalization constant given by  $B=\frac{\sigma_{\hat{ {H}}}^2\frac{\beta_\star}{\theta_\star}}{\sigma_{\hat{ {H}}}^2\frac{\beta_\star}{\theta_\star}+2\lambda} $ where $\beta_\star$ is the solution to \eqref{box-minimax} in $\beta$, and that $t\notin\left\{\frac{Bi}{\sqrt{\mathcal{E}}},i= 1,3\cdots, M-1\right\}$, 
	it holds that:
$$
	\underset{K\to\infty}{{\rm{plim}}}  {\rm SEP}=\widetilde{{\rm SEP}}_{\text{\tiny Box-RLS}},
$$
	where $\widetilde{\rm {SEP}}_{\text{\tiny Box-RLS}}$ is given in \eqref{eq:SEP_BoXc}.
	\begin{figure*}
	\begin{equation}
		\label{eq:SEP_BoXc}
		\begin{aligned}	\widetilde{\rm SEP}_{\text{\tiny Box-RLS}}&=\frac{4}{M}\sum_{i=1,3,\dots,M-3}{\bf 1}_{\{\frac{t}{B}\geq\frac{i+1}{\sqrt{\mathcal{E}}}\}}Q\left(\frac{\sqrt{\rho_{ {d}}\sigma_{\hat{ {H}}}^2}}{\sqrt{\mathcal{E}}\theta_\star}\right)+\frac{2}{M}\sum_{i=1,3,\cdots,M-3}\left\{{\bf 1}_{\{\frac{i-1}{\sqrt{\mathcal{E}}}\leq \frac{t}{B}\leq \frac{i+1}{\sqrt{\mathcal{E}}}\}}Q\left(\frac{\sqrt{\rho_{ {d}}\sigma_{\hat{ {H}}}^2}}{\sqrt{\mathcal{E}}\theta_\star}\right)+{\bf 1}_{\{\frac{t}{B}\leq \frac{i-1}{\sqrt{\mathcal{E}}}\}}\right\}\\
			&+\frac{2}{M}{\bf 1}_{\{\frac{t}{B}\geq \frac{M-2}{\sqrt{\mathcal{E}}}\}}Q\left(\frac{\sqrt{\rho_{ {d}}\sigma_{\hat{ {H}}}^2}}{\sqrt{\mathcal{E}} \theta_\star}\right)+\frac{2}{M}{\bf 1}_{\{\frac{t}{B}\leq \frac{M-2}{\sqrt{\mathcal{E}}}\}}.
	\end{aligned}
	\end{equation}
	\hrule
\end{figure*}

	If $t\geq \frac{M-1}{\sqrt{\mathcal{E}}}$, then $\widetilde{\rm SEP}_{\text{\tiny Box-RLS}}$ is simplified to:
	\begin{equation}
		\widetilde{{\rm SEP}}_{\text{\tiny Box-RLS}}=2\left(1-\frac{1}{M}\right)Q \biggl( \frac{\sqrt{\rho_ {d} \sigma_{ {\hat{H}}}^2 }}{\sqrt{\mathcal{E}}\theta_\star} \biggr).
	\end{equation}
\end{theorem}
\begin{proof}
The proof of Theorem~\ref{Box_RLS SER} is also based on the CGMT framework and is given in Appendix \ref{sec:box}.
\end{proof}
\begin{remark}
	Figure \ref{xz2} and Figure \ref{fig221} reveal a perfect match between the analytical expressions of MSE and SEP given by Theorem~\ref{Box-RLS MSE} and Theorem~\ref{Box_RLS SER} and the numerical simulations. It also shows that the Box-RLS outperforms the ordinary RLS. Figure \ref{xz2} also suggests that as $M$ increases, the performance of the Box-RLS approaches that of the un-boxed RLS. This is because in this figure $t=\frac{M-1}{\sqrt{\mathcal{E}}}$ and as such as $M \to \infty$, the box-constraint $[-t,t]$ tends to $(-\infty,\infty)$ which is the whole real line $\mathbb{R}$, thereby reducing  the Box-RLS to the RLS.
\end{remark}
\begin{remark}
It should be noted that when $t\geq\frac{M-1}{\sqrt{\mathcal{E}}}$, the MSE and SEP expressions take the same form as in the RLS case, with the single difference that $\theta_\star$ and $ \beta_\star$ do not admit a closed-form expression.  
Moreover, as for the RLS, the optimal regularization coefficient for Box-RLS is  given by $\lambda_\star = {\argmin}_{\lambda\geq 0} \theta_\star$, since $\lambda$ appears in the expressions for the MSE and SEP only through $\theta_\star$. However, in contrast to the RLS,  the optimal regularization coefficient cannot be obtained in closed-form, but could be retrieved by invoking any bisection algorithm. Moreover, as opposed to the RLS, its value depends on $M$. 
\end{remark}
\begin{remark}
Figure \ref{fig4} plots the optimal regularization coefficient computed using a bisection algorithm as a function of $\rho_d$ for RLS and Box-RLS for different values of $M$. As a first observation, we note that the optimal regularization coefficient for Box-RLS becomes zero starting from moderate values of $\rho_d$. Moreover, the Box-RLS needs less regularization, due to its achieved improvement over the RLS. On the other hand, in  low SNR regions corresponding to low $\rho_d$ values, the optimal regularization coefficient for both RLS and Box-RLS are higher than  $\frac{1}{\rho_{ {d}}}$ which coincides with the optimal regularization coefficient in the perfect CSI case. 
This can be explained by the fact, under imperfect CSI cases, more regularization is needed in low SNR regions, because of the degradation caused by channel estimation errors. 
  
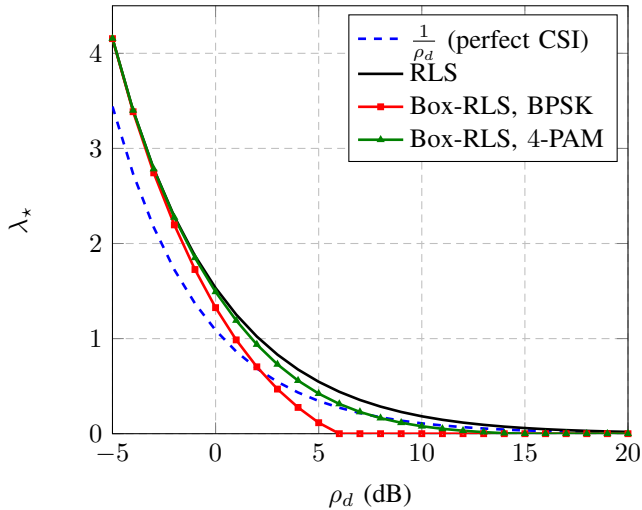
\begin{figure}
  \centering
%
%
\definecolor{OliveGreen}{rgb}{0,0.5,0}%
\begin{tikzpicture}[scale=1,font=\small]
    \renewcommand{\axisdefaulttryminticks}{4}
    \tikzstyle{every major grid}+=[style=densely dashed]
    \tikzstyle{every axis y label}+=[yshift=-10pt]
    \tikzstyle{every axis x label}+=[yshift=5pt]
    \tikzstyle{every axis legend}+=[cells={anchor=west},fill=white,
        at={(0.98,0.98)}, anchor=north east, font=\small ]
       \begin{axis}				[
      xmin=-5,
      ymin=0,
      xmax=20,
      ymax=4.5,
      grid=major,
      scaled ticks=true,
   			xlabel={$\rho_d \text{  (dB)}$},
   			ylabel={$\lambda_\star$}			
   			]
\addplot
 [color=blue,dashed ,line width=1pt]
   table[row sep=crcr]{%
-5	3.4405580942632\\
-4	2.73293243748242\\
-3	2.17084539868614\\
-2	1.72436379339769\\
-1	1.36971084803205\\
0	1.088\\
1	0.864229119380018\\
2	0.68648159079445\\
3	0.545291710186472\\
4	0.433140601562205\\
5	0.34405580942632\\
6	0.273293243748242\\
7	0.217084539868614\\
8	0.172436379339769\\
9	0.136971084803205\\
10	0.1088\\
11	0.0864229119380018\\
12	0.068648159079445\\
13	0.0545291710186472\\
14	0.0433140601562205\\
15	0.034405580942632\\
16	0.0273293243748242\\
17	0.0217084539868614\\
18	0.0172436379339769\\
19	0.0136971084803205\\
20	0.01088\\
21	0.00864229119380018\\
22	0.0068648159079445\\
23	0.00545291710186473\\
24	0.00433140601562205\\
25	0.0034405580942632\\
26	0.00273293243748242\\
27	0.00217084539868614\\
28	0.00172436379339769\\
29	0.00136971084803205\\
30	0.001088\\
31	0.000864229119380018\\
32	0.00068648159079445\\
33	0.000545291710186472\\
34	0.000433140601562205\\
35	0.00034405580942632\\
};
\addlegendentry{$\frac{1}{\rho_d}$ (perfect CSI)}

\addplot [color=black,line width=1pt]
  table[row sep=crcr]{%
-5	4.15725770016418\\
-4	3.40065233738956\\
-3	2.78566852277249\\
-2	2.28343007380496\\
-1	1.87148958002469\\
0	1.53244444444444\\
1	1.25278132016269\\
2	1.02192638242239\\
3	0.831490252525747\\
4	0.674694778175488\\
5	0.545959882944985\\
6	0.440619744849209\\
7	0.354733812549138\\
8	0.284960710962633\\
9	0.228469902778846\\
10	0.182874074074074\\
11	0.146172319407085\\
12	0.116699289543701\\
13	0.0930785145994368\\
14	0.0741796100612779\\
15	0.0590795935920746\\
16	0.0470285579271314\\
17	0.0374197669807297\\
18	0.0297640355303269\\
19	0.0236680899985631\\
20	0.0188165079365079\\
21	0.0149567907890611\\
22	0.011887123767544\\
23	0.00944640309729428\\
24	0.00750615226607409\\
25	0.00596399637240288\\
26	0.00473841155406479\\
27	0.00376451142524889\\
28	0.002990672775778\\
29	0.00237583786637068\\
30	0.00188736051159073\\
31	0.00149928815090388\\
32	0.00119099280633315\\
33	0.000946080800787821\\
34	0.000751524937135416\\
35	0.000596974038426557\\
};
\addlegendentry{RLS}

\addplot
[color=red,mark size=0.8pt,mark=square,line width=1.0pt]
  table[row sep=crcr]{%
-5	4.15424581987431\\
-4	3.38585860392213\\
-3	2.74304928828863\\
-2	2.19575501246768\\
-1	1.72650523915043\\
0	1.32553878206971\\
1	0.986474945198123\\
2	0.703273145381162\\
3	0.469171413213606\\
4	0.276273684491576\\
5	0.115532139016519\\
6	5.96086098654914e-05\\
7	5.96086098654914e-05\\
8	5.96086098654914e-05\\
9	5.96086098654914e-05\\
10	5.96086098654914e-05\\
11	5.96086098654914e-05\\
12	5.96086098654914e-05\\
13	5.96086098654914e-05\\
14	5.96086098654914e-05\\
15	5.96086098654914e-05\\
16	6.07870631868558e-05\\
17	4.89230689787146e-05\\
18	4.34667185797029e-05\\
19	5.6679915609408e-05\\
20	6.54912056601158e-05\\
21	3.7613662484417e-05\\
22	6.36483166782655e-05\\
23	5.71651555860418e-05\\
24	4.34493399547347e-05\\
25	3.77139265870937e-05\\
26	4.86540643110321e-05\\
27	3.91195412456833e-05\\
28	5.77358783821564e-05\\
29	6.60656563422069e-05\\
30	4.58610944855201e-05\\
31	4.84195352719802e-05\\
32	4.60072923578154e-05\\
33	4.1330272796911e-05\\
34	5.06278714063414e-05\\
35	3.58952074289834e-05\\
};
\addlegendentry{Box-RLS, BPSK}

\addplot
[color=OliveGreen,mark size=0.8pt,mark=triangle, line width=1.0pt]
  table[row sep=crcr]{%
-5	4.15724945478929\\
-4	3.40034706019549\\
-3	2.78335887721174\\
-2	2.27465840830513\\
-1	1.84984727229589\\
0	1.49208057091081\\
1	1.19055537203053\\
2	0.938031413445578\\
3	0.728963441437463\\
4	0.558331352666662\\
5	0.421294840179442\\
6	0.312976823947022\\
7	0.228664592895155\\
8	0.164015515013328\\
9	0.115085806282384\\
10	0.0784655350399345\\
11	0.0513249841675759\\
12	0.0312743254020996\\
13	0.0163976939512193\\
14	0.00514350879978435\\
15	5.96086098654914e-05\\
16	5.96086098654914e-05\\
17	5.96086098654914e-05\\
18	5.96086098654914e-05\\
19	5.96086098654914e-05\\
20	5.96086098654914e-05\\
21	5.96086098654914e-05\\
22	5.96086098654914e-05\\
23	5.96086098654914e-05\\
24	5.96086098654914e-05\\
25	5.96086098654914e-05\\
26	5.96086098654914e-05\\
27	5.96086098654914e-05\\
28	5.96086098654914e-05\\
29	5.96086098654914e-05\\
30	5.96086098654914e-05\\
31	5.96086098654914e-05\\
32	5.96086098654914e-05\\
33	5.96086098654914e-05\\
34	5.96086098654914e-05\\
35	5.96086098654914e-05\\
};
\addlegendentry{Box-RLS, $4$-PAM}


\end{axis}
\end{tikzpicture} \vskip-4mm
\centering
  \caption{\scriptsize {Optimal regularization coefficient $\lambda_\star$ as a function of the data power. We used $\delta=1.2, K=400, \alpha =0.5, T_p =400$, and $T =1000$.}}
  \label{fig4}
\end{figure}
\end{remark}
\begin{remark}
Similar to the regularization coefficient, we can set the threshold $t$ to the optimal value that minimizes the MSE and SEP expressions, that is  $t_\star = \argmin_{t>0} \theta_\star$. Figure \ref{fig:threshold} shows the optimal box-threshold as a function of $\rho_d$ when the regularization coefficient is already optimized as well.  As can be seen, for practical SNR regions, the optimal threshold coincides with $\frac{M-1}{\sqrt{\mathcal{E}}}$, which  is the maximum value of $\xv_0$. For this reason, we will use in the subsequent simulations this value for the threshold $t$.     
 
	\begin{figure}[htbp]
\centering
     \begin{subfigure}
         \centering
%
%
\definecolor{mycolor1}{rgb}{0.00000,0.44700,0.74100}%
\begin{tikzpicture}[scale=1,font=\small]
\begin{axis}[%
width=2.0in,
height=0.9in,
scale only axis,
xmin=-5,
xmax=30,
ymin=0.6,
ymax=1,
xticklabels={,,}
ylabel style={font=\color{white!15!black}},
ylabel={$\sqrt{\mathcal{E}} t_\star$},
      grid=major,
      scaled ticks=true,
legend style={at={(0.97,0.03)}, anchor=south east, legend cell align=left, align=left, draw=white!15!black}
]
\addplot [color=red, line width=1pt]
  table[row sep=crcr]{%
-5	 0.655784996233717\\
-4	 0.699667577315661\\
-3	 0.740154774059821\\
-2	 0.77672939261185\\
-1	 0.809236937867073\\
0	0.837816501281342\\
1	0.86276875967115\\
2	0.884597882991335\\
3	0.903777640289111\\
4	0.920801752341452\\
5	0.936147343327707\\
6	0.957623575701975\\
7	0.981872525807661\\
8	0.993699834616644\\
9	0.998291422969977\\
10	0.999648997142122\\
11	0.999932980465548\\
12	0.999968467903892\\
13	0.999966927345475\\
14	0.999962620014242\\
15	0.999959829600373\\
16	0.999974429153196\\
17	0.999970813675536\\
18	0.999971345555117\\
19	0.999969034904514\\
20	0.999966335365578\\
21	0.999966366490225\\
22	0.999961215778105\\
23	0.999980132211762\\
24	0.999959476958079\\
25	0.99998602507152\\
26	0.999966251446396\\
27	0.999965857886831\\
28	0.999965140190463\\
29	0.999963697107555\\
30	0.99996238061099\\
};
\addlegendentry{BPSK}

\end{axis}
\end{tikzpicture} \vskip-4mm  
     \end{subfigure}
         \begin{subfigure}
         \centering
         \definecolor{OliveGreen}{rgb}{0,0.5,0}%
\begin{tikzpicture}[scale=1,font=\small]
\begin{axis}[%
width=2.0in,
height=0.9in,
scale only axis,
xmin=-5,
xmax=30,
ymin=1.5,
ymax=3,
xticklabels={,,}
ylabel style={font=\color{white!15!black}},
ylabel={$\sqrt{\mathcal{E}} t_\star$},
      grid=major,
      scaled ticks=true,
legend style={at={(0.97,0.03)}, anchor=south east, legend cell align=left, align=left, draw=white!15!black}
]
\addplot [color=OliveGreen, dashed, line width=1.0pt]
  table[row sep=crcr]{%
-5	  1.64513355963242\\
-4	  1.76498782268879\\
-3	   1.87830095435446\\
-2	   1.98322236085632\\
-1	     2.07892037574507\\
0	2.1654684168812\\
1	2.24350166033929\\
2	2.31417389444816\\
3	2.37883334504961\\
4	2.43893183943836\\
5	2.49587539453415\\
6	2.55093974932897\\
7	2.60507842042143\\
8	2.65900273774476\\
9	2.71276281688572\\
10	2.7659482244386\\
11	2.81748544590484\\
12	2.86561223457354\\
13	2.90803476661852\\
14	2.9425897556917\\
15	2.97023904737386\\
16	2.98759421461489\\
17	2.99575460575744\\
18	2.9988871901456\\
19	2.99975681125119\\
20	2.999942750708\\
21	2.99990839054146\\
22	2.99994774759922\\
23	2.99992031027954\\
24	2.99991695195547\\
25	2.99990117314568\\
26	2.9998964040024\\
27	2.99989162467992\\
28	2.99989003443697\\
29	2.9999096218919\\
30	2.99988588099282\\
};
\addlegendentry{4-PAM}

\end{axis}
\end{tikzpicture} \vskip-4mm  
     \end{subfigure}
               \begin{subfigure}
         \centering
         \definecolor{mycolor1}{rgb}{0.00000,0.44700,0.74100}%
\begin{tikzpicture}[scale=1,font=\small]
\begin{axis}[%
width=2.0in,
height=0.9in,
scale only axis,
xmin=-5,
xmax=32,
xlabel={$\rho_d$ (dB)},
ymin=3,
ymax=7,
ylabel style={font=\color{white!15!black}},
ylabel={$\sqrt{\mathcal{E}} t_\star$},
      grid=major,
      scaled ticks=true,
legend style={at={(0.97,0.03)}, anchor=south east, legend cell align=left, align=left, draw=white!15!black}
]
\addplot [color=blue, dashdotted, line width=1.0pt]
  table[row sep=crcr]{%
-5	 3.46540863107506\\
-4	 3.72238861227807\\
-3	 3.96664393459317\\
-2	 4.19383634746515\\
-1	 4.40214085457313\\
0	4.59140957704138\\
1	4.76266621250509\\
2	4.91838562007667\\
3	5.06121521421907\\
4	5.19404394421599\\
5	5.3197186085117\\
6	5.44105465917373\\
7	5.55983612785968\\
8	5.67761494871064\\
9	5.79515361201378\\
10	5.91243468897544\\
11	6.02887932227113\\
12	6.14324430762274\\
13	6.2545681001569\\
14	6.36160039270286\\
15	6.4632701214185\\
16	6.55904836079443\\
17	6.64802060412883\\
18	6.72961749112961\\
19	6.80284091268854\\
20	6.86579249353254\\
21	6.91641399012802\\
22	6.95336473287054\\
23	6.9783073556648\\
24	6.99161927759436\\
25	6.99735856496192\\
26	6.99933052405093\\
27	6.99984933786622\\
28	6.99980978539383\\
29	6.99979112513207\\
30	6.9997674381332\\
31	6.99975617748796\\
32	9.99079031538668\\
};
\addlegendentry{8-PAM}

\end{axis}
\end{tikzpicture}%
    \end{subfigure}
        \caption{\scriptsize{The optimal normalized box-threshold $\sqrt{\mathcal{E}} t_\star$ as a function the data power.  We used $\delta=1.2, K=400, T_p =400, \alpha =0.5$, and $T =1000$.}}
        \label{fig:threshold}
\end{figure}
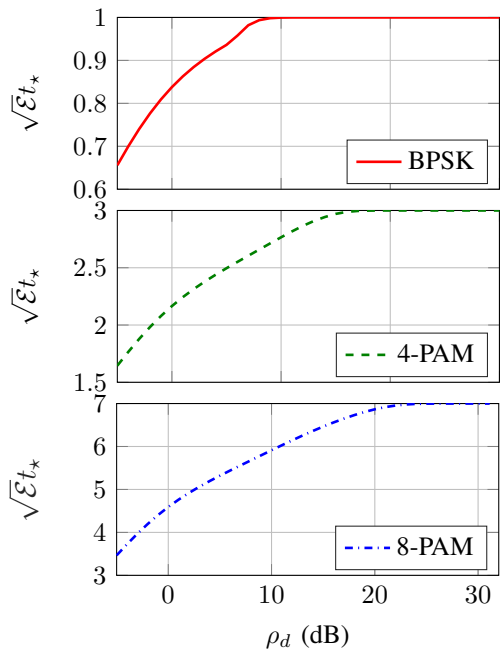
\end{remark}
\section{Optimal Data Power Allocation and Optimal Training Duration Allocation}
In this section, we leverage the asymptotic expressions of the MSE and SEP derived thus far for the RLS and Box-RLS  to determine the optimal power distribution between the  training and data symbols.  Particularly, we show that for all considered decoders, the optimal allocation schemes boils down to maximizing the effective SNR of the system $\rho_{\rm eff}$ defined in \eqref{rho_eff}. Additionally, we derive for each decoder the asymptotic expression of the goodput and derive the optimal  fraction of power allocated to the pilot transmission  as well as the training duration. (i.e.  $(\tau_{ {p}}, \alpha)$ that maximizes the goodput). In this respect, we illustrate that, while the optimal power allocation remains to be the one that maximizes the effective SNR, the optimal number of training symbols coincides with the number of transmitting antennas $K$, which is also the minimum number of training symbols that needs to be employed to satisfy orthogonality between pilot sequences. 

\subsection{Simplifying the Decoders' Expressions}
\subsubsection{LS decoder}
To begin with, we consider first the LS decoder for which $\lambda=0$ and $\delta>1$. Hence, from \eqref{RLS-MSE}, $\Upsilon(0,\delta) =0$, and as such: 
\begin{equation*}
\theta_\star =\sqrt{\frac{\rho_{ {d}} \sigma_{ {\Delta}}^2 +1}{\delta -1}},
\end{equation*}
The MSE limit in \eqref{RLS-MSE} reduces thus to
\begin{align}
\widetilde{\rm{ MSE}}_{\text{\tiny{LS}}}  :=\underset{K\to\infty}{{\rm{plim}}}  {\rm{ MSE}}_{\text{\tiny{LS}}} =& \frac{\delta (\rho_{ {d}} \sigma_{ {\Delta}}^2 +1)}{(\delta -1)\rho_{ {d}} \sigma_{ {\hat{H}}}^2} -\frac{\rho_{ {d}} \sigma_{ {\Delta}}^2 +1}{\rho_{ {d}} \sigma_{ {\hat{H}}}^2} \nonumber \\
=& \frac{\delta}{(\delta-1)\rho_\text{eff}} - \frac{1}{\rho_\text{eff}} \label{LS_MSE} \nonumber \\
=& \frac{1}{(\delta -1) \rho_\text{eff}},
\end{align}
where $\rho_\text{eff}$ is the effective SNR defined in \eqref{rho_eff}. 
The result in \eqref{LS_MSE} recovers the well-known formula of the MSE of LS with the difference being that $\rho_d$ which stands for the SNR in the perfect CSI case is replaced by $\rho_\text{eff}$.  
Similarly, for $\delta>1$, and from \eqref{SER-RLS1}, the SEP of the LS decoder can also be expressed in terms of $\rho_\text{eff}$ as follows
\begin{align}\label{LS_pp}
\underset{K\to\infty}{{\rm{plim}}} {\rm{SEP}}_{\text{\tiny{LS}}}  &= 2 \bigg( 1 - \frac{1}{M} \bigg) Q \bigg(\sqrt{\frac{(\delta -1)}{\mathcal{E} } \rho_{\text{eff}}}\bigg) \nonumber \\
&= 2 \bigg( 1 - \frac{1}{M} \bigg) Q \bigg(\sqrt{\frac{1}{\mathcal{E} \cdot \widetilde{\text{ MSE}}_{\text{\tiny{LS}}}  }}\bigg).
\end{align}
Again, the first equation in \eqref{LS_pp} parallels the well-known result for the LS and BPSK signaling but under perfect CSI \cite{hassibi2005sphere}, (in which case the BER converges in probability to  ${\text{BER}}_{LS} =Q(\sqrt{(\delta - 1)\rho_{ {d}}})$) in that it takes the same form with $\rho_{\text{eff}}$ replacing $\rho_{ {d}}$. Hence, our result generalizes  \cite{hassibi2005sphere} to  encompass $M$-PAM modulation  and imperfect CSI scenarios.
\subsubsection{RLS decoder}
We proceed now with the RLS decoder. 
The MSE expression in \eqref{RLS-MSE} can also be written in terms of $\rho_{\text{eff}}$ as
\begin{equation}
\widetilde{\rm{ MSE}}_{\text{\tiny{RLS}}}  :=\underset{K\to\infty}{{\rm{plim}}} {\rm{ MSE}}_{\text{\tiny{RLS}}}  = \frac{\delta \theta_\star^2}{\rho_{ {d}}\sigma_{ {\hat{H}}}^2}- \frac{1}{\rho_{\text{eff}}},
\end{equation}
from which it follows that $\frac{\rho_{ {d}}\sigma_{ {\hat{H}}}^2}{\theta_\star^2} = \frac{\delta}{\widetilde{\rm{ MSE}}_{\text{\tiny{RLS}}} +\frac{1}{\rho_{\text{eff}}}}$.
This yields the following interesting relationship between the MSE and SEP for the RLS decoder.  
\begin{align}\label{SER_MSE_RLS}
	\widetilde{\rm SEP}_{\text{\tiny RLS}}:=\underset{K\to\infty}{{\rm{plim}}} { \rm{SEP}}_{\text{\tiny{RLS}}}  &= 2 \bigg( 1 - \frac{1}{M} \bigg) Q \bigg(\sqrt{\frac{\delta }{\mathcal{E} \cdot (\widetilde{\rm{ MSE}}_{\text{\tiny{RLS}}} +\frac{1}{\rho_{\text{eff}}}) }}\bigg).
\end{align}
Such an expression holds for any $\lambda>0$, and not necessarily $\lambda_\star$. But when $\lambda=\lambda_\star$, with  $\lambda_\star =\frac{\sigma_{ {\hat{H}}}^2}{\rho_{\text{eff}}}$, we obtain  after some algebraic manipulations the following expression for the MSE:
\begin{align}\label{RLS_pp}
\underset{K\to\infty}{{\rm{plim}}} {\rm{MSE}}_{\text{\tiny{RLS}}} 
 & =\frac{1}{2} \bigg( -\frac{1}{\rho_{\text{eff}}} + (1-\delta) + \sqrt{\frac{1}{\rho_{\text{eff}}^2} + 
\frac{2(1+\delta ) }{\rho_{\text{eff}}} +(1-\delta )^2} \bigg)\nonumber \\
 &= \frac{1}{2} \bigg( -(\delta -1 +\frac{1}{\rho_{\text{eff}}}) + \sqrt{(\delta -1 +\frac{1}{\rho_{\text{eff}}})^2 +  \frac{4}{\rho_{\text{eff}}}} \bigg).
\end{align}
Note that in the perfect CSI case for which the optimal regularization coefficient is $\frac{1}{\rho_{ {d}}}$, the right-hand side of \eqref{RLS_pp} is exactly the minimum mean squared error estimator (MMSE) (see\cite[Theorem 8]{wu2012optimal}), where $\rho_{\text{eff}}$ is replaced by ${\rho_{ {d}}}$.
\subsubsection{Box-RLS decoder}
In a similar way, for the Box-RLS decoder, we have the same asymptotic relationships between MSE and SEP:
\begin{equation}
\widetilde{\rm{ MSE}}_{\text{\tiny{Box-RLS}}}:=\underset{K\to\infty}{{\rm{plim}}} {\rm{MSE}}_{\text{\tiny{Box-RLS}}} = \frac{\delta \theta_\star^2}{\rho_{ {d}} \sigma_{ {\hat{H}}}^2} - \frac{1}{\rho_{\text{eff}}},
\end{equation}
and, for $t \geq \frac{M-1}{\sqrt{\mathcal{E}}}$:
\begin{align}
	\widetilde{\rm SEP}_{\text{\tiny{Box-RLS}}}&:=\underset{K\to\infty}{{\rm{plim}}}  {\rm{SEP}}_{\text{\tiny{Box-RLS}}}\nonumber\\
	&= 2 \bigg( 1 - \frac{1}{M} \bigg) Q \bigg(\sqrt{\frac{\delta }{\mathcal{E} \cdot (\widetilde{\rm{MSE}}_{\text{\tiny{Box-RLS}}}+\frac{1}{\rho_{\text{eff}}}) }}\bigg),
\end{align}
which again reveals that minimizing the MSE is equivalent to minimizing the SEP. 
\subsection{Optimal Power Allocation in MSE and SEP Sense}
For the RLS decoder, we prove in Appendix \ref{Azx} that both $\widetilde{\rm{ MSE}}_{\text{\tiny{RLS}}}$, and $\widetilde{\rm{ SEP}}_{\text{\tiny{RLS}}}$ are monotonically increasing functions in $\frac{1}{\rho_{\text{eff}}}$. Hence, minimizing the MSE or SEP is equivalent to maximizing $\rho_{\text{eff}}$. This can be easily seen to be the case of the LS decoder. However, for the Box-RLS decoder, such a statement could not be checked analytically as $\theta_\star$ does not possess a closed-form expressions. However, based on extensive simulations, we conjecture that both $\widetilde{\rm{MSE}}_{\text{\tiny{Box-RLS}}}$ and $\widetilde{\rm{SEP}}_{\text{\tiny{Box-RLS}}}$ increase with $\frac{1}{\rho_{\text{eff}}}$. 
All these considerations suggest that the optimal power allocation is the one that maximizes $\rho_{\text{eff}}$ over $\alpha$, i.e.,
\begin{equation}
\alpha_\star = \argmax_{0< \alpha<1} \rho_{\text{eff}}.
\end{equation}
Recall that $\rho_{\text{eff}} = \frac{\rho_ {d} \sigma_{ {\hat{H}}}^2}{1+ \rho_ {d} \sigma_{ {\Delta}}^2}$. Substituting the expressions for $\sigma_{ {\hat{H}}}^2$ and $\sigma_{ {\Delta}}^2$ gives
\begin{align}
\rho_\text{eff}= \frac{\tau_ {p} \rho_ {p} \rho_ {d}}{(1+ \rho_ {d}) + \tau_ {p} \rho_ {p}}.
\end{align}
Further, upon using $\rho_ {p} = \frac{(1-\alpha) \rho \tau}{\tau_{ {p}}}$, and $\rho_ {d} = \frac{\alpha \rho \tau}{\tau_{ {d}}}$, the effective SNR becomes  
\begin{align*}
\rho_\text{eff}&= \frac{\rho \tau}{\tau_ {d} -1} \cdot \frac{\alpha (1- \alpha)}{-\alpha + \frac{1 + \rho \tau}{\rho \tau(1- \frac{1}{\tau_ {d}})}}.
\end{align*}
With this expression at hand, we determine in the following Theorem the optimal power allocation that maximizes the effective SNR:
\begin{theorem}(Optimal Power Allocation):\label{Power_Th} \normalfont
The optimal power allocation $\alpha_\star$ that maximizes the effective SNR in a training-based system is given by
\begin{equation}\label{optimal_power}
\alpha_\star =
\begin{cases}
\vartheta - \sqrt{\vartheta(\vartheta -1)},  & \text{if $\tau_ {d} > 1$,} \\
\frac{1}{2}, & \text{if $\tau_ {d} =1$,} \\
\vartheta + \sqrt{\vartheta(\vartheta -1)} & \text{if $\tau_ {d} < 1$,}
\end{cases}
\end{equation}
where $\vartheta = \frac{1 + \rho \tau}{\rho \tau (1 - \frac{1}{\tau_ {d}})}$.
\end{theorem}
\begin{proof}
The proof of this theorem is given in Appendix \ref{power_proof}.
\end{proof}
It is worth mentioning that the use of this power allocation has already been proposed in the early work of  \cite{hassibi2003much} as the one that maximizes a lower bound on the capacity. Interestingly, we retrieve the same power allocation scheme which we prove to be optimum in the MSE/SEP sense for RLS and LS decoders and conjecture that it is also optimum for the Box-RLS decoder.  

\begin{remark}
\begin{itemize}
\item At high SNR ($\rho \gg1$), $\vartheta \approx \frac{\tau_{ {d}}}{\tau_{ {d}} -1}$, then $\alpha_\star \approx \frac{\sqrt{\tau_{ {d}}}}{1+\sqrt{\tau_{ {d}}}}$.
\item At low SNR ($\rho \ll 1$), $\vartheta \approx \frac{\tau_{ {d}}}{\rho \tau(\tau_{ {d}}-1)}$, and $\alpha_\star \approx \frac{1}{2}$. This means that, at low SNR, \textit{half} of the transmit energy should be devoted to training and the other half to data transmission.
\end{itemize}
\end{remark}
\begin{remark}[Numerical Illustration]
The asymptotic predictions of the MSE and the SEP are plotted as functions of the data power ratio $\alpha$ in Figure \ref{fig:MSE_alpha} and Figure \ref{fig:SER_alpha} when $\delta =2, K = 256, T=1000, T_{ {p}} =256, M=2$ and $\rho =15$ dB. As can be seen, the optimal power allocation, $\alpha_\star$, is the same in the MSE and SEP sense for the different decoders considered here, namely LS, RLS and Box-RLS. The same conclusion has been found for other settings confirming the conjecture that for Box-RLS, the optimal power allocation is obtained by maximizing the effective SNR.
\end{remark}
\begin{figure}[hh]
\centering
\definecolor{mycolor1}{rgb}{0.47059,0.67059,0.18824}%
\definecolor{mycolor2}{rgb}{0,0.5,0}%
\definecolor{mycolor3}{rgb}{0.14902,0.14902,0.14902}%
\begin{tikzpicture}
 \renewcommand{\axisdefaulttryminticks}{4}
    \tikzstyle{every major grid}+=[style=densely dashed]
    \tikzstyle{every axis y label}+=[yshift=-10pt]
    \tikzstyle{every axis x label}+=[yshift=5pt]
    \tikzstyle{every axis legend}+=[cells={anchor=west},fill=white,
        at={(0.01,0.01)}, anchor=south west, font=\tiny ]
\begin{axis}[
xmin=0,
xmax=1,
xlabel={$\alpha$},
ymin=-20,
ymax=10,
 grid=major,
 scaled ticks=true,
ylabel={$\rm MSE$ (dB)},
legend pos=north east
]

\addplot [color=blue, dashed, line width=1.0pt]
  table[row sep=crcr]{%
0.001	13.7522642938386\\
0.011	3.35366133479564\\
0.021	0.560976644237925\\
0.031	-1.11460177992811\\
0.041	-2.3127067969147\\
0.051	-3.24417520541753\\
0.061	-4.00509145132674\\
0.071	-4.6474030347897\\
0.081	-5.20239503952094\\
0.091	-5.69037470745019\\
0.101	-6.12527255209067\\
0.111	-6.51705987263272\\
0.121	-6.87311910507284\\
0.131	-7.19906817409022\\
0.141	-7.49928073818944\\
0.151	-7.77722774671097\\
0.161	-8.03570924367427\\
0.171	-8.27701619254655\\
0.181	-8.50304623312322\\
0.191	-8.71538826097845\\
0.201	-8.9153853904172\\
0.211	-9.10418260669515\\
0.221	-9.28276336614868\\
0.231	-9.45197808171086\\
0.241	-9.61256655870905\\
0.251	-9.76517585744646\\
0.261	-9.91037465479619\\
0.271	-10.0486648944957\\
0.281	-10.1804913152726\\
0.291	-10.3062493015306\\
0.301	-10.4262913959748\\
0.311	-10.5409327357741\\
0.321	-10.6504556157642\\
0.331	-10.7551133383672\\
0.341	-10.8551334765006\\
0.351	-10.9507206500633\\
0.361	-11.0420588966615\\
0.371	-11.1293137016537\\
0.381	-11.2126337403087\\
0.391	-11.2921523751136\\
0.401	-11.3679889434601\\
0.411	-11.4402498646407\\
0.421	-11.5090295899739\\
0.431	-11.5744114156911\\
0.441	-11.6364681747625\\
0.451	-11.6952628209609\\
0.461	-11.75084891604\\
0.471	-11.8032710288429\\
0.481	-11.8525650533768\\
0.491	-11.8987584513281\\
0.501	-11.9418704231012\\
0.511	-11.9819120101889\\
0.521	-12.0188861304951\\
0.531	-12.0527875470845\\
0.541	-12.0836027697067\\
0.551	-12.1113098872898\\
0.561	-12.1358783283981\\
0.571	-12.1572685453535\\
0.581	-12.1754316163007\\
0.591	-12.1903087578952\\
0.601	-12.2018307394669\\
0.611	-12.2099171873848\\
0.621	-12.2144757658469\\
0.631	-12.2154012173391\\
0.641	-12.2125742424202\\
0.651	-12.2058601941416\\
0.661	-12.1951075570836\\
0.671	-12.1801461744291\\
0.681	-12.1607851783368\\
0.691	-12.1368105686724\\
0.701	-12.1079823722861\\
0.711	-12.0740312986714\\
0.721	-12.0346547869009\\
0.731	-11.9895123116893\\
0.741	-11.9382197812013\\
0.751	-11.8803428129013\\
0.761	-11.8153886122186\\
0.771	-11.7427960962447\\
0.781	-11.6619237925948\\
0.791	-11.5720348895422\\
0.801	-11.4722785990163\\
0.811	-11.3616666909349\\
0.821	-11.2390436221913\\
0.831	-11.1030480479906\\
0.841	-10.9520625568325\\
0.851	-10.7841470311509\\
0.861	-10.5969487944123\\
0.871	-10.3875791218235\\
0.881	-10.1524397873566\\
0.891	-9.88697325364278\\
0.901	-9.58529226181158\\
0.911	-9.23961144372102\\
0.921	-8.8393386495214\\
0.931	-8.3695478998509\\
0.941	-7.80824830969748\\
0.951	-7.12109277552359\\
0.961	-6.24996678042144\\
0.971	-5.08432864304534\\
0.981	-3.36980817761104\\
0.991	-0.247574897446142\\
};
\addlegendentry{LS}

\addplot [color=mycolor2, dashdotted, line width=1.0pt]
  table[row sep=crcr]{%
0.001	-0.338788881844735\\
0.011	-2.37836950997842\\
0.021	-3.55367213898155\\
0.031	-4.41032459868902\\
0.041	-5.09264050715015\\
0.051	-5.66276921888778\\
0.061	-6.15377119532222\\
0.071	-6.5855433774293\\
0.081	-6.97108284165921\\
0.091	-7.31939200015261\\
0.101	-7.63698740887294\\
0.111	-7.92875122527047\\
0.121	-8.19844339202249\\
0.131	-8.4490257733226\\
0.141	-8.682875945707\\
0.151	-8.90193313961925\\
0.161	-9.10780079978566\\
0.171	-9.3018204736029\\
0.181	-9.48512620122876\\
0.191	-9.65868531310652\\
0.201	-9.82332954346028\\
0.211	-9.9797791102376\\
0.221	-10.1286615980712\\
0.231	-10.2705269416272\\
0.241	-10.4058594418223\\
0.251	-10.5350874956958\\
0.261	-10.6585915440748\\
0.271	-10.7767106152192\\
0.281	-10.8897477515219\\
0.291	-10.997974539564\\
0.301	-11.1016349142676\\
0.311	-11.2009483707091\\
0.321	-11.2961126889516\\
0.331	-11.3873062556652\\
0.341	-11.4746900496106\\
0.351	-11.5584093450477\\
0.361	-11.6385951769079\\
0.371	-11.7153656034584\\
0.381	-11.7888267957237\\
0.391	-11.8590739777231\\
0.401	-11.9261922373663\\
0.411	-11.9902572244111\\
0.421	-12.0513357490513\\
0.431	-12.1094862923554\\
0.441	-12.1647594378097\\
0.451	-12.2171982315553\\
0.461	-12.2668384774873\\
0.471	-12.3137089721479\\
0.481	-12.3578316832621\\
0.491	-12.3992218747907\\
0.501	-12.4378881804818\\
0.511	-12.4738326270732\\
0.521	-12.5070506074919\\
0.531	-12.5375308036121\\
0.541	-12.5652550573302\\
0.551	-12.5901981878827\\
0.561	-12.6123277524498\\
0.571	-12.6316037461116\\
0.581	-12.6479782361528\\
0.591	-12.6613949244849\\
0.601	-12.6717886305529\\
0.611	-12.6790846854545\\
0.621	-12.6831982260735\\
0.631	-12.6840333757325\\
0.641	-12.681482295123\\
0.651	-12.6754240839379\\
0.661	-12.6657235095825\\
0.671	-12.6522295343668\\
0.681	-12.6347736064394\\
0.691	-12.6131676720871\\
0.701	-12.5872018574555\\
0.711	-12.5566417556806\\
0.721	-12.5212252400789\\
0.731	-12.4806587043997\\
0.741	-12.4346126057608\\
0.751	-12.3827161528243\\
0.761	-12.3245509382717\\
0.771	-12.2596432568578\\
0.781	-12.1874547727279\\
0.791	-12.107371094266\\
0.801	-12.018687669714\\
0.811	-11.9205922145459\\
0.821	-11.8121425952613\\
0.831	-11.6922386822941\\
0.841	-11.5595860811998\\
0.851	-11.4126487493587\\
0.861	-11.2495861273835\\
0.871	-11.0681682559662\\
0.881	-10.8656588720007\\
0.891	-10.6386506954578\\
0.901	-10.382827142519\\
0.911	-10.092606738967\\
0.921	-9.76059249940685\\
0.931	-9.37668014207913\\
0.941	-8.92653091712422\\
0.951	-8.38876301168096\\
0.961	-7.72927243264443\\
0.971	-6.88811140667877\\
0.981	-5.7422851257472\\
0.991	-3.95279171208567\\
};
\addlegendentry{RLS}

\addplot [color=red, line width=1.0pt]
  table[row sep=crcr]{%
0.011	-2.33219421167247\\
0.021	-4.42879236186204\\
0.031	-5.94205631231233\\
0.041	-7.0988503123387\\
0.051	-8.0194466809389\\
0.061	-8.77743054362166\\
0.071	-9.41893414937155\\
0.081	-9.97369922915514\\
0.091	-10.4616140430618\\
0.101	-10.8964930500635\\
0.111	-11.2882748145159\\
0.121	-11.6443323839267\\
0.131	-11.9702809480058\\
0.141	-12.2704933566278\\
0.151	-12.5484403166076\\
0.161	-12.8069217982053\\
0.171	-13.0482287421464\\
0.181	-13.2742587811184\\
0.191	-13.4866008084443\\
0.201	-13.6865979377059\\
0.211	-13.8753951539237\\
0.221	-14.0539759133566\\
0.231	-14.2231906289115\\
0.241	-14.3837791059071\\
0.251	-14.5363884046436\\
0.261	-14.681587201993\\
0.271	-14.8198774416924\\
0.281	-14.9517038624693\\
0.291	-15.0774618487272\\
0.301	-15.1975039431714\\
0.311	-15.3121452829707\\
0.321	-15.4216681629608\\
0.331	-15.5263258855638\\
0.341	-15.6263460236972\\
0.351	-15.7219331972599\\
0.361	-15.8132714438581\\
0.371	-15.9005262488503\\
0.381	-15.9838462875053\\
0.391	-16.0633649223102\\
0.401	-16.1392014906567\\
0.411	-16.2114624118373\\
0.421	-16.2802421371705\\
0.431	-16.3456239628877\\
0.441	-16.4076807219591\\
0.451	-16.4664753681576\\
0.461	-16.5220614632366\\
0.471	-16.5744835760396\\
0.481	-16.6237776005734\\
0.491	-16.6699709985247\\
0.501	-16.7130829702978\\
0.511	-16.7531245573855\\
0.521	-16.7900986776917\\
0.531	-16.8240000942811\\
0.541	-16.8548153169033\\
0.551	-16.8825224344864\\
0.561	-16.9070908755947\\
0.571	-16.9284810925501\\
0.581	-16.9466441634973\\
0.591	-16.9615213050918\\
0.601	-16.9730432866636\\
0.611	-16.9811297345815\\
0.621	-16.9856883130436\\
0.631	-16.9866137645357\\
0.641	-16.9837867896168\\
0.651	-16.9770727413382\\
0.661	-16.9663201042802\\
0.671	-16.9513587216257\\
0.681	-16.9319977255334\\
0.691	-16.9080231158691\\
0.701	-16.8791949194828\\
0.711	-16.845243845868\\
0.721	-16.8058673340976\\
0.731	-16.760724858886\\
0.741	-16.709432328398\\
0.751	-16.6515553600979\\
0.761	-16.5866011594152\\
0.771	-16.5140086434413\\
0.781	-16.4331363397914\\
0.791	-16.3432474367388\\
0.801	-16.2434911462129\\
0.811	-16.1328792381315\\
0.821	-16.0102561693879\\
0.831	-15.8742605951872\\
0.841	-15.7232751040291\\
0.851	-15.5553595783476\\
0.861	-15.3681613416089\\
0.871	-15.1587916690201\\
0.881	-14.9236523345532\\
0.891	-14.6581858008396\\
0.901	-14.3565048090099\\
0.911	-14.0108239909322\\
0.921	-13.6105511968572\\
0.931	-13.1407604485887\\
0.941	-12.5794608767843\\
0.951	-11.8923056251456\\
0.961	-11.0211848162493\\
0.971	-9.85566220989074\\
0.981	-8.14435364877179\\
0.991	-5.13950311304858\\
};
\addlegendentry{Box-RLS}

\addplot [color=mycolor3, line width=1.0pt, forget plot]
  table[row sep=crcr]{%
0.631	-20\\
0.631	-19\\
0.631	-18\\
0.631	-17\\
0.631	-16\\
0.631	-15\\
0.631	-14\\
0.631	-13\\
0.631	-12\\
0.631	-11\\
0.631	-10\\
0.631	-9\\
0.631	-8\\
0.631	-7\\
0.631	-6\\
0.631	-5\\
0.631	-4\\
0.631	-3\\
0.631	-2\\
0.631	-1\\
0.631	0\\
0.631	1\\
0.631	2\\
0.631	3\\
0.631	4\\
0.631	5\\
0.631	6\\
0.631	7\\
0.631	8\\
0.631	9\\
0.631	10\\
0.631	11\\
0.631	12\\
0.631	13\\
0.631	14\\
0.631	15\\
};
\node at (axis cs: 0.77,-5) (nodeA) {\scriptsize {$\alpha_{\star} = 0.629$}};
  \draw[black,thin,->] (nodeA) - - (axis cs: 0.631,-5)   ;
\end{axis}
\end{tikzpicture}%
\caption{\scriptsize{MSE as a function of the data power ratio $\alpha$.}}
\label{fig:MSE_alpha}
\end{figure}
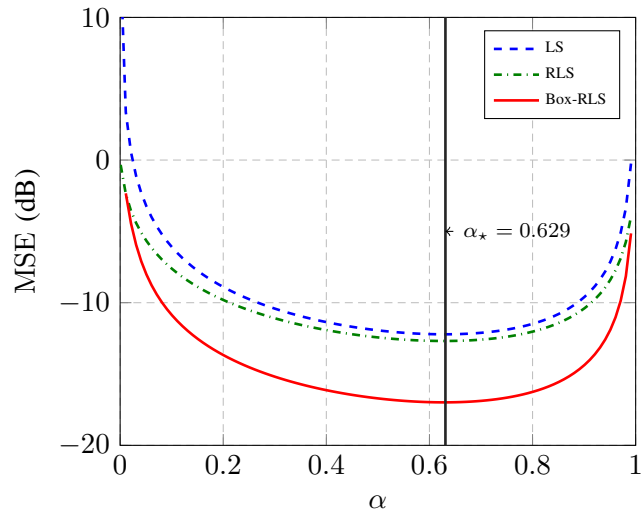
\begin{figure}[hh]
\centering
\captionsetup{justification=centering,margin=2cm}
\input{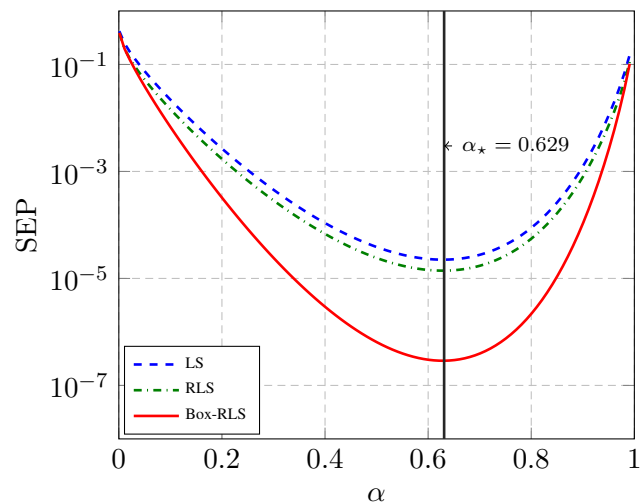}
\caption{\scriptsize{SEP as a function of the data power ratio $\alpha$.}}
\label{fig:SER_alpha}
\end{figure}
\subsection{Joint Optimization of Power Allocation and Training Duration in the Goodput Sense}
We now consider the goodput metric for the joint optimization of the power allocation and the training duration. From its  definition in \eqref{good_def}, its asymptotic value can be written as:   
\begin{equation}\label{Goodput_limit}
\underset{K\to\infty}{{\rm{plim}}} G = \bigg(1 -\frac{\tau_ {p}}{\tau}\bigg) \big(1- \underset{K\to\infty}{{\rm{plim}}} {\rm{SEP}} \big),
\end{equation}
where the limit in the right hand side is given by (\ref{SER-RLS1}) for RLS, and by \eqref{eq:SEP_BoXc} for Box-RLS. The above expression can be used to find the optimal pair $(\tau_{ {p}}^{\star},\alpha_\star)$ that \textit{maximizes} the goodput limit in (\ref{Goodput_limit}).    
The result is summarized below.
\begin{proposition}({{Joint optimization in goodput sense}}): \normalfont
The optimal pair $(\tau_{ {p}}^{\star},\alpha_\star)$ that maximizes the goodput limit in (\ref{Goodput_limit}) is given by:
$\tau_{ {p}}^{\star} =1$ (or $T_{ {p}}^{\star} = K$), and $\alpha_\star$ is the same as in \eqref{optimal_power} for all $\rho$ and $\tau$ (or $T$).
\end{proposition}
\begin{proof}
The proof of this proposition is given in Appendix \ref{goodput_proof}. 
\end{proof}
\begin{remark}
A major outcome of the above result is that the optimal number of training symbols that maximizes the goodput is given by the minimum number of the required training symbols that is the number of transmit antennas, $K$. This result differs  from the finding of \cite{hassibi2003much} in which it has been proven that in case of equal distribution of power between training and data, the optimal number of training symbols may be larger than the number of transmit antennas. 
\end{remark}
\section{Conclusions}
\label{sec:conc}
Based on the CGMT framework, this work carries out a large-system performance analysis of the regularized least squares (RLS) and box-regularized least squares (Box-RLS) decoders used to recover signals from $M$-ary constellations when the channel matrix is estimated using the LMMSE and  is modeled by i.i.d. real Gaussian entries.  Although our analysis relies on asymptotic growth assumptions, numerical results demonstrated the accuracy of the theoretical predictions even for limited system dimensions. Compared to previous related works, the main feature of the present work is our consideration of imperfect CSI, which allowed us to derive the optimal power allocation between training and data. However, considering imperfect CSI, posed several technical challenges and brought us to develop novel technical tools to establish convergence results. We believe that these results can be leveraged in the future to further facilitate carrying out rigorous analysis based on the use of the CGMT framework.   
\section*{Acknowledgment}
We would like to thank Prof. Ahmed-Sultan Salem for the very helpful comments, discussions and suggestions. We also would like to thank Houssem Sifaou for fruitful discussions.
\appendices
\section{Gaussian Min-max Theorem}
	The key ingredient of the analysis is the \textbf{C}onvex \textbf{G}aussian \textbf{M}in-max \textbf{T}heorem (\textbf{CGMT}), a concrete formulation for it can be found in \cite{thrampoulidis2018precise}. 
	The CGMT is a tool that allows analyzing the behavior of solutions of stochastic optimization problems that can be cast into the following form:
	\begin{equation} \label{primary}
		\Phi(\Gm) := \underset{\wv \in \mathcal{S}_{w}}{\operatorname{\min}}  \ \underset{\uv \in \mathcal{S}_{u}}{\operatorname{\max}} \ \uv^{T} \Gm \wv + \psi( \wv, \uv), \\
	\end{equation}
	where ${\bf G}\in\mathbb{R}^{N\times K}$ with i.i.d. standard normal entries, $\mathcal{S}_w$ and $\mathcal{S}_{u}$ are  sets of $\mathbb{R}^K$ and $\mathbb{R}^{N}$ and  $\psi:\mathbb{R}^K\times \mathbb{R}^{N}\to \mathbb{R}$ is continuous convex-concave function on $\mathcal{S}_w\times \mathcal{S}_u$. Problem \eqref{primary} is referred to as Primary Problem (PO) and its analysis is in general not tractable. The CGMT associates with it an Auxiliary Optimization (AO) problem given  by     
\begin{align}
&\phi(\gv, \sv) := \underset{\wv \in \mathcal{S}_{w}}{\operatorname{\min}}  \ \underset{\uv \in \mathcal{S}_{u}}{\operatorname{\max}} \ \| \wv \| \gv^{T} \uv - \| \uv \| \sv^{T} \wv + \psi( \wv, \uv), \label{AA2}
\end{align}
where $\gv \in   \mathbb{R}^{N}$, and  $\sv \in \mathbb{R}^K$ have i.i.d. standard Gaussian entries. 
	The initial formulation of the CGMT establishes that the (AO) has  the same asymptotic behavior as the (PO) in the regime in which $N$ and $K$ grow simultaneously with the same pace under the condition that the sets $\mathcal{S}_{{w}}$ and $\mathcal{S}_{{u}}$ are convex and compacts. Particularly, if for some $\nu\in\mathbb{R}$, the optimal cost of the (AO) concentrates around $\nu$ in the sense that
	$$
	\mathbb{P}\left[\left|\phi({\bf g},{\bf s})-\nu\right|\right]\to 0,
	$$
	then the optimal cost of the PO concentrates also around $\nu$, satisfying similarly:
	$$
	\mathbb{P}\left[\left|\Phi({\bf G})-\nu\right|\right]\to 0.
	$$
	Recently in \cite{kammoun_christos}, the compactness of $\mathcal{S}_{{ u}}$ is shown to be possibly relaxed provided that the order of the min-max in \eqref{AA2} can be inverted, that is $\phi({\bf g},{\bf s})$ is also given by: 
	\begin{equation}
		\phi({\bf g},{\bf s}):=\max_{{\bf u}\in\mathcal{S}_{ u}} \min_{{\bf w}\in\mathcal{S}_w} \|{\bf w}\|{\bf g}^{T}{\bf u} - \|{\bf u}\| {\bf s}^{T}{\bf w} +\psi({\bf w},{\bf u}). \label{eq:inversion}
	\end{equation}
	More formally, we have the following result: 
	\begin{theorem}[CGMT \cite{thrampoulidis2018precise}]
		Let $\mathcal{S}$ be any arbitrary open subset of $\mathcal{S}_w $, and $\mathcal{S}^c = \mathcal{S}_w \setminus\mathcal{S}$. Denote $\phi_{\mathcal{S}^c}(\gv,\sv)$ the optimal cost of the optimization in (\ref{AA2}), when the minimization over $\wv$ is constrained over $\wv \in \mathcal{S}^c$. Assume that $\mathcal{S}_u$ is convex while $\mathcal{S}_w$ is convex and compact. Assume also that \eqref{eq:inversion} holds true.  Consider the regime $K, N\to\infty$ such that $\frac{N}{K}\to \delta$, which will be denoted by $K\to\infty$.  Suppose that there exist constants $\bar{\phi}$ and $\eta >0$ such that in the limit as $K \rightarrow + \infty$ it holds with probability approaching one: (i) $\phi(\gv,\sv) \leq \bar{\phi} +\eta$, and, (ii) $\phi_{\mathcal{S}^c}(\gv,\sv) \geq \bar{\phi} + 2\eta$. Let ${\bf w}_\Phi$ and ${\bf w}_\phi$ denote respectively the solutions in ${\bf w}$ to the (PO) and the (AO). 
Then, $\lim_{K \rightarrow \infty} \mathbb{P}[\wv_{\phi} \in \mathcal{S}] = 1$, and $\lim_{K \rightarrow \infty} \mathbb{P}[\wv_{\Phi} \in \mathcal{S}] = 1.$ 
		\label{eq:th_CGMT}
\end{theorem}
	\begin{remark}
		It is worth mentioning that the result in Theorem \ref{eq:th_CGMT} goes beyond  the asymptotic equivalence between the costs of the (AO) and the (PO) to  the localization of the (PO) and (AO) solutions. More specifically,	one can easily see that conditions $(i)$ and $(ii)$ in Theorem~\ref{eq:th_CGMT} imply that the solution of the (AO) lies in the set $\mathcal{S}$ with probability approaching $1$.  Theorem \ref{eq:th_CGMT} allows us to carry over this property to  the solution of the (PO), that is $\wv_{\Phi}$ is in $\mathcal{S}$ with probability approaching $1$. 
	\end{remark}
	\begin{remark}
		\label{remark_imp}
		To satisfy $(i)$ and $(ii)$ in Theorem \ref{eq:th_CGMT}, one can prove that $\phi({\bf g},{\bf s})$ converges  to $\overline{\phi}$   while $\phi_{\mathcal{S}^{c}}({\bf g},{\bf s})$ is lower-bounded by a quantity that converges to $\overline{\phi}_{\mathcal{S}^c}$    with  \begin{equation}\overline{\phi}_{\mathcal{S}^c}>\overline{\phi}. \label{eq:cond_eq}\end{equation} In practice, it is usually the case that $\overline{\phi}$ and $\overline{\phi}_{\mathcal{S}^{c}}$ represent optimal costs of the same  optimization problem but with the solution of the latter being constrained to be  away from the optimal solution of the former. Under this setting, showing that 
			the optimization problem  whose optimal cost is $\overline{\phi}$ admits a unique solution directly implies \eqref{eq:cond_eq}.  
	\end{remark}
\section{Proofs of Box-RLS}
\label{sec:box}
In this appendix we prove Theorem~\ref{Box-RLS MSE} and Theorem \ref{Box_RLS SER}. For simplicity, we will divide the steps of the proof into subsections.
	\subsection{Identifying the (PO) and the (AO)}
For convenience, we consider the error vector $\wv := \xv- \xv_0 $, and also the box set:
\begin{equation}
\mathcal{B}= \bigg\{ \wv \in \mathbb{R}^K | -t - x_ {{0,j}}  \leq w_j \leq t- x_ {{0,j}}, j \in \{1,2, \cdots, K\} \bigg\},
\end{equation}
With this notation, the problem in (\ref{eq:Box-RLS matrix}) can be reformulated as
\begin{equation}\label{Lasso_w}
\widehat{\wv} = \text{arg} \ \underset{\wv \in \mathcal{B}}{\operatorname{\min}} \ \biggl\|  \sqrt{\frac{\rho_{ {d}}}{K}}\widehat{\Hm} \wv -\sqrt{\frac{\rho_{ {d}}}{K}} \Deltam \xv_0 - \zv \biggr\|^2+ \rho_{ {d}}\lambda \| \xv_0 +\wv  \|^2.
\end{equation}
\noindent
To bring the problem in (\ref{Lasso_w}) to the form of \eqref{primary} required by the CGMT, we express the loss function of (\ref{Lasso_w}) in its dual form through the Fenchel's conjugate
$$\ \left\|  \sqrt{\frac{\rho_{ {d}}}{K}}\widehat{\Hm} \wv -\sqrt{\frac{\rho_{ {d}}}{K}} \Deltam \xv_0 - \zv \right\|^2
= \max_{\uv}  \uv^T \left(\sqrt{\frac{\rho_{ {d}}}{K}}\widehat{\Hm} \wv -\sqrt{\frac{\rho_{ {d}}}{K}} \Deltam \xv_0 - \zv \right) - \frac{\| \uv \|^2}{4}.$$
Hence, the problem in (\ref{Lasso_w}) is equivalent to the following:
\begin{align}\label{minmax_1}
	\underset{\wv \in \mathcal{B}}{\operatorname{\min}} \ \underset{\uv}{\operatorname{\max}} &\  \uv^T  \sqrt{\frac{\rho_{ {d}}}{K}}\widehat{\Hm} \wv -\uv^T \sqrt{\frac{\rho_{ {d}}}{K}} \Deltam \xv_0 - \uv^T \zv - \frac{\| \uv \|^2}{4}  + \rho_{ {d}}\lambda \| \xv_0 +\wv \|^2, 
\end{align}
To reach the desired PO form, we introduce the variables
$\vv= \begin{bmatrix}
       \sigma_{ {\hat{H}}} \sqrt{\rho_{ {d}}}\wv            \\[0.3em]
        -\sigma_{ {\Delta}} \sqrt{\rho_{ {d}}}{\xv}_0 \\[0.3em]
     \end{bmatrix}$ $\in \mathbb{R}^{2K}$, 
     $\Gm= \begin{bmatrix}
       \widetilde{\Hm} &    \widetilde{\Deltam}      \\[0.3em]
     \end{bmatrix}$ $\in \mathbb{R}^{N \times 2K}$ and
     $\Cm= \begin{bmatrix}
       \frac{1}{\sigma_{ {\hat{H}}}}\Id_{K} &   \frac{-1}{ \sigma_{ {\Delta}}}  \Id_{K}    \\[0.3em]
     \end{bmatrix}$ $\in \mathbb{R}^{K \times 2K}$, where $\widetilde{\Hm}$ and $\widetilde{\Deltam}$ are $N \times K$ independent matrices with i.i.d. standard normal entries. Now,   using these variables, and after normalization by $1/K$, the above problem can be written as: 
\begin{align}
	{\Phi}:=&\frac{1}{\sqrt{K}}\underset{\vv \in \mathcal{D}}{\operatorname{\min}} \ \underset{\uv}{\operatorname{\max}} \ \frac{1}{K} \uv^T  \Gm \vv - \frac{1}{\sqrt{K}}\uv^T \zv - \frac{\| \uv \|^2}{4\sqrt{K}}  + \frac{\lambda}{\sqrt{K}} \| \Cm \vv \|^2, \label{POprob}
\end{align}
where $\mathcal{D} = \{\vv^T \in \mathbb{R}^{1 \times 2K}, \ \ {\bf v}= [ \sigma_{ {\hat{H}}} \sqrt{\rho_{ {d}}}{\wv}^T \quad -\sigma_{ {\Delta}} \sqrt{\rho_{ {d}}}{{\xv}_0}^T]^{T} |
 {\wv} \in \mathcal{B}\}$.
We note that \eqref{POprob} is in the form of the (PO) and associate with it the following (AO) 
\begin{align}\label{eq:AO1}
	\phi:=\frac{1}{\sqrt{K}}&\underset{\vv \in \mathcal{D} }{\operatorname{\min}} \ \underset{\uv}{\operatorname{\max}} \ \| \vv \| \frac{1}{K}\gv^T \uv - \frac{1}{K}\| \uv \| \qv^T \vv  -  \frac{1}{\sqrt{K}} \uv^T \zv  - \frac{1}{4\sqrt{K}} \| \uv\|^2 + \frac{1}{\sqrt{K}} \lambda \| \Cm \vv \|^2,
\end{align}
where $\qv \in \mathbb{R}^{2K}$ and $\gv \in \mathbb{R}^N$ are independent standard normal vectors. Note that for the moment, we relate the (PO) to the unbounded (AO) as $\mathcal{S}_u=\mathbb{R}^{N}$ is not compact. In the sequel, we check that \eqref{eq:inversion} holds true, which gives support to considering the unbounded (AO) according to Theorem \ref{eq:th_CGMT}. 
\subsection{Scalarizing the (AO)} 
The next step is to simplify the (AO) as it appears in (\ref{eq:AO1}) into an optimization problem involving only scalar variables. Since the vectors $\gv$ and $\zv$ are independent and have i.i.d. Gaussian entries, then, 
$  \| \vv \| \gv- \sqrt{K} \zv$ has i.i.d. entries $\mathcal{N}(0,\| \vv \|^2 + K)$. Hence, for our purposes and using some abuse of notation so that $\gv$ continues to denote a vector with i.i.d. standard normal entries, the corresponding terms in (\ref{eq:AO1}) can be combined as $\sqrt{\| \vv \|^2 + K}\gv^T \uv$, instead. {{Therefore, (\ref{eq:AO1}) is equivalent to}}
\begin{align}\label{simpleAO1}
	\phi=& \ \underset{\vv \in \mathcal{D}}{\operatorname{\min}} \ \underset{\uv}{\operatorname{\max}} \ \frac{\gv^T \uv}{K}\sqrt{\frac{1}{K}\| \vv \|^2 + 1} -\frac{1}{K\sqrt{K}} \| \uv \| \qv^T \vv  - \frac{1}{4 {K}} \| \uv \|^2 + \frac{1}{{K}} \lambda\| \Cm \vv \|^2 .
\end{align}
Expressing the above problem in terms of the original $\wv$ variable:
{{
\begin{align}\label{PO20}
	\phi\!=\! \underset{{\wv}\in \mathcal{B}}{\operatorname{\min}} \ \underset{\uv}{\operatorname{\max}} & \ \frac{\gv^T \uv}{K}\sqrt{\sigma_{ {\hat{H}}}^2 \rho_{ {d}} \frac{1}{K}\| {\wv} \|^2 + \sigma_{ {\Delta}}^2 \rho_{ {d}} \frac{1}{K}\| {\xv}_0\|^2+ 1} - \frac{1}{4{K}} \| \uv \|^2\ \nonumber\\
	&-\!\frac{\| \uv \|}{K\sqrt{K}}\sqrt{\rho_{ {d}}} (\sigma_{ {\hat{H}}}  {\wv}^T{\qv^{\boldsymbol{1}}} -\sigma_{ {\Delta}} {\xv}_0^T{\qv^{\boldsymbol{2}}} ) + \frac{1}{{K}} \lambda \rho_{ {d}} \| {\xv}_0 +{\wv} \|^2,
\end{align}}}
\noindent
where $\qv^{\boldsymbol{1}},\qv^{\boldsymbol{2}} \in \mathbb{R}^K$ are independent standard normal vectors.

Fixing the norm of $\frac{\uv}{\sqrt{K}}$ to ${\beta}: =\frac{\| \uv \|}{\sqrt{K}}$, it is easy to see that its optimal direction should be aligned   with $\gv$. Working with ${\bf x}$ instead of ${\bf w}$   results  into the following optimization problem:
\begin{align}
\phi=\min_{\substack{-t\leq x_j\leq t\\ j=1,\dots,K}} \max_{\beta>0} & \ \beta \frac{\| \gv \|}{\sqrt{K}} \sqrt{\frac{\rho_{ {d}}}{K}( \sigma_{ {\hat{H}}}^2 \|  {\xv} \|^2 -2\sigma_{\hat{H}}^2{\bf x}_0^{T}{\bf x}+\|{\bf x}_0\|^2) +1} -\frac{\beta^2}{4} \nonumber\\
&-\beta\sqrt{\rho_d}\sigma_{\hat{H}}\frac{1}{K}{\bf x}^{T}{\bf q}^{\bf 1}+ \frac{1}{K}\beta \sqrt{\rho_d}{\bf x}_0^{T}\tilde{\bf h}+\lambda \rho_d\frac{1}{K}\|{\bf x}\|^2,\label{eq:pr}
\end{align}
where $\tilde{\bf h}=\sigma_{\hat{H}} {\bf q}^{\bf 1}+\sigma_{\Delta}{\bf q}^{\bf 2}$. 
At this point, it is worth mentioning that \eqref{eq:pr} is convex in ${\bf x}$ and concave in $\beta$. Based on this, we can prove that \eqref{eq:inversion} holds true. Indeed using \cite[Cor. 37.3.2]{rockafellar2015convex},  we can flip the order of $\min_{\bf x}\max_{\beta}$. Now, consider the following problem:
\begin{align}
	\max_{\beta} \max_{\substack{\tilde{\bf u}\\ \|\tilde{\bf u}\|=1} }\min_{\substack{-t\leq x_j\leq t\\ j=1,\dots,K}}& \  \beta {\gv^{T}\tilde{\bf u}} \sqrt{\frac{\rho_{ {d}}}{K}( \sigma_{ {\hat{H}}}^2 \|  {\xv} \|^2 -2\sigma_{\hat{H}}^2{\bf x}_0^{T}{\bf x}+\|{\bf x}_0\|^2) +1} -\frac{\beta^2}{4}\nonumber\\
&-\beta\sqrt{\rho_d}\sigma_{\hat{H}}\frac{1}{K}{\bf x}^{T}{\bf q}^{\bf 1}+ \frac{1}{K}\beta \sqrt{\rho_d}{\bf x}_0^{T}\tilde{\bf h}+\lambda \rho_d\frac{1}{K}\|{\bf x}\|^2.\label{eq:pr1}
\end{align}
Then, based on \cite[Lemma 8]{svm_abla} and setting ${\bf u}=\sqrt{K}\beta\tilde{\bf u}$ we can prove that \eqref{eq:pr1} is the same as \eqref{eq:pr} in which the order of the min-max is inverted. This completes the proof of \eqref{eq:inversion}, which as aforementioned, allows us to extend the scope of the CGMT for optimization problems in which the variable ${\bf u}$ is constrained to lie in a non-compact set.  

Now, getting back to the optimization problem in \eqref{eq:pr} and  flipping the order of $\min_{\bf x}\max_{\beta}$ results into the following optimization problem:
\begin{align}
\phi=\max_{\beta>0}  \min_{\substack{-t\leq x_j\leq t\\ j=1,\dots,K}} & \nonumber\\
\hat{\mathcal{H}}(\beta,{\bf x}):=&\beta \frac{\| \gv \|}{\sqrt{K}} \sqrt{\frac{\rho_{ {d}}}{K}( \sigma_{ {\hat{H}}}^2 \|  {\xv} \|^2 -2\sigma_{\hat{H}}^2{\bf x}_0^{T}{\bf x}+\|{\bf x}_0\|^2) +1} -\frac{\beta^2}{4}\nonumber\\
&-\beta\sqrt{\rho_d}\sigma_{\hat{H}}\frac{1}{K}{\bf x}^{T}{\bf q}^{\bf 1}+ \beta \frac{1}{K}\sqrt{\rho_d}{\bf x}_0^{T}\tilde{\bf h}+\lambda \rho_d\frac{1}{K}\|{\bf x}\|^2.\label{eq:pr2}
\end{align}
Prior to proceeding further, we shall first check that the optimization over $\beta$ of the above problem is not achieved in the limit $\beta\to 0$ and more specifically, there exists $\tilde{\delta}>0$ such that taking the supremum over $\beta>\tilde{\delta}$ instead of $\beta>0$ would almost surely not change the optimal cost of \eqref{eq:pr2}. Towards this goal, first note that 
$$
 \left|\min_{\substack{-t\leq x_j\leq t\\ j=1,\dots,K}}  \hat{\mathcal{H}}(\beta,{\bf x})\right|\leq \left|\beta \frac{\|{\bf g}\|}{\sqrt{K}}-\frac{\beta^2}{4}+\beta\left|\frac{1}{K}\sqrt{\rho_{ {d}}}{\bf x}_0^{T}\tilde{\bf h}\right|\right|,  $$
 and the function $\beta\mapsto \beta \left(\frac{\|{\bf g}\|}{\sqrt{K}}+\frac{1}{K}\sqrt{\rho_{ {d}}}\left|{\bf x}_0^{T}\tilde{\bf h}\right|\right)-\frac{\beta^2}{4}$ is increasing and positive for all $\beta\in \left[0,2 \left(\frac{\|{\bf g}\|}{\sqrt{K}}+\frac{1}{K}\sqrt{\rho_{ {d}}}\left|{\bf x}_0^{T}\tilde{\bf h}\right|\right)\right]$. As  $2 \left(\frac{\|{\bf g}\|}{\sqrt{K}}+\frac{1}{K}\sqrt{\rho_{ {d}}}\left|{\bf x}_0^{T}\tilde{\bf h}\right|\right)$ converges in probability to $2\sqrt{\delta}$, with probability approaching one, for all  
 $\tilde{\delta}\in(0,\sqrt{\delta})$, and all $\beta\in(0,\tilde{\delta})$,
 \begin{equation}
 \left|\min_{\substack{-t\leq x_j\leq t\\ j=1,\dots,K}}  \hat{\mathcal{H}}(\beta,{\bf x})\right|\leq 4\tilde{\delta}\sqrt{\delta}.
	 \label{eq:min_proof}
 \end{equation}
To conclude it suffices  to prove that there exists a $\beta_0$ such that with probability approaching one,
\begin{equation}
 \min_{\substack{-t\leq x_j\leq t\\ j=1,\dots,K}}  \hat{\mathcal{H}}(\beta_0,{\bf x}) > \Theta,
\label{eq:HH}
\end{equation}
where $\Theta$ is a some positive constant. Indeed, if \eqref{eq:HH} is satisfied then almost surely, 
$$
\sup_{\beta>0}\min_{\substack{-t\leq x_j\leq t\\ j=1,\dots,K}}\hat{\mathcal{H}}(\beta_0,{\bf x}) >\Theta.
$$
Setting $\tilde{\delta}=\min(\frac{\Theta}{4\sqrt{\delta}},\sqrt{\delta})$ in \eqref{eq:min_proof}, we conclude thus that the supremum over $\beta$ could not be attained in the interval $[0,\tilde{\delta}]$. 
To keep the flow of the proof, \eqref{eq:HH} is proved in Appendix \ref{app:pp}. With this result at hand, we are now ready to proceed to the optimization of the (AO).  
Let  $\chi=\frac{\rho_d}{K}\left(\sigma_{\hat{H}}^2\|{\bf x}\|^2-2\sigma_{\hat{H}}^2{\bf x}_0^T{\bf x}+\|{\bf x}\|_0^2\right)+1$.  
To make the above optimization problem separable, we express the term in the square root in a variational form using the identity:  $\sqrt{\chi} = \underset{r>0}{\operatorname{\min}} \ \frac{1}{2r} + \frac{r\chi}{2}$. Note that at optimum, $r_\star=\frac{1}{\sqrt{\chi}}$.  
Using the fact that $\chi\geq 1$ and as such bigger than any small positive constant, we also have 
$$
\sqrt{\chi}=\min_{0<r\leq C'} \frac{1}{2r}+\frac{r\chi}{2}.
$$
where $C'$ is any constant greater than $1$. Similarly, as $\chi$ is almost surely bounded by some constant, we can also argue that:
$$
\sqrt{\chi}=\min_{\epsilon' <r\leq C'} \frac{1}{2r}+\frac{r\chi}{2}.
$$
where $\epsilon'$ is a sufficiently small positive constant. 
Using this relation, the optimization problem \eqref{eq:pr1} becomes 
{{
\begin{align}\label{AA23}
	\phi= \underset{{\beta} \geq 0}{\operatorname{\max}}  \ \underset{\epsilon' <r\leq C'}{\operatorname{\min}} & \  {\beta} \biggl( \frac{  \| \gv \|}{2r\sqrt{K}} + \frac{ r\| \gv \| }{2 \sqrt{K}}  + \frac{ \rho_{ {d}} r\| \gv \| \| {\xv}_0 \|^2}{2 \sqrt{K} K }+ \sqrt{\rho_{ {d}}}  \frac{1}{K}\tilde{\hv}^{T} {\xv}_0 \biggr) \nonumber\\
	&+ \frac{1}{K}\sum_{j = 1}^{K} \biggl[  \underset{-t\leq x_j \leq t }{\operatorname{\min}} \ \biggl(\frac{{\beta} \sigma_{ {\hat{H}}}^2 \rho_{ {d}} r\| \gv \|}{2 \sqrt{K} } + {\lambda} \rho_{ {d}}\biggr) x_j^2  - {\beta} \biggl( \frac{\rho_{ {d}} \sigma_{ {\hat{H}}}^2 r\| \gv \| }{\sqrt{K} } {x}_{0,j}  +  \sqrt{\rho_{ {d}}}\sigma_{ {\hat{H}}} q^1_{j}  \biggr) {x}_j \biggr]-\frac{{\beta}^2}{4}. 
\end{align} 
}}
For ease of notation, define $\tilde{\theta} :=\frac{r\| \gv \|}{\sqrt{K}}$. As $\frac{1}{\sqrt{K}}\|{\bf g}\|$ is almost surely bounded above and below, the variable $\tilde{\theta}$ is almost surely bounded above by a constant $C$ which we shall assume as large as needed  and also bounded below by some positive constant $\epsilon$ . Introducing this notation leads to: 
{{
\begin{align}\label{AA23}
\phi=	 \ \underset{\beta > 0}{\operatorname{\max}} \ \underset{\epsilon\leq \tilde{\theta}  \leq C}{\operatorname{\min}} & \ \frac{  \beta \| \gv \|^2}{2\tilde{\theta}K} + \frac{ \beta \tilde{\theta}}{2 } - \frac{\beta^2}{4} + \frac{\beta \tilde{\theta}\rho_{ {d}} \| {\xv}_0 \|^2}{2 K}  + \sqrt{\rho_{ {d}}} \beta \frac{1}{K}{\xv}_0^{T} \tilde{\hv} \nonumber \\
	& + \frac{1}{K}\sum_{j = 1}^{K} \biggl[ \underset{ -t\leq {x}_j \leq t }{\operatorname{\min}} \ \biggl(\frac{\beta \tilde{\theta}\rho_{ {d}}\sigma_{ {\hat{H}}}^2}{2} + \lambda \rho_{ {d}} \biggr) x_j^2  -  \beta \biggl( \rho_{ {d}} \sigma_{ {\hat{H}}}^2 {x}_{0,j}\tilde{\theta}+\sqrt{\rho_{ {d}}} \sigma_{ {\hat{H}}} q^{1}_{j}   \biggr) {x}_j \biggr]. 
\end{align} 
For $\beta > 0$, the optimal solution in the variables ${x}_j$, $j=1,\dots,K$ of (\ref{AA23}) is given by:
\begin{equation}\label{AO_Sol_box}
\tilde{x}_j =
\begin{cases} 
	-t  ,\ \text{if}  \  q^1_{j} < x_{0}^{-}(\tilde{\theta},\beta,x_{0,j}), \\ 
           
	t  ,\ \text{if}  \  q^1_{j} > x_{0}^{+}(\tilde{\theta},\beta,x_{0,j}), \\      
   
		 \frac{\beta \biggl( {\rho_{ {d}} \sigma_{ {\hat{H}}}^2}x_{0,j}\tilde{\theta} + \sqrt{\rho_{ {d}}} \sigma_{ {\hat{H}}} q^1_{j} \biggr)}{\rho_{ {d}} \sigma_{ {\hat{H}}}^2 \beta \tilde{\theta} + 2 \lambda \rho_{ {d}}}  ,\ \text{otherwise},
\end{cases}
\end{equation}
where 
\begin{align}
	x_{0}^{-}(\tilde{\theta},\beta,x_{0,j})&=-t \biggl(\tilde{\theta}\sqrt{\rho_{ {d}}} \sigma_{ {\hat{H}}} + \frac{2 \lambda \rho_{ {d}}}{\beta\sqrt{\rho_{ {d}}}\sigma_{ {\hat{H}}} } \biggr) - \sqrt{\rho_{ {d}}} \sigma_{ {\hat{H}}} x_{0,j}\tilde{\theta},\label{eq:x0j-}\\
	x_{0}^{+}(\tilde{\theta},\beta,x_{0,j})& = t \biggl(\tilde{\theta}\sqrt{\rho_{ {d}}} \sigma_{ {\hat{H}}} + \frac{2 \lambda \rho_{ {d}}}{\sqrt{\rho_{ {d}}}\sigma_{ {\hat{H}}} \beta} \biggr) - \sqrt{\rho_{ {d}}} \sigma_{ {\hat{H}}} x_{0,j} \tilde{\theta}.\label{eq:x0j+}
\end{align}
{{To simplify notation, define $\xi = \sqrt{\rho_{ {d}}} \sigma_{ {\hat{H}}}$, then $x_{0}^{-}(\tilde{\theta},\beta,x_{0,j}) =-t\bigl({\xi}\tilde{\theta} + \frac{2 \lambda\rho_{ {d}}}{\xi \beta} \bigr) - {\xi} x_{0,j}\tilde{\theta}$, and $x_{0}^{+}(\tilde{\theta},\beta,x_{0,j}) =t\bigl({\xi}\tilde{\theta} + \frac{2 \lambda\rho_{ {d}}}{\xi \beta} \bigr) - {\xi} x_{0,j}\tilde{\theta}$}}.
With these notations at hand, the above optimization problem reduces to the following SO
\begin{align}\label{SO4}
	& \ \underset{\beta > 0}{\operatorname{\max}}  \ \underset{\epsilon <\tilde{\theta}<C }{\operatorname{\min}} \widetilde{D}(\tilde{\theta},\beta,\gv,\qv^{\boldsymbol{1}}) :=  \frac{  \beta \| \gv \|^2}{2\tilde{\theta}K} + \frac{ \beta \tilde{\theta} }{2 } - \frac{ \beta^2}{4}  + \frac{\beta\tilde{\theta} \rho_{ {d}} \| {\xv}_0 \|^2}{2 K} + \sqrt{\rho_{ {d}}} \beta   \frac{1}{K}\tilde{\hv}^{T} {\xv}_0+ \frac{1}{K}\sum_{j = 1}^{K} v (\tilde{\theta},\beta ; q^1_{j}),
\end{align} 
where
\begin{equation}\label{soft_TH}
	v( \tilde{\theta},\beta ; q^1_{j}) =
\begin{cases} 
	t \bigl(\tilde{c}_j +\beta \xi q^1_j \bigr), \ \text{if}  \   q^1_{j} < x_{0}^{-}(\tilde{\theta},\beta,x_{0,j}), \\ 
           
	       t \bigl( \tilde{d}_j - \beta  \xi q^1_j  \bigr), \ \text{if}   \  q^1_{j} > x_{0}^{+}(\tilde{\theta},\beta,x_{0,j}), \\      
   
		 -\frac{\beta^2 \bigl( \xi^2 x_{0,j} \tilde{\theta}+\xi  q^1_{j} \bigr)^2}{2\xi^2 \beta\tilde{\theta} + 4 \lambda \rho_{ {d}}}  ,\ \text{otherwise},
\end{cases}
\end{equation}
where $\tilde{c}_j = \frac{-\beta \xi}{2} x_{0}^{-}(\tilde{\theta},\beta,x_{0,j}) + \frac{\beta \tilde{\theta}\xi^2}{2 }x_{0,j}$, and $\tilde{d}_j = \frac{\beta \xi}{2} x_{0}^{+}(\tilde{\theta},\beta,x_{0,j}) - \frac{\beta \tilde{\theta}\xi^2}{2 }x_{0,j}$.
\subsection{ Asymptotic analysis of the SO problem}
After simplifying the (AO) as in (\ref{SO4}), we are now in a position to analyze its limiting behavior.
Using the Weak Law of Large Numbers (WLLN) \footnote{We write $\overset{P}{\longrightarrow}$ to denote convergence in probability as $K \to \infty$.},  $\frac{1}{K} \| \gv \|^2 \overset{P}{\longrightarrow} \frac{N}{K} := \delta$, $\frac{1}{K} {\xv}_0^T \tilde{\hv}  \overset{P}{\longrightarrow} 0$, and $\frac{1}{K} \|{\xv}_{0} \|^2 \overset{P}{\longrightarrow} 1$.

{{To analyze the behavior of the summand,  recall that each $x_{0,j}$ takes values $\pm1/\sqrt{\mathcal{E}}, \pm3/\sqrt{\mathcal{E}},..., \pm(M-1)/\sqrt{\mathcal{E}}$ with equal probability $1/M$.}} Let $i = \pm1, \pm3,..., \pm(M-1)$ and denote by
$\ell_i =-t(\xi\tilde{\theta}+\frac{2\lambda \rho_d}{\xi \beta})-\frac{\xi i\tilde{\theta}}{ \sqrt{\mathcal{E}}}$, $\mu_i =t(\xi\tilde{\theta}+\frac{2\lambda\rho_d}{\xi \beta})-\frac{\xi i \tilde{\theta}}{ \sqrt{\mathcal{E}}}$, $c_i = \frac{-\beta  \xi}{2} \ell_i + \frac{\beta \xi^2 i\tilde{\theta}}{2 \sqrt{\mathcal{E}}}$, and $d_i = \frac{\beta \xi}{2} \mu_i - \frac{\beta \xi^2 \tilde{\theta}i}{2 \sqrt{\mathcal{E}} }$.\\
Hence,  it can be shown that for all $\tilde{\theta}>0$ and $\beta > 0$,
\begin{equation*}
	\frac{1}{K} \sum_{j = 1}^{K} v ( \tilde{\theta},\beta ; q^1_{j}) \overset{P}{\longrightarrow} Y(\tilde{\theta},\beta),
\end{equation*} 
where 
\begin{equation}
	Y(\tilde{\theta},\beta):= \frac{1}{M} \sum_{i= \pm1,\pm3,\cdots,\pm(M-1)} \tilde{Y}(\tilde{\theta},\beta,i),
\end{equation}
with 
\begin{align}
	\tilde{Y}(\theta,\beta,i)&:=\mathbb{E}_{h \thicksim\mathcal{N}(0,1)} [v (\tilde{\theta},\beta;h,\ell_i,\mu_i) ]\nonumber\\
	&=   -\frac{\beta^2}{{2 \xi^2 \beta\tilde{\theta}} + 4 \lambda\rho_{ {d}}} \int_{\ell_{i}}^{\mu_{i}}  \bigg( \frac{\xi^2 i \tilde{\theta}}{\sqrt{\mathcal{E}} } + \xi  h \bigg)^2 p(h) {\rm{d}}h 
+\int_{- \infty}^{\ell_{i}} t(c_i +\beta  \xi h ) p(h) {\rm{d}}h+ \int_{\mu_{i}}^{\infty} t (d_i -\beta \xi  h ) p(h) {\rm{d}}h  \nonumber \\
	=& \ t \big(c_i Q(-\ell_i) + d_i Q(\mu_i) \big) - \beta \xi  t \big( p(\ell_i) + p(\mu_i) \big) 
	-\frac{\beta^2}{{2 \xi^2 \beta\tilde{\theta}} + 4 \lambda\rho_{ {d}}} \int_{\ell_{i}}^{\mu_{i}}  \bigg( \frac{\xi^2 i\tilde{\theta}}{\sqrt{\mathcal{E}} } + \xi  h \bigg)^2 p(h) \text{d}h \nonumber \\
 = \ & t \big(c_i Q(-\ell_i) + d_i Q(\mu_i) \big) -  \beta \xi t \big( p(\ell_i) + p(\mu_i) \big) 
	-\frac{\beta^2}{{2 \xi^2 \beta}\tilde{\theta} + 4 \lambda \rho_{ {d}}} \bigg( \bigg(\xi^2  + \frac{\xi^4 i^2\tilde{\theta}^2}{\mathcal{E} }\bigg)\big(Q(\ell_i) - Q(\mu_i) \big) \nonumber \\
	+& \xi \big(\xi \ell_i  + \frac{2 \xi^2 i \tilde{\theta}}{\sqrt{\mathcal{E}} } \big) p(\ell_i)- \xi \big(\xi\mu_i + \frac{2 \xi^2 i\tilde{\theta}}{\sqrt{\mathcal{E}} } \big) p(\mu_i)  \bigg).
\end{align}
For a given $\beta$, consider the sequence of functions $$\varphi_K:\tilde{\theta}\mapsto \frac{1}{K}\sum_{j=1}^{K}v( \tilde{\theta},\beta ; q^1_{j}). $$ 
We can easily see that this sequence of functions is concave in $\tilde{\theta}$ since it has been derived by taking the infimum of linear functions in $\tilde{\theta}$ (cf. \eqref{AA23}). Hence,  $\tilde{\theta}\mapsto Y(\theta,\beta)$ is concave in $\tilde{\theta}$. Since the convergence of concave functions is uniform over compact sets \cite[Theorem.II.1]{cox}, $\tilde{\theta}\mapsto \varphi_K(\tilde{\theta})$ converges uniformly to $\tilde{\theta}\mapsto Y(\tilde{\theta},\beta)$.
Moreover, it is easy to prove that $\tilde{\theta} \mapsto \frac{\beta \|{\bf g}\|^2}{2\tilde{\theta}K}+\frac{\beta\tilde{\theta}}{2}+\frac{\beta \tilde{\theta}\rho_{ {d}}\|{\bf x}_0\|^2}{2K}$ converges uniformly to $\tilde{\theta}\mapsto \frac{\beta\delta}{2\tilde{\theta}} +\frac{\beta\tilde{\theta}}{2}+\frac{\beta\tilde{\theta}\rho_{ {d}}}{2}$ on the compact set $[\epsilon,C]$. Combining both results yields that $\tilde{\theta}\mapsto \widetilde{D}(\tilde{\theta},\beta, {\bf g},{\bf q}^{\bf 1})$ converges uniformly to $\tilde{\theta}\mapsto \overline{D}(\theta,\beta)$, where
$$
\overline{D}(\tilde{\theta},\beta):=\frac{\beta\delta}{2\tilde{\theta}}-\frac{\beta^2}{4}+\frac{\beta\tilde{\theta}}{2}+\frac{\beta\tilde{\theta}\rho_{ {d}}}{2}+Y(\tilde{\theta},\beta).
$$
As a consequence,
$$
 \underset{\epsilon <\tilde{\theta}<C }{\operatorname{\min}} \widetilde{D}(\tilde{\theta},\beta,{\bf g},{\bf q}^{\bf 1})\to  \underset{\epsilon <\tilde{\theta}<C }{\operatorname{\min}} \overline{D}(\tilde{\theta},\beta).
$$
We need now to prove that the supremum over $\beta$ converges to the supremum of the right-hand side of the above equation. To this end, note that the function $\beta \mapsto \underset{\epsilon <\tilde{\theta}<C }{\operatorname{\min}} \widetilde{D}(\tilde{\theta},\beta,{\bf g},{\bf q}^{\bf 1})$ is concave and converges pointwise to $\beta\mapsto \underset{\epsilon <\tilde{\theta}<C }{\operatorname{\min}} \overline{D}(\tilde{\theta},\beta)$. Moreover, it is easy to check that $\lim_{\beta\to\infty} \underset{\epsilon <\tilde{\theta}<C }{\operatorname{\min}} \overline{D}(\tilde{\theta},\beta)=-\infty$. Using Lemma 10 in \cite{thrampoulidis2018precise}, we conclude that:
\begin{equation}
	\phi=	\sup_{\beta\geq 0} \underset{\epsilon <\tilde{\theta}<C }{\operatorname{\min}} \widetilde{D}(\tilde{\theta},\beta,{\bf g},{\bf q}^{\bf 1}) \overset{P}{\longrightarrow} \overline{\phi}:=\sup_{\beta\geq 0}\underset{\epsilon <\tilde{\theta}<C }{\operatorname{\min}} \overline{D}(\tilde{\theta},\beta). \label{eq:supre}
\end{equation}
\subsection{Proof of the uniqueness of $(\tilde{\theta}_\star,\beta_\star)$ solving $\sup_{\beta\geq 0} \min_{\epsilon<\tilde{\theta}<C}\overline{D}(\tilde{\theta},\beta)$}
Based on the above convergence, it follows from the CGMT that the optimal cost of the (PO) converges to the asymptotic limit of the (AO) which is given by $\sup_{\beta >0}\min_{\epsilon\leq \tilde{\theta} <C} \overline{D}(\tilde{\theta},\beta)$. However, our interest does not directly concern the characterization of the asymptotic limit of the (PO) but that of functionals of the  vector ${\bf w}={\bf x}-{\bf x}_0$ that can be  linked to some important metrics like MSE or SEP. As explained in Remark \ref{remark_imp}, proving that the max-min problem in \eqref{eq:supre} admits a unique solution $(\beta_\star,\tilde{\theta}_\star)$ would allow us to transfer any property of the solution of the (AO) to that of the (PO).  
Unfortunately, the objective in the max-min problem \eqref{eq:supre} is not convex in $\tilde{\theta}$, and hence the same approach pursued in \cite{thrampoulidis2018precise} could not be used here. A new approach to handle this problem is thus proposed. 
To begin with,  we notice that since $\beta\mapsto -\frac{\beta^2}{4}$ is strictly concave, $\beta\mapsto \underset{\epsilon <\tilde{\theta}<C }{\operatorname{\min}} \overline{D}(\tilde{\theta},\beta)$ is strictly concave in $\beta$. It thus has a unique maximum  as it satisfies $\lim_{\beta\to\infty} \underset{\epsilon <\tilde{\theta}<C }{\operatorname{\min}} \overline{D}(\tilde{\theta},\beta)=-\infty$.    
Denote by $\beta_\star$ such a maximum. Let us prove that there exists a unique $\theta_\star$ that minimizes function $h$ defined as $h:\tilde{\theta}\mapsto \overline{D}(\tilde{\theta},\beta_\star)$. 
The proof of this result will be carried out into the following steps:
\begin{enumerate}
	\item First, we prove that the minimum should be in the interior domain of $(\epsilon,C)$ for $C$ sufficiently large and $\epsilon$ sufficiently small. 
	\item Next, we establish that $Y(\tilde{\theta},\beta_\star)$ satisfies:
		\begin{equation}
		\left|\tilde{\theta}\frac{\partial^3 Y(\tilde{\theta},\beta_\star)}{\partial \tilde{\theta}^3}\right|< 3\left|\frac{\partial^2 Y(\tilde{\theta},\beta_\star)}{\partial \tilde{\theta}^2}\right|.
			\label{eq:req}
		\end{equation}
	\item Starting from the observation that $\tilde{\theta}_\star$ is in the interior domain of the optimization set and based on the previously established results, we prove that $h$ admits a unique minimum.    
\end{enumerate}  
We start by establishing the first statement. It is obvious that the optimum could not be reached when $\tilde{\theta}$ is in the vicinity of zero since  $\lim_{\tilde{\theta}\to 0^+} \overline{D}(\tilde{\theta},\beta_\star)=\infty$. Similarly, to prove that the minimum is not reached when $\tilde{\theta}$ grows to infinity, it suffices to check that $\lim_{\tilde{\theta}\to\infty} \overline{D}(\tilde{\theta},\beta_\star)=\infty$. Simple calculations lead to:
\begin{align*}
	Y(\tilde{\theta},\beta_\star)\underset{\tilde{\theta}\to\infty}{\sim} -\frac{\beta \xi^2(M-1)^2}{M\mathcal{E}}\tilde{\theta}&=-\frac{3\beta_\star \xi^2(M-1)^2}{M(M^2-1)}\tilde{\theta}\\
	&=-\frac{3\beta_\star \rho_{ {d}}\sigma_{ {\hat{H}}^2}(M-1)}{M(M+1)}\tilde{\theta}.
\end{align*}
Using this approximation, we thus have:
$$
\overline{D}(\tilde{\theta},\beta_\star)\underset{\tilde{\theta}\to\infty}{\sim}\tilde{\theta}\beta_\star\left( -\frac{3 \rho_{ {d}}\sigma_{ {\hat{H}}}^2(M-1)}{M(M+1)} + \frac{1}{2}+\frac{\rho_{ {d}}}{2} \right).
$$
It is easy to check that for $M\geq 2$, $\frac{3(M-1)}{M(M+1)}\leq \frac{1}{2}$. As $\sigma_{ {\hat{H}}}^2<1$,  we thus have $\lim_{\theta\to\infty} \overline{D}(\tilde{\theta},\beta_\star)=\infty$. 

To prove \eqref{eq:req}, we need to compute the first three derivatives of the function $\tilde{\theta}\mapsto Y(\tilde{\theta},\beta_\star)$. After simple calculations, we can establish that:
\begin{align}
	\frac{\partial Y(\tilde{\theta},\beta_\star)}{\partial \tilde{\theta}}=\frac{1}{M}& \sum_{\substack{i=\pm 1, \cdots,\\ \pm(M-1)}}\Big[\int_{\ell_i}^{\mu_i} 2\beta_\star^3\xi^2\left(\frac{\frac{\xi^2i\tilde{\theta}}{\sqrt{\mathcal{E}}}+\xi h}{2\xi^2\beta_\star\tilde{\theta}+4\lambda\rho_{ {d}}}\right)^2 p(h)dh\nonumber\\
	&-\frac{2\beta_\star^2\xi^2i}{\sqrt{\mathcal{E}}}\int_{\ell_i}^{\mu_i}\left(\frac{\frac{\xi^2i\tilde{\theta}}{\sqrt{\mathcal{E}}}+\xi h}{2\xi^2\beta_\star\tilde{\theta}+4\lambda\rho_{ {d}}}\right)p(h)dh +\int_{-\infty}^{\ell_i} t^2\beta_\star \frac{\xi^2}{2}p(h)dh+\int_{\mu_i}^{\infty} t^2\beta_\star \frac{\xi^2}{2}p(h)dh\Big].
\end{align}
Based on this expression, we compute the second and third derivatives of $Y(\tilde{\theta},\beta_\star)$ as:
\begin{align}
	&\frac{\partial^2 Y(\tilde{\theta},\beta_\star)}{\partial \tilde{\theta}^2}=-2\beta_\star^2\frac{1}{M}\!\!\!\sum_{\substack{i=\pm 1, \cdots,\\ ,\pm(M-1)}}\!\!\!\int_{\ell_i}^{\mu_i}\frac{(4\lambda\rho_{ {d}}\frac{\xi^2i}{\sqrt{\mathcal{E}}}-2\xi^3\beta_\star h)^2 }{(2\xi^2\beta_\star\tilde{\theta}+4\lambda\rho_{ {d}})^3}p(h)dh,\label{sec}
\end{align}
\begin{align}
	&\frac{\partial^3 Y(\tilde{\theta},\beta_\star)}{\partial \tilde{\theta}^3}=6\beta_\star^2\frac{1}{M}\!\!\!\sum_{\substack{i=\pm 1,, \cdots,\\ \pm(M-1)}}\!\!\!\int_{\ell_i}^{\mu_i}\frac{2\xi^2\beta_\star(4\lambda\rho_{ {d}}\frac{\xi^2i}{\sqrt{\mathcal{E}}}-2\xi^3\beta_\star h)^2 }{(2\xi^2\beta_\star\tilde{\theta}+4\lambda\rho_{ {d}})^4}p(h)dh.\label{third}
\end{align}
Leveraging \eqref{sec} and \eqref{third}, it is easy to check that:
\begin{align*}
	&\tilde{\theta}\frac{\partial^3 Y(\tilde{\theta},\beta_\star)}{\partial \tilde{\theta}^3}=-3\frac{\partial^2 Y(\tilde{\theta},\beta_\star)}{\partial \tilde{\theta}^2} -24\beta_\star^2\frac{1}{M}\!\!\!\sum_{\substack{i=\pm 1,\\ \cdots,\pm(M-1)}}\!\!\!\int_{\ell_i}^{\mu_i}\frac{4\lambda\rho_{ {d}}(4\lambda\rho_{ {d}}\frac{\xi^2i}{\sqrt{\mathcal{E}}}-2\xi^3\beta_\star h)^2 }{(2\xi^2\beta_\star\tilde{\theta}+4\lambda\rho_{ {d}})^4}p(h)dh,
\end{align*}
from which we deduce that 
$$
\left|\tilde{\theta} \frac{\partial^3 Y(\tilde{\theta},\beta_\star)}{\partial \tilde{\theta}^3}\right| < 3\left|\frac{\partial^2 Y(\tilde{\theta},\beta_\star)}{\partial \tilde{\theta}^2}\right|.
$$
With this result at hand, we are now ready to prove that function $h$ admits a unique minimum. We already proved that any minimum should lie in the interior domain of $(\epsilon,C)$. Assume that there exists two minimizers of $h$ which we denote by $\tilde{\theta}_{\star,1}$ and $\tilde{\theta}_{\star,2}$ such that $\tilde{\theta}_{\star,1}<\tilde{\theta}_{\star,2}$. The first order and second order conditions imply that:
\begin{equation}
	\frac{{\rm{d}} h}{{\rm{d}} \tilde{\theta}}\Big|_{\tilde{\theta}=\tilde{\theta}_{\star,i}}=0 \ \ \textnormal{and} \ \ \  \frac{{\rm{d}}^2 h}{{\rm{d}}\tilde{\theta}^2}\Big|_{\tilde{\theta}=\tilde{\theta}_{\star,i}}\geq 0, \ \ i=1,2.  \label{eq:cond1}
\end{equation}
Hence, there exists $\tilde{\theta}_3\in(\tilde{\theta}_{\star,1},\tilde{\theta}_{\star,2})$ such that \begin{equation}\frac{{\rm{d}}^2 h}{{\rm{d}}\tilde{\theta}^2}\Big|_{\tilde{\theta}=\tilde{\theta}_3}=0.\label{eq:cond2}\end{equation} We will prove that this will lead to contradiction unless $\tilde{\theta}_{\star,1}=\tilde{\theta}_{\star,2}$. To this end, first we notice that:
$$
\frac{{\rm{d}}^2 h}{{\rm{d}}\tilde{\theta}^2}=-\beta_\star\delta \tilde{\theta}^{-3}+\frac{\partial^2Y(\tilde{\theta},\beta_\star)}{\partial \tilde{\theta}^2}.
$$
Hence, from \eqref{eq:cond1} and \eqref{eq:cond2} we obtain  the following relations 
\begin{align}
	&\beta_\star\tilde{\delta}+\tilde{\theta}_{\star,1}^3\frac{\partial^2Y(\tilde{\theta},\beta_\star)}{\partial \tilde{\theta}^2}\Big|_{\tilde{\theta}=\tilde{\theta}_{\star,1}}\geq 0,\label{eq:1}\\
	&\beta_\star\tilde{\delta}+\tilde{\theta}_{\star,2}^3\frac{\partial^2Y(\tilde{\theta},\beta_\star)}{\partial \tilde{\theta}^2}\Big|_{\tilde{\theta}=\tilde{\theta}_{\star,2}}\geq 0,\label{eq:2}\\
	& \beta_\star\tilde{\delta}+\tilde{\theta}_{\star,3}^3\frac{\partial^2Y(\tilde{\theta},\beta_\star)}{\partial \tilde{\theta}^2}\Big|_{\tilde{\theta}=\tilde{\theta}_{\star,3}}=0.  \label{eq:3}
\end{align}
Consider function $k:\tilde{\theta}\mapsto \beta\tilde{\delta}+ \tilde{\theta}^3 \frac{\partial^2Y(\tilde{\theta},\beta_\star)}{\partial \tilde{\theta}^2}$. The derivative of $k$ with respect to $\tilde{\theta}$ is given by:
$$
k'(\tilde{\theta})= \tilde{\theta}^{2}\left(3\frac{\partial^2 Y(\tilde{\theta},\beta_\star)}{\partial \tilde{\theta}^2}+\tilde{\theta}\frac{\partial^3 Y(\tilde{\theta},\beta_\star)}{\partial \tilde{\theta}^3}\right).
$$
From \eqref{eq:req}, $k'(\tilde{\theta})< 0$ and as such $k$ is  decreasing. Hence, the relations \eqref{eq:1}, \eqref{eq:2} and \eqref{eq:3} could not simultaneously hold. Hence, $\tilde{\theta}_{\star,1}=\tilde{\theta}_{\star,2}$, and as a consequence $h$ admits a unique minimizer which we denote by $\tilde{\theta}_\star$. 
Now, since for any $\beta>0$, $\tilde{\theta}\mapsto\overline{D}(\tilde{\theta},\beta)$ goes to infinity when $\tilde{\theta}$ approaches zero or grows to infinity, we thus have:
\begin{equation}
\min_{\epsilon\leq\tilde{\theta}\leq C}\overline{D}(\tilde{\theta},\beta) = \inf_{\tilde{\theta}\geq 0}\overline{D}(\tilde{\theta},\beta), \ \ \forall \ \  \beta>0.
	\label{eq:rel}
\end{equation}
The above relation holds for all $\beta>0$. We already proved in Section B that the optimal solution in $\beta$ is almost surely away from zero. Thus, 
\begin{equation}
\sup_{\beta>0} \min_{\epsilon\leq\tilde{\theta}\leq C}\overline{D}(\tilde{\theta},\beta) = \sup_{\beta>0} \inf_{\tilde{\theta}>0}\overline{D}(\tilde{\theta},\beta).
\end{equation}
\subsection{Asymptotic behavior of metrics depending on the solution of the (PO) }
So far, we proved that the PO cost converges to the asymptotic cost of the (AO). We prove now that the uniqueness of the minimizer $\tilde{\theta}_\star$ allows us to carry over this convergence to metrics depending on the solution of the (PO). The recipe is as follows. Let $\eta>0$ and define $$\mathcal{S}_\eta=\left\{{\bf v} \ | \ {\bf v}\in\mathcal{D},  \  \textnormal{and}  \ \frac{1}{K}\|{\bf v}\|^2\in \left(\frac{\delta}{\tilde{\theta}_\star^2}-1-\eta,\frac{\delta}{\tilde{\theta_\star}^2}-1+\eta\right) \right\}.$$ Consider
the ``perturbed'' version of the (AO) in \eqref{eq:AO1} as follows:
\begin{align}
	{\phi}_\eta:=	\frac{1}{\sqrt{K}} \min_{{\bf v}\notin \mathcal{S}_\eta} &\max_{{\bf u}}\frac{1}{K} \|{\bf v}\|{\bf g}^{T}{\bf u}-\frac{1}{K}\|{\bf u}\|{\bf q}^{T}{\bf v} -\frac{1}{\sqrt{K}}{\bf u}^{T}{\bf z} -\frac{1}{4\sqrt{K}}\|{\bf u}\|^2+\frac{1}{\sqrt{K}}\lambda \|{\bf Cv}\|^2.
\end{align}
Following the same analysis carried out previously, we lower bound $\phi_{\eta}$ by
$$
\phi_{\eta}\geq \max_{\beta\geq 0} \min_{\substack{\epsilon\leq \tilde{\theta}\leq C\\ \left|\frac{\|{\bf g}\|^2}{\tilde{\theta}^2K}-\frac{\delta}{\tilde{\theta}_\star^2}\right|>\eta}}\widetilde{D}(\tilde{\theta},\beta,{\bf g},{\bf q}^{\bf 1}).
$$
It is easy to note that if $\tilde{\theta}$ in $\left(\epsilon, C\right)$ is such that $\left|\frac{\|{\bf g}\|^2}{K\tilde{\theta}^2}-\frac{{\delta}}{\tilde{\theta}_\star^2}\right|>\eta$ then 
$$
\left|\tilde{\theta}-\tilde{\theta}_\star\frac{\|{\bf g}\|}{\sqrt{K}\sqrt{\delta}}\right| \geq \frac{\eta \tilde{\theta}^2\tilde{\theta}_\star^2}{\sqrt{\delta}\left(\tilde{\theta}_\star\frac{\|{\bf g}\|}{\sqrt{K}}+\tilde{\theta}\sqrt{\delta}\right)}\geq \frac{\eta\epsilon^2\tilde{\theta}_\star^2}{\sqrt{\delta}\left(C\sqrt{\delta}+\tilde{\theta}_\star\frac{\|{\bf g}\|}{\sqrt{K}}\right)},
$$
or also
$$
\left|\tilde{\theta}_\star-\tilde{\theta}\right|>  \frac{\eta \epsilon \tilde{\theta}_\star}{\frac{1}{\epsilon}+\frac{\sqrt{\delta}}{\tilde{\theta}_\star}}-\tilde{\theta} \left|1-\frac{\|{\bf g}\|}{\sqrt{K}\sqrt{\delta}}\right|.
$$
Since $\left|1-\frac{\|{\bf g}\|}{\sqrt{K}\sqrt{\delta}}\right|$ converges to zero almost surely, there exists $\tilde{\eta}>0$ such that with probability approaching $1$,
$$
\epsilon\leq \tilde{\theta}\leq C \  \text{and} \  \left|\frac{\|{\bf g}\|^2}{K\tilde{\theta}^2}-\frac{\delta}{\tilde{\theta}_\star^2}\right|\geq \eta \  \Longrightarrow \  \left|\tilde{\theta}-\tilde{\theta}_\star\right|\geq \tilde{\eta}.
$$
Hence, with probability approaching $1$,
$$
\phi_{\eta}\geq \sup_{\beta\geq 0} \min_{\substack{\epsilon\leq\tilde{\theta}\leq C\\ \left|\tilde{\theta}-\tilde{\theta}_\star\right|\geq \tilde{\eta}}} \widetilde{D}(\tilde{\theta},\beta,{\bf g},{\bf q}^{\bf 1}).
$$
Following the same asymptotic analysis as in Section C, we can prove similarly that
$$
\sup_{\beta\geq 0} \min_{\substack{\epsilon\leq\tilde{\theta}\leq C\\ \left|\tilde{\theta}-\tilde{\theta}_\star\right|\geq \tilde{\eta}}} \widetilde{D}(\tilde{\theta},\beta,{\bf g},{\bf q}^{\bf 1})\overset{P}{\longrightarrow}\overline{\phi}_\eta:= \sup_{\beta\geq 0} \inf_{\substack{\tilde{\theta}\geq 0\\  \left|\tilde{\theta}-\tilde{\theta}_\star\right|\geq \tilde{\eta}}} \overline{D}(\tilde{\theta},\beta).
$$
Clearly,
$$
\overline{\phi}_\eta\geq \inf_{\substack{\tilde{\theta}\geq 0 \\ \left|\tilde{\theta}-\tilde{\theta}_\star\right|\geq \tilde{\eta}}} \overline{D}(\tilde{\theta},\beta_\star).
$$
As $\tilde{\theta}_\star$ is the unique minimizer of $\inf_{\tilde{\theta}\geq 0}\overline{D}(\tilde{\theta},\beta_\star)$,
$$
\overline{\phi}_{\eta}> \overline{\phi}.
$$
Based on the CGMT in Theorem \ref{eq:th_CGMT} and recalling Remark \ref{remark_imp}, we thus have:
\begin{equation}
	\lim_{K\to \infty }\mathbb{P}\left[{\bf v}_{\rm PO}\in\mathcal{S}_\eta\right]=1,
	\label{eq:v}
\end{equation}
where ${\bf v}_{\rm PO}$ is the solution in ${\bf v}$ of the (PO) in \eqref{POprob}, 
or equivalently:
$$
\frac{1}{K} \|{\bf v}_{\rm PO} \|^2 \overset{P}{\longrightarrow} \frac{\delta}{\tilde{\theta}_\star^2}-1 = \delta\theta_\star^2-1,
$$
where $\theta_\star=\frac{1}{\tilde{\theta}_\star}$. 
\subsection{Convergence of Lipschitz functions of the estimated vector $\widehat{\bf x}$}
The objective here is to study the asymptotic behavior of Lipschitz functions of the solution to the PO which we denote by $\hat{\bf x}$. As will be seen in the next section, such a result is fundamental for the asymptotic analysis of the symbol error rate and can be of independent interest to analyze any other performance metric.   
Let $\beta_\star$ and $\tilde{\theta}_\star$ be the unique solutions of the optimization problem $\sup_{\beta>0}\inf_{\tilde{\theta}>0}\overline{D}(\tilde{\theta},\beta)$. Recall that ${\bf x}_0$ represents the transmitted vector, whose elements are drawn with equal probability from the $M$-PAM constellation.  
\begin{lemma}
For $j=1,\cdots, K$, and for fixed $\tilde{\theta}$ and $\beta$, 
	define $\kappa_j(.;\tilde{\theta},\beta,x_{0,j}):\mathbb{R}\to\mathbb{R}$ as:
$$
	\kappa_j(x;\tilde{\theta},\beta,x_{0,j})=\left\{\begin{array}{l}
	-t, \ \ \textnormal{if} \ \ x<x_{0}^{-}(\tilde{\theta},\beta,x_{0,j}) \\
	t,  \ \ \textnormal{if} \ \ x>x_{0}^{+}(\tilde{\theta},\beta,x_{0,j})\\
	\frac{\beta\left(\rho_{ {d}}\sigma_{\hat{H}}^2x_{0,j}\tilde{\theta}+\sqrt{\rho_{ {d}}}\sigma_{\hat{ {H}}}x\right)}{\rho_{ {d}}\sigma_{\hat{ {H}}}^2\beta\tilde{\theta}+2\lambda\rho_{ {d}}}, \ \ \textnormal{otherwise},
\end{array}
	\right.
$$
	where $x_{0}^{-}(\tilde{\theta},\beta,x_{0,j})$ and $x_{0}^{+}(\tilde{\theta},\beta,x_{0,j})$ are given by \eqref{eq:x0j-} and \eqref{eq:x0j+}. 
Let $\widehat{\bf x}$ be the solution of the (PO). Then for all for 
all Lipschitz functions $\psi:\mathbb{R}\to\mathbb{R}$ with Lipschitz constant $L$, it holds:
\begin{equation}
	\frac{1}{K}\sum_{i=1}^K \psi(\widehat{x}_i)-\frac{1}{K}\sum_{j=1}^{K}\mathbb{E}\left[\psi(\kappa_j(q_j;\tilde{\theta}_\star,\beta_\star,x_{0,j}))\right] \overset{P}{\longrightarrow} 0,
	\label{eq:conv}
\end{equation}
where $q_1,\cdots,q_K$ are independent and identically distributed standard Gaussian random variables and $(\beta_\star,\tilde{\theta}_\star)$ is the unique solution to the optimization problem $\sup_{\beta\geq 0} \inf_{\tilde{\theta}\geq 0}\overline{D}(\tilde{\theta},\beta)$.
	\label{lem:conv}
\end{lemma}
\begin{proof}
To avoid heavy notations, we will remove $x_{0,j}$ from the notations of $\kappa_j(\tilde{\theta},\beta,x_{0,j})$ and that of $x_{0}^{-}(\tilde{\theta},\beta,x_{0,j})$ and $x_{0}^{+}(\tilde{\theta},\beta,x_{0,j})$ as it will not play any role in the proof. 
To prove \eqref{eq:conv}, we consider the set:
$$
	\mathcal{S}_{\epsilon}=\left\{{\bf x} \ \ | \ \  \left|\frac{1}{K}\sum_{i=1}^K \psi(x_i)-\frac{1}{K}\sum_{j=1}^{K}\mathbb{E}\left[\psi(\kappa_j(q_j;\tilde{\theta}_\star,\beta_\star))\right]\right|<\epsilon\right\}.
$$
Then, in view of the CGMT, for \eqref{eq:conv} to hold true, it suffices to show that with probability approaching 1, 
\begin{equation}
	\max_{\beta>0} \min_{\substack{-t<x_j<t\\ j=1,\cdots,K\\ {\bf x}\notin \mathcal{S}_\epsilon }}\hat{\mathcal{H}}(\beta,{\bf x})\geq \overline{\phi},\label{eq:suf}
\end{equation}
where we recall that  $\hat{\mathcal{H}}(\beta,{\bf x})$ is the objective of the optimization problem in \eqref{eq:pr1}. 
Noting that
$$
\max_{\beta>0} \min_{\substack{-t<x_j<t\\ j=1,\cdots,K\\ {\bf x}\notin \mathcal{S}_\epsilon }}\hat{\mathcal{H}}(\beta,{\bf x})\geq  \min_{\substack{-t<x_j<t\\ j=1,\cdots,K\\ {\bf x}\notin \mathcal{S}_\epsilon }}\hat{\mathcal{H}}(\beta_\star,{\bf x}),
$$
the proof of \eqref{eq:suf} boils down to proving that
\begin{equation}
\min_{\substack{-t<x_j<t\\ j=1,\cdots,K \\ {\bf x}\notin \mathcal{S}_\epsilon }}\hat{\mathcal{H}}(\beta_\star,{\bf x})>\overline{\phi}.
	\label{eq:pbt}
\end{equation}
A key step towards showing \eqref{eq:pbt} is to analyze the asymptotic behavior of the following optimization problem:
\begin{equation}
\min_{\substack{-t<x_j<t\\ j=1,\cdots,K  }}\hat{\mathcal{H}}(\beta_\star,{\bf x}).
	\label{eq:H}
\end{equation}
Particularly, the following statements will be shown in the sequel:
\begin{itemize}
	\item[(i)] The following convergence holds true:
		\begin{equation}
			\min_{\substack{-t<x_j<t\\ j=1,\cdots,K  }}\hat{\mathcal{H}}(\beta_\star,{\bf x}) \overset{P}{\longrightarrow}\overline{\phi}=\sup_{\beta\geq 0}\min_{\tilde{\theta}\geq 0}\overline{D}(\tilde{\theta},\beta) \label{eq:conv_p}
		\end{equation}
	\item[(ii)] Let $\hat{\tilde{{\bf x}}}=\left[\hat{\tilde{x}}_1,\cdots,\hat{\tilde{x}}_K\right]^{T}$ be the solution to the optimization problem in \eqref{eq:H}. Then, it holds  with probability approaching 1, 
		\begin{equation}
			\forall \  {\bf x}\in [-t,t]^{K}\backslash \mathcal{S}_\epsilon, \ \ 	\frac{1}{\sqrt{K}} \|{\bf x}-\hat{\tilde{{\bf x}}}\| \geq \frac{\epsilon}{2L}. \label{eq:21}
		\ \ 	\end{equation}
	\item[(iii)]    For any  ${\bf x}=\left[x_1,\cdots,x_K\right]^{T}\in \left[-t,t\right]^{K}$, there exists a constant $\overline{C}$ such that: 
		\begin{equation}
			\hat{\mathcal{H}}(\beta_\star,{\bf x}) \geq 	\hat{\mathcal{H}}(\beta_\star,\hat{\tilde{{\bf x}}})  +\frac{\overline{C}}{2K}\|{\bf x}-\hat{\tilde{{\bf x}}}\|^2. \label{eq:22}
		\end{equation}
\end{itemize}
Prior to proving the above statements, let us see how they lead to the desired inequality \eqref{eq:pbt}. Putting together \eqref{eq:conv_p} and \eqref{eq:22} shows that for any $\ell>0$, with probability approaching $1$, 
\begin{equation}
	\hat{\mathcal{H}}(\beta_\star,{\bf x})\geq \overline{\phi}-\ell+\frac{\overline{C}}{2K}\|{\bf x}-\hat{\tilde{{\bf x}}}\|^2.
	\label{eq:ll}
\end{equation}
From \eqref{eq:21}, we have for any ${\bf x}\in [-t,t]^{K}\backslash \mathcal{S}_\epsilon$, 
\begin{equation}
\frac{1}{K}\|{\bf x}-\hat{\tilde{{\bf x}}}\|^2\geq \frac{\epsilon^2}{4L}.
	\label{eq:ineq}
\end{equation}
Now, setting $\ell=\frac{\overline{C}\epsilon^2}{16L}$ in \eqref{eq:ll}, and using \eqref{eq:ineq} we get:
$$
\hat{\mathcal{H}}(\beta_\star,{\bf x})\geq \overline{\phi}+\frac{\overline{C}\epsilon^2}{16L}.
$$
The above inequality holds for any ${\bf x}\in [-t,t]^{K}\backslash \mathcal{S}_\epsilon$ hence,
$$
\min_{\substack{-t<x_j<t\\ j=1,\cdots,K\\ {\bf x}\notin\mathcal{S}_\epsilon}} \hat{\mathcal{H}}(\beta_\star,{\bf x})\geq \overline{\phi}+\frac{\overline{C}\epsilon^2}{16L}>\overline{\phi},
$$
thus proving \eqref{eq:pbt}.

\noindent\underline{Proof of \eqref{eq:conv_p}}
Following the same calculations used in the analysis of the (AO), we can prove that for sufficiently small $\epsilon$ and large constant $C$, with probability approaching $1$,
\begin{align}
	\!\!\min_{\substack{-t\leq x_j\leq t\\ j=1,\cdots,K }}\hat{\mathcal{H}}(\beta_\star,{\bf x})&\!\!=\!\!
	\min_{\epsilon\leq \tilde{\theta}\leq C} \widetilde{D}(\tilde{\theta},\beta_\star,{\bf g},{\bf q}^{\bf 1}).
\end{align}
Based on similar asymptotic analysis to that carried out in Section B, we can prove that $\tilde{\theta}\mapsto \widetilde{D}(\tilde{\theta},\beta_\star,{\bf g},{\bf q}^{\bf 1})$ converges uniformly to $\tilde{\theta}\mapsto \overline{D}(\tilde{\theta},\beta_\star)$ and hence,
\begin{equation}
	\min_{\substack{-t\leq x_j\leq t\\ j=1,\cdots,K }}\hat{\mathcal{H}}(\beta_\star,{\bf x})\overset{P}{\longrightarrow} \overline{\phi}=\min_{\epsilon\leq \tilde{\theta}\leq C}\overline{D}(\tilde{\theta},\beta_\star). \label{eq:conv_u}
\end{equation}
\noindent\underline{Proof of \eqref{eq:21}}
Let $\hat{\tilde{\theta}} \in\left\{\argmin_{\epsilon\leq\tilde{\theta}\leq C} \widetilde{D}(\tilde{\theta},\beta_\star,{\bf g},{\bf q}^{\bf 1})\right\}$. In view of \eqref{AO_Sol_box}, for $j=1,\cdots,K$, $\left\{\hat{\tilde{x}}_j=\kappa_j(q_j^1,\hat{\tilde{\theta}},\beta_\star)\right\}$ is a solution to \eqref{eq:H}. On the other hand, as the minimum of $\tilde{\theta}\mapsto  \overline{D}(\tilde{\theta},\beta_\star)$ is unique, we conclude that $\hat{\tilde{\theta}}$ converges in probability to $\tilde{\theta}_\star$. 
From this convergence, we argue that 
\begin{equation}
	\max_{1\leq j\leq K}\left|\kappa_j(q_j^{1},\hat{\tilde{\theta}},\beta_\star)-\kappa_j(q_j^{1},\tilde{\theta}_\star,\beta_\star)\right|\overset{P}{\longrightarrow} 0. \label{eq:kappa_jconv}
\end{equation}
Prior to showing \eqref{eq:kappa_jconv}, let us explain how it leads to \eqref{eq:21}. Indeed from the Lipschitz assumption it holds that:
\begin{equation}
	\max_{1\leq j\leq K} \left|\psi(\kappa_j(q_j^{1},\hat{\tilde{\theta}},\beta_\star))-\psi(\kappa_j(q_j^{1},\tilde{\theta}_\star,\beta_\star))\right|\overset{P}{\longrightarrow} 0.
	\label{eq:comb}
\end{equation}
Moreover, from the law of large numbers, 
$$
\frac{1}{K}\sum_{j=1}^{K}\psi(\kappa_j(q_j^{1},\tilde{\theta}_\star,\beta_\star))-\mathbb{E}\frac{1}{K}\sum_{j=1}^{K}\psi(\kappa_j(q_j^{1},\tilde{\theta}_\star,\beta_\star))\overset{P}{\longrightarrow}0,
$$
which combined with \eqref{eq:comb} yields:
$$
\frac{1}{K}\sum_{j=1}^{K}\psi(\kappa_j(q_j^{1},\hat{\tilde{\theta}},\beta_\star))-\mathbb{E}\frac{1}{K}\sum_{j=1}^{K}\psi(\kappa_j(q_j^{1},\tilde{\theta}_\star,\beta_\star))\overset{P}{\longrightarrow}0.
$$
Hence, for $\epsilon>0$, with probability approaching one,
$$
\left|\frac{1}{K}\sum_{j=1}^{K}\psi(\kappa_j(q_j^{1},\hat{\tilde{\theta}},\beta_\star))-\mathbb{E}\frac{1}{K}\sum_{j=1}^{K}\psi(\kappa_j(q_j^{1},\tilde{\theta}_\star,\beta_\star))\right|\leq \epsilon.
$$
By definition of the set $\mathcal{S}_\epsilon$ and the triangle inequality, it holds with probability approaching 1 that for all ${\bf x}\in [-t,t]^{K}\backslash \mathcal{S}_\epsilon$
$$
\left|\frac{1}{K}\sum_{j=1}^{K}\psi(x_i)-\frac{1}{K}\sum_{j=1}^{K}\psi(\hat{\tilde{x}}_j))\right|\geq \frac{\epsilon}{2}.
$$
Then, based on the Lipshitz property of $\psi$, it holds that:
$$
\left\|{\bf x}-\hat{\tilde{\bf x}}\right\|\geq \frac{\epsilon}{2L},
$$
which shows \eqref{eq:21}.

Now, to prove \eqref{eq:kappa_jconv}, note that since $\hat{\tilde{\theta}}\overset{P}{\longrightarrow}\tilde{\theta}_\star$,  for any $\eta>0$, with probability approaching one,
\begin{equation}
-\eta\leq \hat{\tilde{\theta}}-\tilde{\theta}_\star\leq \eta,
	\label{eq:thet}
\end{equation}
from which we deduce that:
$$
\max_{1\leq j\leq K} \left|x_{0}^{+}(\hat{\tilde{\theta}},\beta_\star)-x_{0}^{+}(\tilde{\theta}_\star,\beta_\star)\right|\leq\overline{C}\eta,
$$
where $\overline{C}=\sqrt{\rho_{ {d}}}\sigma_{\hat{ {H}}}(\max_{1\leq j\leq K}|x_{0,j}|+t)$
and similarly,
$$
\left|x_{0}^{-}(\hat{\tilde{\theta}},\beta_\star)-x_{0}^{-}(\tilde{\theta}_\star,\beta_\star)\right|\leq \overline{C}\eta.
$$
Using the fact that $x_{0}^{-}(\tilde{\theta}_\star,\beta_\star)<x_{0}^{+}(\tilde{\theta}_\star,\beta_\star)$, and choosing $\eta$ sufficiently small, we conclude that if for some $j=1,\cdots,K$, $q_j^{1}\leq x_{0}^{-}(\tilde{\theta}_\star,\beta_\star)$, then $q_j^{1}\leq x_{0}^{+}(\tilde{\hat{\theta}},\beta_\star)$, and similarly, if  $q_j^{1}\geq x_{0}^{+}(\tilde{\theta}_\star,\beta_\star)$, then $q_j^{1}\geq x_{0}^{-}(\tilde{\hat{\theta}},\beta_\star)$. As a consequence, 
\begin{align}
	&\max_{1\leq j\leq K}\left|\kappa_j(q_j^{1},\hat{\tilde{\theta}},\beta_\star,)-\kappa_j(q_j^{1},\tilde{\theta}_\star,\beta_\star)\right|\nonumber\\
	&\leq \left|t+\frac{\beta_\star(\rho_{ {d}}\sigma_{\hat{ {H}}}^2x_{0,j}\hat{\tilde{\theta}}+\sqrt{\rho_{ {d}}}\sigma_{\hat{ {H}}}q_j^1)}{\rho_{ {d}}\sigma_{\hat{ {H}}}^2\beta_\star\hat{\tilde{\theta}}+2\lambda\rho_{ {d}}}\right| {\bf 1}_{\{x_{0}^{-}(\hat{\tilde{\theta}},\beta_\star)\leq q_j^{1}\leq x_{0}^{-}(\theta_\star,\beta_\star)\}}\nonumber\\
	&+\left|t+\frac{\beta_\star(\rho_{ {d}}\sigma_{\hat{ {H}}}^2x_{0,j}{\tilde{\theta}_\star}+\sqrt{\rho_{ {d}}}\sigma_{\hat{ {H}}}q_j^1)}{\rho_{ {d}}\sigma_{\hat{ {H}}}^2\beta_\star{\tilde{\theta}}_\star+2\lambda\rho_{ {d}}}\right| {\bf 1}_{\{x_{0}^{-}({\tilde{\theta}}_\star,\beta_\star)\leq q_j^{1}\leq x_{0}^{-}(\hat{\tilde{\theta}},\beta_\star)\}}\nonumber\\
	&+\left|t-\frac{\beta_\star(\rho_{ {d}}\sigma_{\hat{ {H}}}^2x_{0,j}\hat{\tilde{\theta}}+\sqrt{\rho_{ {d}}}\sigma_{\hat{ {H}}}q_j^1)}{\rho_{ {d}}\sigma_{\hat{ {H}}}^2\beta_\star\hat{\tilde{\theta}}+2\lambda\rho_{ {d}}}\right| {\bf 1}_{\{x_{0}^{+}(\tilde{\theta}_\star,\beta_\star)\leq q_j^{1}\leq x_{0}^{+}(\hat{\tilde{\theta}},\beta_\star)\}}\nonumber\\
	&+\left|t-\frac{\beta_\star(\rho_{ {d}}\sigma_{\hat{ {H}}}^2x_{0,j}\tilde{\theta}_\star+\sqrt{\rho_{ {d}}}\sigma_{\hat{ {H}}}q_j^1)}{\rho_{ {d}}\sigma_{\hat{ {H}}}^2\beta_\star\tilde{\theta}_\star+2\lambda\rho_{ {d}}}\right| {\bf 1}_{\{x_{0}^{+}(\hat{\tilde{\theta}},\beta_\star)\leq q_j^{1}\leq x_{0}^{+}(\tilde{\theta}_\star,\beta_\star)\}}\nonumber\\
	&+\left|\frac{\beta_\star(\rho_{ {d}}\sigma_{\hat{ {H}}}^2x_{0,j}{\tilde{\theta}_\star}+\sqrt{\rho_{ {d}}}\sigma_{\hat{ {H}}}q_j^1)}{\rho_{ {d}}\sigma_{\hat{ {H}}}^2\beta_\star{\tilde{\theta}}_\star+2\lambda\rho_{ {d}}}-\frac{\beta_\star(\rho_{ {d}}\sigma_{\hat{ {H}}}^2x_{0,j}\hat{\tilde{\theta}}+\sqrt{\rho_{ {d}}}\sigma_{\hat{ {H}}}q_j^1)}{\rho_{ {d}}\sigma_{\hat{ {H}}}^2\beta_\star\hat{\tilde{\theta}}+2\lambda\rho_{ {d}}}\right|
	 {\bf 1}_{\{\max(x_{0}^{-}(\tilde{\theta}_\star,\beta_\star),x_{0}^{-}(\hat{\tilde{\theta}},\beta_\star))\leq q_j^1 \leq \min\left(x_{0}^{+}(\hat{\tilde{\theta}},\beta_\star),x_{0}^{+}(\tilde{\theta}_\star,\beta_\star)\right)\}}.\label{eq:formula}
\end{align}
Each term of the right-hand side of \eqref{eq:formula} can be bounded by a linear function of $\eta$ using \eqref{eq:thet}, thereby proving \eqref{eq:kappa_jconv}. 

\noindent\underline{Proof of \eqref{eq:22}}
It can be checked that ${\bf x}\mapsto \hat{\mathcal{H}}(\beta_\star,{\bf x})$ is strongly convex, and  its Hessian satisfies $\nabla^2 \hat{\mathcal{H}}(\beta_\star,{\bf x}) \succeq \frac{2\lambda}{K}\rho_{ {d}}{\bf I}_K$. Hence, for any ${\bf x}\in[-t,t]^{K}$,
$$
\hat{\mathcal{H}}(\beta_\star,{\bf x})\geq \hat{\mathcal{H}}(\beta_\star,\hat{\tilde{{\bf x}}}) + \frac{2\lambda}{K}\rho_{ {d}} \left\|{\bf x}-\hat{\tilde{{\bf x}}}\right\|^2.
$$
\end{proof}
\subsection{From Lipschitz to the indicator function of solutions to the (PO)}
\begin{lemma}
	Let $\hat{\bf x}$ be the solution to the (PO). Let $c\in\mathbb{R}$ such that $c\notin \left\{-t,t\right\}$. Then, 
$$
	\frac{1}{K}\sum_{j=1}^K {\bf 1}_{\{{\hat{x}_j}\leq c\}}-\frac{1}{K}\sum_{j=1}^{K}\mathbb{P}\left[{\kappa_j(q,\tilde{\theta}_\star,\beta_\star,x_{0,j})}\leq c\right] \overset{P}{\longrightarrow}0,
$$
where $q$ is assumed to be drawn from a standard normal distribution.  
	\label{lem:indicator}
\end{lemma}
\begin{proof}
	Similar to the proof of Lemma \ref{lem:conv}, for $j=1,\cdots,K$, to easy notations, we shall remove $x_{0,j}$ from the notation $\kappa_j(q,\tilde{\theta}_\star,\beta_\star,x_{0,j})$. 
Let $\eta>0$, and consider the following functions parametrized by $\eta$,
	$$\overline{\psi}_\eta(\alpha):=\left\{\begin{array}{ll}
		1,& \alpha\leq c\\
		1-\frac{1}{\eta}(\alpha-c),  & c\leq \alpha\leq c+\eta\\
		0, & \alpha\geq c+\eta,
		\end{array}\right.
		$$
		and 
$$\underline{\psi}_\eta(\alpha):=\left\{\begin{array}{ll}
		1,& \alpha\leq c-\eta\\
		-\frac{1}{\eta}(\alpha-c),  & c-\eta\leq \alpha\leq c\\
0, & \alpha\geq c. \end{array}
		\right.
		$$
		Both functions are Lipschitz with Lipschitz constant $\frac{1}{\eta}$. Moreover, for any $\alpha$,
		$$
	\underline{\psi}_\eta(\alpha)\leq	{\bf 1}_{\{{\alpha\leq c\}}} \leq \overline{\psi}_\eta(\alpha).
		$$
		Define also function $\tilde{\psi}_\eta(\alpha)=\overline{\psi}_\eta(\alpha)-\underline{\psi}_\eta(\alpha)$. Then it is easy to see that:
		\begin{equation}
		\tilde{\psi}_\eta(\alpha)\leq {\bf 1}_{\{c-\eta\leq \alpha \leq c+\eta\}}.
			\label{eq:tilde_psi}
		\end{equation}
		Hence, 
		\begin{equation}
	\frac{1}{K}\sum_{j=1}^K\underline{\psi}_\eta(\hat{x}_j)\leq	\frac{1}{K}\sum_{j=1}^K{\bf 1}_{\{{\hat{x}_j\leq c\}}} \leq \frac{1}{K}\sum_{j=1}^K\overline{\psi}_\eta(\hat{x}_j),
			\label{eq:xj}
		\end{equation}
		and 
		\begin{align}
			&	\frac{1}{K}\sum_{j=1}^{K}\mathbb{E}\left[\underline{\psi}_\eta(\kappa_j(q,\tilde{\theta}_\star,\beta_\star))\right] \leq \frac{1}{K}\sum_{j=1}^{K}\mathbb{P}\left[\kappa_j(q,\tilde{\theta}_\star,\beta_\star)\leq c\right]\nonumber\\
			&\leq  \frac{1}{K}\sum_{j=1}^{K}\mathbb{E}\left[\overline{\psi}_\eta(\kappa_j(q,\tilde{\theta}_\star,\beta_\star))\right]. \label{eq:xj1}
		\end{align}
		Using \eqref{eq:xj} and \eqref{eq:xj1}, it follows that:
		\begin{align}
			&\frac{1}{K}\sum_{j=1}^K\underline{\psi}_\eta(\hat{x}_j)-\frac{1}{K}\sum_{j=1}^{K}\mathbb{E}\left[\overline{\psi}_\eta(\kappa_j(q,\tilde{\theta}_\star,\beta_\star))\right] \nonumber\\
			&\leq\frac{1}{K}\sum_{j=1}^K{\bf 1}_{\{{\hat{x}_j\leq c\}}}- \frac{1}{K}\sum_{j=1}^{K}\mathbb{P}\left[\kappa_j(q,\tilde{\theta}_\star,\beta_\star)\leq c\right]\nonumber\\
			&\leq \frac{1}{K}\sum_{j=1}^K\overline{\psi}_\eta(\hat{x}_j)-\frac{1}{K}\sum_{j=1}^{K}\mathbb{E}\left[\overline{\psi}_\eta(\kappa_j(q,\tilde{\theta}_\star,\beta_\star))\right].
		\end{align}
	From Lemma \ref{lem:conv}, for $\epsilon>0$, with probability approaching one:
	\begin{equation}
	\left|\frac{1}{K}\underline{\psi}_\eta(\hat{x}_j)-\frac{1}{K}\sum_{j=1}^{K}\mathbb{E}\left[\underline{\psi}_\eta(\kappa_j(q,\tilde{\theta}_\star,\beta_\star))\right]\right|\leq \epsilon,
	\end{equation}
and
	\begin{equation}
	\left|\frac{1}{K}\overline{\psi}_\eta(\hat{x}_j)-\frac{1}{K}\sum_{j=1}^{K}\mathbb{E}\left[\overline{\psi}_\eta(\kappa_j(q,\tilde{\theta}_\star,\beta_\star))\right]\right|\leq \epsilon.
	\end{equation}
	Hence,
	\begin{align}
			&-\epsilon-\frac{1}{K}\sum_{j=1}^{K}\mathbb{E}\left[\tilde{\psi}_\eta(\kappa_j(q,\tilde{\theta}_\star,\beta_\star))\right]\nonumber\\
			&\leq\frac{1}{K}\sum_{j=1}^K{\bf 1}_{\{{\hat{x}_j\leq c\}}}- \frac{1}{K}\sum_{j=1}^{K}\mathbb{P}\left[\kappa_j(q,\tilde{\theta}_\star,\beta_\star)\leq c\right]\nonumber\\
			&\leq \epsilon+\frac{1}{K}\sum_{j=1}^{K}\mathbb{E}\left[\tilde{\psi}_\eta(\kappa_j(q,\tilde{\theta}_\star,\beta_\star))\right].	\label{eq:ine3}
	\end{align}
	From \eqref{eq:tilde_psi}, 
	\begin{align*}
		&	\lim_{\eta\to 0} \frac{1}{K}\sum_{j=1}^{K}\mathbb{E}\left[\tilde{\psi}_\eta(\kappa_j(q,\tilde{\theta}_\star,\beta_\star))\right]\\
		&\leq \lim_{\eta\to 0} \max_{1\leq j\leq K} \mathbb{P}\left[c-\eta \leq \kappa_j(q,\tilde{\theta}_\star,\beta_\star)\leq c+\eta\right].
	\end{align*}
	As $c\notin\left\{-t,t\right\}$, the right-hand side of the above inequality converges to zero. Hence, there exists $\eta_0$ such that for all $\eta\leq \eta_0$,
	\begin{equation}
	\frac{1}{K}\sum_{j=1}^{K}\mathbb{E}\left[\tilde{\psi}_\eta(\kappa_j(q,\tilde{\theta}_\star,\beta_\star))\right]\leq \epsilon.
		\label{eq:ine1}
	\end{equation}
	Combining \eqref{eq:ine3} and \eqref{eq:ine1}, we get the desired result, that is that for $\epsilon>0$, with probability approaching $1$,
$$
\left|\frac{1}{K}\sum_{j=1}^K{\bf 1}_{\{{\hat{x}_j\leq c\}}}- \frac{1}{K}\sum_{j=1}^{K}\mathbb{P}\left[\kappa_j(q,\tilde{\theta}_\star,\beta_\star)\leq c\right]\right|\leq 2\epsilon. 
$$
\end{proof}
\begin{corollary}
	Consider the setting of Lemma \ref{lem:indicator}.
	 We thus have:
	$$
	\frac{1}{K}\sum_{j=1}^{K}\mathbb{P}\left[\hat{x}_j\leq c\right] -\frac{1}{K}\sum_{j=1}^{K}\mathbb{P}\left[\kappa_j(q,\tilde{\theta}_\star,\beta_\star,x_{0,j})\leq c \right] \to 0.
	$$
	\label{cor:indicator}
\end{corollary}
\begin{proof}
	The proof follows straightforwardly by applying the dominated convergence theorem. 
\end{proof}
\subsection{{{Applying the CGMT: MSE of Box-RLS}}}
Let $\widehat{\bf x}$ be the solution of \eqref{eq:Box-RLS matrix}.
Recall that the MSE is given by:
$$
{\rm MSE}=\frac{1}{K}\|\widehat{\bf x}-{\bf x}_0\|^2,
$$ 
which can be also written as:
\begin{align*}
	{\rm MSE}&=\frac{1}{\rho_d\sigma_{ {\hat{H}}}^2}\left(\frac{1}{K}\|{\bf v}\|^2-\frac{\rho_d}{K}\sigma_{ {\Delta}}^2\|{\bf x}_0\|^2\right).\\
\end{align*}
As $\frac{1}{K}\|{\bf v}\|^2\overset{P}{\longrightarrow} \delta\theta_\star^2-1$, and $\frac{1}{K}\|{\bf x}_0\|^2\overset{P}{\longrightarrow} 1$, we thus have:
$$
{\rm MSE}\overset{P}{\longrightarrow} \frac{1}{\rho_d\sigma_{ {\hat{H}}}^2}\left(\delta \theta_\star^2-1-{\rho_d}\sigma_{ {\Delta}}^2\right).
$$
\subsection{{{Applying the CGMT: SEP for Box-RLS}}}
In this subsection, we study the limiting behavior of the SEP defined in \eqref{eq:SEP}.  
 For $j=1,\cdots,K$, consider the output of the (PO) problem which denoted by $\hat{\bf x}$ . Recall the expression of the SEP  
 $$
 {\rm SEP}=\frac{1}{K}\sum_{j=1}^{K}\mathbb{P}\left[\hat{s} \neq x_{0,j}, \text{with} \ \ \hat{s}=\argmin_{s \in \mathcal{C}}  |\frac{\hat{x}_j}{B}-s|\right].
 $$
  In PAM-constellations, we distinguish inner symbols when ${x}_{0,j}$ belongs to $\frac{1}{\sqrt{\mathcal{E}}}\{\pm1,\pm3,\cdots,\pm(M-3) \}$ from  edge symbols when ${x}_{0,j}=\pm \frac{M-1}{\sqrt{\mathcal{E}}}$.
Let us consider first the case where ${x}_{0,j}$ is a certain inner symbol.   
In this case, there is an error if and only if $| \frac{\hat{x}_j }{B}-x_{0,j} | > \frac{1}{\sqrt{\mathcal{E}}}$, where $B = \frac{ \xi^2 \beta_\star \tilde{\theta}_\star }{\xi^2 {\beta_\star}\tilde{\theta}_\star + 2 \lambda \rho_{ {d}}}$.
Similarly, considering the edge points $x_{0,j}=\pm\frac{M-1}{\sqrt{\mathcal{E}}}$, we deduce that there is an error if $ \frac{\hat{x}_j }{B}-x_{0,j}  > \frac{1}{\sqrt{\mathcal{E}}}$ when $x_{0,j}=\frac{M-1}{\sqrt{\mathcal{E}}}$ and if $\frac{\hat{x}_j }{B}-x_{0,j}  < -\frac{1}{\sqrt{\mathcal{E}}}$ when $x_{0,j}=-\frac{M-1}{\sqrt{\mathcal{E}}}$. The SEP thus becomes:
\begin{align*}
	{\rm SEP}&=\frac{1}{K}\sum_{j=1}^K\frac{1}{M}\sum_{\substack{i=\pm 1,\pm 3,\\
	\cdots,\pm (M-2)}}\mathbb{P}\left[\left|\frac{\hat{x}_j}{B}-{ x}_{0,j}\right|> \frac{1}{\sqrt{\mathcal{E}}} \ | x_{0,j}=\frac{i}{\sqrt{\mathcal{E}}}\right]\\
	&+\frac{1}{K}\sum_{j=1}^{K}\frac{1}{M}\mathbb{P}\left[\frac{\hat{x}_j}{B}-{ x}_{0,j}< -\frac{1}{\sqrt{\mathcal{E}}} \ | \ x_{0,j}=\frac{M-1}{\sqrt{\mathcal{E}}}\right]\\
	&+\frac{1}{K}\sum_{j=1}^{K}\frac{1}{M}\mathbb{P}\left[\frac{\hat{x}_j}{B}-{ x}_{0,j}> \frac{1}{\sqrt{\mathcal{E}}} \ | \ x_{0,j}=-\frac{M-1}{\sqrt{\mathcal{E}}}\right].\\
\end{align*}
Based on Corollary \ref{cor:indicator}, we have
$$
{\rm SEP}-\widetilde{{\rm SEP}} \to 0,
$$
where $\widetilde{{\rm SEP}}$ follows after some tedious but straightforward calculations and is given in \eqref{eq:SEP_BoX}.
\begin{figure*}
	\begin{equation}
		\label{eq:SEP_BoX}
	\begin{aligned}	\widetilde{\rm SEP}&=\frac{4}{M}\sum_{i=1,3,\dots,M-3}{\bf 1}_{\{\frac{t}{B}\geq\frac{i+1}{\sqrt{\mathcal{E}}}\}}Q\left(\frac{\xi\tilde{\theta}_\star}{\sqrt{\mathcal{E}}}\right)+\frac{2}{M}\sum_{i=1,3,\cdots,M-3}\left\{{\bf 1}_{\{\frac{i-1}{\sqrt{\mathcal{E}}}\leq \frac{t}{B}\leq \frac{i+1}{\sqrt{\mathcal{E}}}\}}Q\left(\frac{\xi\tilde{\theta}_\star}{\sqrt{\mathcal{E}}}\right)+{\bf 1}_{\{\frac{t}{B}\leq \frac{i-1}{\sqrt{\mathcal{E}}}\}}\right\}\\
		&+\frac{2}{M}{\bf 1}_{\{\frac{t}{B}\geq \frac{M-2}{\sqrt{\mathcal{E}}}\}}Q\left(\frac{\xi\tilde{\theta}_\star}{\sqrt{\mathcal{E}}}\right)+\frac{2}{M}{\bf 1}_{\{\frac{t}{B}\leq \frac{M-2}{\sqrt{\mathcal{E}}}\}}.
	\end{aligned}
	\end{equation}
	\hrule
\end{figure*}
\section{Un-Boxed RLS Proofs}\label{sec:unbox}
The analysis of the Un-Boxed RLS scheme is similar to the analysis of the Box-RLS. Below we provide a brief sketch of the proof. Following the same analysis as before, we identify the same (PO) and (AO) with the single difference that the constraints on $\left\{x_j \right\}_{j=1}^{K}$ are now removed. Particularly, the (AO) associated with the RLS writes as:
\begin{align}
	\phi=	& \ \underset{\beta > 0}{\operatorname{\max}} \ \underset{0\leq \tilde{\theta}  \leq C}{\operatorname{\min}}\ \frac{  \beta \| \gv \|^2}{2\tilde{\theta}K} + \frac{ \beta \tilde{\theta}}{2 } - \frac{\beta^2}{4} + \frac{\beta \tilde{\theta}\rho_{ {d}} \| {\xv}_0 \|^2}{2 K}   + \sqrt{\rho_{ {d}}} \beta \frac{1}{K}{\xv}_0^{T} \tilde{\hv} \nonumber \\
&+ \frac{1}{K}\sum_{j = 1}^{K} \biggl[ \underset{ x_j\in\mathbb{R}}{\operatorname{\min}} \ \biggl(\frac{\beta \tilde{\theta}\rho_{ {d}}\sigma_{ {\hat{H}}}^2}{2} + \lambda \rho_{ {d}} \biggr) x_j^2 -  \beta \biggl( \rho_{ {d}} \sigma_{ {\hat{H}}}^2 {x}_{0,j}\tilde{\theta}+\sqrt{\rho_{ {d}}} \sigma_{ {\hat{H}}} q^{1}_{j}   \biggr) {x}_j \biggr].
\end{align}
which is similar to \eqref{AA23} with the difference that the optimization over $x_j$ is on the whole real axis. 
Optimizing over the variables $x_j,j=1,\cdots,K$, we thus obtain:
\begin{align}
	\phi&= \ \underset{\beta > 0}{\operatorname{\max}} \ \underset{0\leq \tilde{\theta}  \leq C}{\operatorname{\min}}\ \frac{  \beta \| \gv \|^2}{2\tilde{\theta}K} + \frac{ \beta \tilde{\theta}}{2 } - \frac{\beta^2}{4} + \frac{\beta \tilde{\theta}\rho_{ {d}} \| {\xv}_0 \|^2}{2 K} + \sqrt{\rho_{ {d}}} \beta \frac{1}{K}{\xv}_0^{T} \tilde{\hv}-\frac{1}{K}\sum_{j=1}^{K} \frac{\beta^2(\xi^2 x_{0,j}\tilde{\theta}+\xi q_j^{1})^2}{2\xi^2\beta\tilde{\theta}+4\lambda \rho_{ {d}}}. \label{eq:cost}
\end{align}
Using the same approach as that in the Box-RLS, we can prove that $\phi$ converges to:
\begin{align}
	\phi\overset{P}{\rightarrow} \overline{\phi}:=\sup_{\beta>0} \inf_{\tilde{\theta}>0} \ \ &\frac{\beta \delta}{2\tilde{\theta}}+\frac{\beta\tilde{\theta}}{2}-\frac{\beta^2}{4} + \frac{\beta\tilde{\theta}\rho_{ {d}}}{2}-\frac{\beta^2\xi^2\tilde{\theta}^2+\xi^2}{2\xi^2\beta\tilde{\theta}+4\lambda\rho_{ {d}}}.
\end{align}
Using the change of variable $\theta=\frac{1}{\tilde{\theta}}$, this brings us to solve the following max-min problem:
$$
\overline{\phi}=\sup_{\beta >0}\inf_{\theta>0}\beta\theta\frac{\delta}{2}+\frac{\beta}{2\theta}(1+\rho_{ {d}})-\frac{\beta^2}{4}-\frac{\beta^2\sigma_{\hat{ {H}}}^2}{\frac{2\beta\sigma_{ {H}}^2}{\theta}+4\lambda}(1+\frac{\rho_{ {d}}\sigma_{\hat{ {H}}}^2}{\theta^2}),$$
for which we prove that there exists a unique solution $(\theta_\star,\beta_\star)$. 
As a matter of fact, writing the first order optimality conditions, we obtain:
\begin{subequations}\label{all_derv}
\begin{align}\label{deriv_t}
\delta \beta - \frac{\beta}{\theta^2}-\frac{\rho_{ {d}} \sigma_{ {\Delta}}^2 \beta}{\theta^2} - \frac{\beta \sigma_{ {\hat{H}}}^2  (\beta^2 \sigma_{ {\hat{H}}}^2 + 4 \rho_{ {d}} \lambda^2 )}{( \beta \sigma_{ {\hat{H}}}^2 + 2 \lambda \theta )^2} = 0, \\
\label{deriv_b}
\delta \theta + \frac{1}{\theta} - \beta + \frac{\rho_{ {d}} \sigma_{ {\Delta}}^2}{\theta} -  \frac{\sigma_{ {\hat{H}}}^2 \theta(\beta^2\sigma_{ {\hat{H}}}^2 + 4  \lambda \theta \beta - 4 \rho_{ {d}}\lambda ^2)}{( \beta \sigma_{ {\hat{H}}}^2 + 2 \lambda \theta )^2} = 0.
\end{align}
\end{subequations}
Combining the two equations together $\big(\frac{1}{ \beta}$ (\ref{deriv_t}) $+ \frac{1}{ \theta}$ (\ref{deriv_b})$\big)$, gives
\begin{equation}
\delta - \frac{\beta}{2 \theta} - \frac{\beta \sigma_{ {\hat{H}}}^2 }{ \beta \sigma_{ {\hat{H}}}^2 + 2 \lambda \theta } =0,
\end{equation}
which is equivalent to a second order polynomial in $\theta$, the solution of which is given by:
\begin{equation}\label{eq:tau}
\theta = \frac{1}{2 \lambda'} \Upsilon(\lambda', \delta) \beta,
\end{equation}
where, $\Upsilon(\lambda', \delta) = \frac{-(\delta - \lambda' -1)+\sqrt{(\delta - \lambda' -1)^2 + 4 \lambda' \delta}}{2 \delta}$, and $\lambda' = \frac{\lambda}{ \sigma_{ {\hat{H}}}^2}$.
Substituting (\ref{eq:tau}) into (\ref{all_derv}) and after some algebraic manipulations, the unique solution $(\theta_\star,\beta_\star)$ of the system is 
\begin{subequations}
\begin{align}
\theta_\star = \sqrt{\frac{\rho_{ {d}} \sigma_{ {\hat{H}}}^2  \kappa^2+  \rho_{ {d}} \sigma_{ {\Delta}}^2 +1 }{\delta-(1-\kappa)^2}}, \label{tau:eq} \\
\beta_\star = 2 \biggl( (\delta -\lambda' -1) + \delta \Upsilon(\lambda', \delta) \biggr) \theta_\star, \label{beta:eq}
\end{align}
\end{subequations}
where $\kappa = \frac{\Upsilon(\lambda', \delta)}{1+  \Upsilon(\lambda', \delta)}$.
Substituting $\kappa$ back gives the same expression as $\theta_\star$ in Theorem \ref{MSE-RLS}.
{{Similar to the Box-RLS case, we can obtain the same MSE and SEP expressions as before. The only difference is that $\theta_\star$ and $\beta_\star$ are now given by the closed form expressions in (\ref{tau:eq}) and (\ref{beta:eq}) respectively.}}
\section{Proof of \eqref{eq:HH}}
\label{app:pp}
To begin with, we perform the change of variable ${\bf s}=\sigma_{\hat{H}}{\bf x}-{\bf x}_0$. Function $\hat{\mathcal{H}}({\beta},{\bf x})$ can be lower-bounded as:
\begin{align*}
\hat{\mathcal{H}}({\beta},{\bf x})\geq &  {\beta}\frac{\|{\bf g}\|}{\sqrt{K}}\sqrt{\rho_d} -\frac{\tilde{\beta}^2}{4}-\tilde{\beta}\sqrt{\rho_d}\frac{{\bf s}^{T}{\bf q}^{\bf 1}}{K}-{\beta}\sqrt{\rho_d} \frac{{\bf x}_0^{T}{\bf q}^{\bf 1}}{K} +{\beta}\sqrt{\rho_{ {d}}} \frac{1}{K}{\bf x}_0^{T}\tilde{\bf h} +\frac{\lambda \rho_d}{\sigma_{\hat{H}}^2} \frac{1}{K}\left(\|{\bf s}\|^2+\|{\bf x}_0\|^2+2{\bf s}^{T}{\bf x}_0\right).
\end{align*}
Denote by $\kappa$ the norm of $\frac{1}{\sqrt{K}}\|{\bf s}\|$.
Minimizing the lower-bound over ${\bf s}$, we obtain:
\begin{align*}
	\hat{\mathcal{H}}({\beta},{\bf x})\geq  \min_{\kappa\geq 0} & \tilde{\beta}\frac{\|{\bf g}\|}{\sqrt{K}}\sqrt{\rho_{ {d}}}-\frac{{\beta}^2}{4}-{\beta}\sqrt{\rho_d}\frac{{\bf x}_0^T{\bf q}^{\bf 1}}{K} +\tilde{\beta}\sqrt{\rho_d}\frac{1}{K}{\bf x}_0^{T}\tilde{\bf h}\\
	&+\frac{\lambda\rho_{ {d}}}{\sigma_{\hat{ {H}}}^2} \kappa^2 +\frac{\lambda\rho_{ {d}}}{\sigma_{\hat{H}}^2} \|{\bf x}_0\|^2-\kappa \left\|\tilde{\beta}\sqrt{\rho_d}\frac{1}{\sqrt{K}}{\bf q}^{\bf 1}-2\frac{\lambda \rho_{ {d}}}{\sqrt{K}\sigma_{\hat{ {H}}}^2}{\bf x}_0 \right\| \\
	&={{\beta}}\frac{\|{\bf g}\|}{\sqrt{K}}\sqrt{\rho_{ {d}}}-\frac{{\beta}^2}{4}-{\beta}\sqrt{\rho_{ {d}}}\frac{{\bf x}_0^T{\bf q}^{\bf 1}}{K} +{\beta}\sqrt{\rho_{ {d}}}\frac{1}{K}{\bf x}_0^{T}\tilde{\bf h} +\frac{\lambda \rho_{ {d}}}{\sigma_{\hat{ {H}}}^2}\|{\bf x}_0\|^2 -\frac{\sigma_{\hat{ {H}}}^2}{4\lambda \rho_{ {d}} K}  \|2\frac{\lambda\rho_{ {d}}}{\sigma_{\hat{ {H}}}^2}{\bf x}_0-{\beta}\sqrt{\rho_{ {d}}}{\bf q}^{\bf 1}\|^2.
\end{align*}
In the asymptotic regime with $K$ and $N$ tending to infinity, and taking ${\beta}$ as fixed, the right-hand side  of the above equation converges to:
$$
{\beta}\sqrt{\rho_d}\sqrt{\delta }-\frac{{\beta^2}}{4}-{\beta}^2\frac{\sigma_{\hat{ {H}}}^2}{4\lambda},
$$
which simplifies to $\frac{\rho_{ {d}}\delta}{1+\frac{\sigma_{\hat{ {H}}}^2}{4\lambda}} $ when ${\beta}$ is replaced by ${\beta}_0:=\frac{2\sqrt{\rho_{ {d}}}\sqrt{\delta}}{1+\frac{\sigma_{\hat{ {H}}}^2}{4\lambda}}$. We thus obtain almost surely
$$
\hat{\mathcal{H}}({\beta}_0,{\bf x})\geq \frac{1}{2}\frac{\rho_{ {d}}\delta}{1+\frac{\sigma_{\hat{ {H}}}^2}{4\lambda}}, 
$$
and as such, almost surely,
$$
\min_{\substack{-t\leq  x_j\leq t\\ j=1,\cdots,K}}\hat{\mathcal{H}}({\beta}_0,{\bf x})\geq \frac{1}{2}\frac{\rho_{ {d}}\delta}{1+\frac{\sigma_{\hat{ {H}}}^2}{4\lambda}}.
$$
Setting $\Theta=\frac{1}{2}\frac{\rho_{ {d}}\delta}{1+\frac{\sigma_{\hat{ {H}}}^2}{4\lambda}}$ ends up the proof of \eqref{eq:HH}.  
\section{{ Monotonicity of the MSE and SEP of the RLS decoder}}\label{Azx}
We showed that for the LS case, optimizing the power allocation in MSE sense is equivalent to optimizing the SEP and it boils down to maximizing $\rho_{\text{eff}}$. 

In this appendix, we will show that this holds also true for the RLS decoder that employs optimal regularization coefficient. Towards this goal, we proceed with the following change of variables  
$J=\frac{1}{\rho_\text{eff}}, c_1 = 2 (1+\delta)$, and $c_2 =(1-\delta)^2$. Then, the MSE and SEP write as:
\begin{align}
\widetilde{\rm{ MSE}}_{\text{\tiny{RLS}}}&=\frac{1}{2}\bigg( -J +(1-\delta) +\sqrt{J^2+c_1 J+c_2}\bigg),\label{eq:MSE_RLS}\\
\widetilde{\rm{ SEP}}_{\text{\tiny{RLS}}}&=2\left(1-\frac{1}{M}\right)
  Q\left(\sqrt{\frac{\delta}{\mathcal{E}\left(\frac{1}{2}\bigg( J +(1-\delta) +\sqrt{J^2+c_1 J+c_2}\bigg)\right)}}\right).\label{eq:SER_RLS}
\end{align}
It appears from \eqref{eq:MSE_RLS} and \eqref{eq:SER_RLS} that to minimize the MSE or the SEP, it suffices to minimize for all $i\in\left\{0,1\right\}$, function $\mathcal{F}_i(J)=(-1)^i J+ \sqrt{J^2+c_1J+c_2}$. We can check easily that the first derivative of  $\mathcal{F}_i$ is given by: $\mathcal{F}_i^{'}(J)=(-1)^{i}+\frac{2J+c_1}{2\sqrt{J^2+c_1J+c_2}}$ which is strictly positive for all $i\in\{0,1\}$. It follows thus that both  $\widetilde{\rm{ MSE}}_{\text{\tiny{RLS}}}$ and $\widetilde{\rm{ SER}}_{\text{\tiny{RLS}}}$ are increasing functions of $\frac{1}{\rho_{\text{eff}}}$ and hence minimizing them amounts to maximizing $\rho_{\text{eff}}$. 
Then, $\widetilde{\rm{ MSE}}_{\text{\tiny{RLS}}} =\frac{1}{2}\bigg( -J +(1-\delta) +\sqrt{J^2+c_1 J+c_2}\bigg)$. On the other hand, from \eqref{SER_MSE_RLS}, it is clear that to minimize the SEP we need to maximize the argument of the $Q$-function, which means to minimize $\widetilde{\rm{ MSE}}_{\text{\tiny{RLS}}} +J$.
\section{Optimal Power Allocation Derivation: Proof of Theorem~\ref{Power_Th}}\label{power_proof}
Here, we derive the optimal power allocation given in Theorem~\ref{Power_Th}. First, rewrite $\rho_\text{eff}$ as follows
\begin{align*}
\rho_\text{eff}&= \frac{\alpha \rho \frac{\tau}{\tau_ {d}}. (1-\alpha) \rho \tau}{(1+ \alpha \rho\frac{\tau}{\tau_ {d}}) + (1- \alpha)\rho \tau} \nonumber \\
& =\frac{(\rho \tau)^2}{\tau_ {d}} \cdot \frac{\alpha (1-\alpha)}{1+ \rho \tau - \alpha \rho \tau (1 - \frac{1}{\tau_ {d}}) } \nonumber \\
& = \frac{\rho \tau}{\tau_ {d} -1} \cdot \frac{\alpha (1- \alpha)}{-\alpha + \frac{1 + \rho \tau}{\rho \tau(1- \frac{1}{\tau_ {d}})}}.
\end{align*}
We need to maximize $\rho_\text{eff}$ over $0< \alpha <1$. To do so, we consider the following cases:\\
1) $\tau_ {d} =1$:
\begin{align*}
\rho_\text{eff} = \frac{(\rho \tau)^2}{1 + \rho \tau} \alpha (1-\alpha).
\end{align*}
In this case, $\alpha_\star = \frac{1}{2}$.\\
2) $\tau_ {d} >1$:
\begin{align*}
\rho_\text{eff} = \frac{\rho \tau}{\tau_ {d} -1} \cdot \frac{\alpha (1-\alpha)}{-\alpha +\vartheta}, \quad \vartheta = \frac{1+\rho \tau}{\rho \tau(1- \frac{1}{\tau_ {d}})} >1.
\end{align*}
Hence, the optimal $\alpha_\star$ that maximizes $\rho_{\text{eff}}$ is given by:
\begin{align*}
\alpha_\star &= \argmax_{\substack{0<\alpha<1 \\ \vartheta>1}} \frac{\alpha(1-\alpha)}{-\alpha+ \vartheta} \\
&= \vartheta -\sqrt{\vartheta(\vartheta-1)}.
\end{align*}
3) $\tau_ {d} <1$:
\begin{align*}
\rho_\text{eff} = \frac{\rho \tau}{1-\tau_ {d} } \cdot \frac{\alpha (1-\alpha)}{\alpha -\vartheta}, \quad \vartheta = \frac{1+\rho \tau}{\rho \tau(1- \frac{1}{\tau_ {d}})} <0.
\end{align*}
In this case, we have
\begin{align*}
\alpha_\star &= \argmax_{\substack{0<\alpha<1 \\ \vartheta<0}} \frac{\alpha(1-\alpha)}{\alpha- \vartheta} \\
&= \vartheta +\sqrt{\vartheta(\vartheta-1)}.
\end{align*}
\section{Optimal Power and Training Time Allocation Derivation based on Goodput}\label{goodput_proof}
In this section, we determine the optimal power and training time allocation that optimizes the asymptotic value of the goodput for LS decoder which we denote here by $\widetilde{G}_{LS}$.

To this end, we proceed with the change of variable $ \gamma = \frac{T_ {p}}{T}$. Then, $\tau_ {p} =\frac{T_ {p}}{K} = \frac{T_ {p}}{T} $,  $ \frac{T}{K} = \gamma \tau, \delta' = \frac{\delta-1}{\mathcal{E}}$, $\rho_ {p} = \frac{(1-\alpha) \rho}{\gamma}$, and $\rho_ {d} = \frac{\alpha \rho}{1-\gamma}$.
The SEP for LS case is
\begin{align*}
\widetilde{\text{SEP}}_{\tiny{\text{LS}}} &= 2(1- \frac{1}{M}) Q \bigg(\sqrt{\frac{\delta -1}{\mathcal{E}} \rho_{\text{eff}}}\bigg) \\
 & = 2(1- \frac{1}{M}) Q \bigg(\sqrt{\frac{\delta -1}{\mathcal{E}} \cdot \frac{\rho_ {d} \rho_ {p} \tau_ {p}}{1+ \rho_ {d} + \rho_ {p} \tau_ {p}}}\bigg)\\
& = 2(1- \frac{1}{M}) Q \bigg(\sqrt{\frac{\delta -1}{\mathcal{E}} \cdot \frac{\frac{\alpha \rho}{1-\gamma} \frac{(1-\alpha) \rho}{\gamma} \gamma \tau}{1+ \frac{\alpha \rho}{1-\gamma}+ \frac{(1-\alpha)\rho}{\gamma} \gamma \tau} }\bigg)\\
& =2\big(1- \frac{1}{M}\big) Q \bigg(\sqrt{ \frac{\delta' \alpha (1-\alpha) \rho^2 \tau}{1+\alpha \rho +(1-\alpha) \rho \tau - \gamma (1 +(1-\alpha) \rho \tau)}}\bigg) \\
&= 2\big(1- \frac{1}{M}\big) Q \bigg(\sqrt{\frac{\frac{\delta' \alpha (1-\alpha) \rho^2 \tau}{1 +(1-\alpha) \rho \tau}}{\frac{1+ \alpha \rho +(1-\alpha) \rho \tau}{1 +(1-\alpha) \rho \tau} - \gamma}} \bigg)\\
& = \widetilde{M} Q \bigg(\sqrt{\frac{b(\alpha)}{a(\alpha) -\gamma}} \bigg),
\end{align*}
where $\widetilde{M} = 2\big(1- \frac{1}{M}\big), a(\alpha)= \frac{1+ \alpha \rho +(1-\alpha) \rho \tau}{1 +(1-\alpha) \rho \tau}>1$ , and $b(\alpha) = \frac{\delta' \alpha (1-\alpha) \rho^2 \tau}{1 +(1-\alpha) \rho \tau}>0$. The asymptotic goodput becomes
\begin{align*}
\widetilde{G}_{\tiny{\text{LS}}}(\alpha,\gamma) &= (1- \gamma) (1- \widetilde{\text{SEP}}_{\tiny{\text{LS}}})\\
& =(1-\gamma)\bigg(1- \widetilde{M} Q\bigg(\sqrt{\frac{b(\alpha)}{a(\alpha) -\gamma}} \bigg) \bigg).
\end{align*}
The power allocation problem amounts thus to solving:
\begin{equation}
(\alpha_\star,\gamma_\star) = \argmax_{\substack{0 < \alpha < 1 \\ \frac{1}{\tau} \leq \gamma < 1}} \widetilde{G}_{\tiny{\text{LS}}}(\alpha,\gamma),
\end{equation}
Recall that we need $T_{ {p}} \geq K$ or $\tau_{ {p}} \geq 1$, but $\tau_{ {p}} = \gamma \tau$, hence $\gamma \geq \frac{1}{\tau}$. We also require $T_{ {p}} < T$, hence $\gamma <1$. This justifies the constraint imposed on $\gamma$ above.\\
To begin with, it is easy to see that  the optimal $\alpha$ is the one that maximizes $\rho_{\rm eff}$ and this is for any $\gamma$. It remains thus to optimize the goodput in terms of $\gamma$. We will solve this by proving that the good-put is a decreasing function with respect to $\gamma$. 
To proceed, let us make the change of variables: $x =\sqrt{\frac{b}{a-\gamma}}$, then $\gamma = a -\frac{b}{x^2}$. Hence, for any $0<\alpha<1$, the goodput is
\begin{equation*}
	\widetilde{G}_{\rm LS}(\alpha,\gamma) = \overline{G}_{\text{\tiny LS}}(x):=\bigg(1- a +\frac{b}{x^2} \bigg) \bigg(1- \widetilde{M} Q(x)\bigg).
\end{equation*}
Taking the derivative of $\overline{G}_{\tiny{\text{LS}}}(x)$ with respect to $x$ yields
\begin{align*}
\overline{G}_{\tiny{\text{LS}}}'(x) &=\frac{-2 b}{x^3}\bigg(1- \widetilde{M} Q(x)\bigg) + \bigg(1- a +\frac{b}{x^2}\bigg)\bigg(\frac{\widetilde{M}}{\sqrt{2 \pi}} e^{\frac{-x^2}{2}}\bigg)\\
& = \widetilde{M} e^{\frac{-x^2}{2}}\bigg[\frac{-2 b}{x^3}e^{\frac{x^2}{2}}\bigg(\frac{1}{\widetilde{M} }- Q(x)\bigg) + \frac{1}{\sqrt{2 \pi}}\bigg(1- a +\frac{b}{x^2}\bigg) \bigg].
\end{align*}
We need to study the sign of $\overline{G}_{\tiny{\text{LS}}}'(x)$. To do so, first write the Taylor series expansion of $e^{\frac{x^2}{2}}\bigg(\frac{1}{\widetilde{M}}-  Q(x)\bigg)$ as:\\
$e^{\frac{x^2}{2}}\bigg(\frac{1}{\widetilde{M}}-  Q(x)\bigg) = \frac{1}{2(M-1)} + \frac{x}{\sqrt{2 \pi}} + \frac{x^2}{4(M-1)}+ \cdots$, then
\begin{align*}
\overline{G}_{\tiny{\text{LS}}}'(x) &= \widetilde{M} e^{\frac{-x^2}{2}}\bigg[ \frac{-b}{(M-1) x^3}- \frac{b}{\sqrt{2 \pi} x^2} 
 -\frac{b}{4(M-1) x}+ \cdots +\frac{1}{\sqrt{2 \pi}} (1- a) \bigg].
\end{align*}
Recall that $a>1$ and note that all the terms in $\overline{G}_{\tiny{\text{LS}}}'(x)$ expression above are negative. Hence, $\overline{G}_{\tiny{\text{LS}}}'(x)<0$ for all $\alpha$.
Now, by the chain rule: $\frac{{\rm d} \overline{G}_{\tiny{\text{LS}}}}{{\rm d} x} = \frac{d \overline{G}_{\tiny{\text{LS}}}}{{\rm d} \gamma}\cdot \frac{{\rm d} \gamma}{{\rm d} x}$, but $\frac{{\rm d} \gamma}{{\rm d} x} =\frac{2 b}{x^3}>0$, then $\frac{{\rm d} \overline{G}_{\tiny{\text{LS}}}}{{\rm d} \gamma}<0$.\\
Hence, $\frac{\partial {G}_{\tiny{\text{LS}}}(\alpha, \gamma)}{\partial \alpha}<0$ for any $\alpha$. This suggests that we choose $\gamma$ as small as possible to maximize $\widetilde{G}_{\tiny{\text{LS}}}(\alpha, \gamma)$.\\
So, $\gamma_\star = \frac{1}{\tau}$. Or, $\tau_{ {p}}^{\star} = \gamma_\star \tau$, then $\tau_{ {p}}^{\star} =1.$\\

The proof above is done for LS but the same conclusions hold for both RLS and also for the Box-RLS under the conjecture that for Box-RLS, the optimal $\alpha$ is the one that maximizes $\rho_{\rm eff}$. We omit details for briefness. 
\section{Comparison with the LMMSE Decoder}\label{LMMSE_Append}
In this appendix we will show that the LMMSE estimator of $\xv_0$ is equivalent to an RLS estimator with the optimal regularizer $\lambda_\star = \frac{1}{\rho_ {d}}+ \sigma_{ {\Delta}}^2$. The LMMSE estimate of $\xv_0$ is given by \cite{kay1993fundamentals}
\begin{equation}
\widehat{\xv}_{\text{LMMSE}} = \Cm_{x y} \Cm_{y y}^{-1} \yv,
\end{equation}
where $ \Cm_{x y} = \mathbb{E} [ (\xv_0 -\mathbb{E}[\xv_0]) (\yv -\mathbb{E}[\yv])^T]$, and $ \Cm_{y y} = \mathbb{E} [ (\yv -\mathbb{E}[\yv]) (\yv -\mathbb{E}[\yv])^T]$. It can be shown that $\mathbb{E}[\yv] =\mathbb{E}[\xv_0] = \mathbf{0}$, and
\begin{equation}
\Cm_{x y} = \sqrt{\frac{\rho_{ {d}}}{K}} \widehat{\Hm}^T.
\end{equation}
To find $ \Cm_{y y}$, let us write $\yv$ as
\begin{align}
 \yv=& \sqrt{\frac{\rho_{ {d}}}{K}} (\widehat{\Hm}+\Deltam) \xv_0 +\zv, \nonumber \\
  =&\sqrt{\frac{\rho_{ {d}}}{K}} \widehat{\Hm} \xv_0 +\widetilde{\zv}, 
\end{align}
where $\widetilde{\zv} \triangleq \sqrt{\frac{\rho_{ {d}}}{K}} \Deltam \xv_0 +\zv$ which is a zero-mean vector with $\Cm_{\tilde{z} \tilde{z}} =\frac{\rho_{ {d}}}{K}\mathbb{E}_ {{\xv_0,\Delta}}[\Deltam\xv_0 \xv_0^T \Deltam^T]+\Cm_{zz}=\frac{\rho_{ {d}}}{K}\mathbb{E}_ {{\Delta}}[\Deltam \mathbb{E}_{\xv_0}[\xv_0 \xv_0^T] \Deltam^T]+\Cm_{zz}=\frac{\rho_{ {d}}}{K}\mathbb{E}_ {{\Delta}}[\Deltam \Cm_{x x} \Deltam^T]+\Cm_{zz} =\frac{\rho_{ {d}}}{K}\mathbb{E}_ {{\Delta}}[\Deltam \Deltam^T]+\Id_N=(\rho_{ {d}} \sigma_{ {\Delta}}^2 +1) \Id_N$, then,
\begin{equation}
\Cm_{y y} =\frac{\rho_{ {d}}}{K} \widehat{\Hm} \widehat{\Hm}^T +(\rho_{ {d}} \sigma_{ {\Delta}}^2 +1) \Id_N. 
\end{equation}
Note that we used $\Cm_{xx} =\Id_K$, $\Cm_{zz} =\Id_N$, and $\mathbb{E}[\Deltam \Deltam^T]= K\sigma_{ {\Delta}}^2 \Id_N$.}\\
Hence,  
\begin{align}\label{LMMSE_a}
\widehat{\xv}_{\text{LMMSE}} =& \sqrt{\frac{\rho_{ {d}}}{K}} \widehat{\Hm}^T \bigg(\frac{\rho_{ {d}}}{K} \widehat{\Hm} \widehat{\Hm}^T + (\rho_{ {d}} \sigma_{ {\Delta}}^2 +1) \Id_N \bigg)^{-1} \yv, \nonumber\\
=&\sqrt{\frac{\rho_{ {d}}}{K}} \bigg(\frac{\rho_{ {d}}}{K} \widehat{\Hm}^T \widehat{\Hm} + (\rho_{ {d}} \sigma_{ {\Delta}}^2 +1) \Id_K \bigg)^{-1} \widehat{\Hm}^T \yv, \nonumber \\
=& \big(\Am^T \Am +(1+\rho_{ {d}} \sigma_{ {\Delta}}^2 )\Id_K \big)^{-1} \Am^T \yv,
\end{align}
where $\Am =\sqrt{\frac{\rho_{ {d}}}{K}} \widehat{\Hm}$, and the second equality follows from the \emph{matrix inversion Lemma}.
The LMMSE estimate in \eqref{LMMSE_a} is equivalent to the RLS solution as given in \eqref{eq:LS matrixb} with $\lambda \rho_{ {d}} = \lambda_\star \rho_{ {d}}= 1+ \rho_{ {d}} \sigma_{ {\Delta}}^2$. 
This shows that the RLS with optimal regularizer is nothing but the popular LMMSE decoder. Finally, the LMMSE estimate can be written in terms of $\widehat{\Hm}$ as
\begin{align}
\widehat{\xv}_{\text{LMMSE}} &=\sqrt{\frac{K}{\rho_{ {d}}}} \bigg( \widehat{\Hm}^T \widehat{\Hm} + \frac{\sigma_{ {\hat{H}}}^2}{\rho_{\text{eff}}} K \Id_K \bigg)^{-1} \widehat{\Hm}^T \yv \nonumber \\
& = \widehat{\xv}_{\text{RLS}}(\lambda_\star).
\end{align}
\newpage
\clearpage
\bibliographystyle{IEEEbib}
\bibliography{References_Arxiv}

\begin{thebibliography}{10}

\bibitem{alrashdi2018optimum}
Ayed~M Alrashdi, Ismail~Ben Atitallah, Tarig Ballal, Christos Thrampoulidis,
  Anas Chaaban, and Tareq~Y Al-Naffouri,
\newblock ``Optimum training for mimo bpsk transmission,''
\newblock in {\em 2018 IEEE 19th International Workshop on Signal Processing
  Advances in Wireless Communications (SPAWC)}. IEEE, 2018, pp. 1--5.

\bibitem{Tela99}
I.~E. Telatar,
\newblock ``Capacity of multi-antenna gaussian channels,''
\newblock {\em European Trans. Telecommun.}, vol. 10, pp. 585 -- 595, 1999.

\bibitem{Fosc98}
G.~J. Foschini and M.~J. Gans,
\newblock ``On limits of wireless communications in a fading environment when
  using multiple antennas,''
\newblock {\em Wirel. Pers. Commun.}, vol. 6, no. 3, pp. 311--335, 1998.

\bibitem{hassibi2003much}
Babak Hassibi and Bertrand~M Hochwald,
\newblock ``How much training is needed in multiple-antenna wireless links?,''
\newblock {\em IEEE Transactions on Information Theory}, vol. 49, no. 4, pp.
  951--963, 2003.

\bibitem{kannu2005capacity}
Arun~P Kannu and Philip Schniter,
\newblock ``Capacity analysis of mmse pilot-aided transmission for doubly
  selective channels,''
\newblock in {\em IEEE 6th Workshop on Signal Processing Advances in Wireless
  Communications, 2005.} IEEE, 2005, pp. 801--805.

\bibitem{gottumukkala2009}
VK~Varma Gottumukkala and Hlaing Minn,
\newblock ``Optimal pilot power allocation for ofdm systems with transmitter
  and receiver iq imbalances,''
\newblock in {\em GLOBECOM 2009-2009 IEEE Global Telecommunications
  Conference}. IEEE, 2009, pp. 1--5.

\bibitem{simko2011}
Michal Simko, Stefan Pendl, Stefan Schwarz, Qi~Wang, Josep~Colom Ikuno, and
  Markus Rupp,
\newblock ``Optimal pilot symbol power allocation in lte,''
\newblock in {\em 2011 IEEE Vehicular Technology Conference (VTC Fall)}. IEEE,
  2011, pp. 1--5.

\bibitem{vsimko2012optimal}
Michal {\v{S}}imko, Qi~Wang, and Markus Rupp,
\newblock ``Optimal pilot symbol power allocation under time-variant
  channels,''
\newblock {\em EURASIP Journal on Wireless Communications and Networking}, vol.
  2012, no. 1, pp. 225, 2012.

\bibitem{wang2011power}
Jun Wang, OliverYu Wen, Hongyang Chen, and Shaoqian Li,
\newblock ``Power allocation between pilot and data symbols for mimo systems
  with mmse detection under mmse channel estimation,''
\newblock {\em EURASIP Journal on Wireless Communications and Networking}, vol.
  2011, no. 1, pp. 785437, 2011.

\bibitem{zhang2016energy}
Ye~Zhang and Wei-Ping Zhu,
\newblock ``Energy-efficient pilot and data power allocation in massive mimo
  communication systems based on mmse channel estimation,''
\newblock in {\em 2016 IEEE International Conference on Acoustics, Speech and
  Signal Processing (ICASSP)}. IEEE, 2016, pp. 3571--3575.

\bibitem{wang2014ber}
Kezhi Wang, Yunfei Chen, Mohamed-Slim Alouini, and Feng Xu,
\newblock ``Ber and optimal power allocation for amplify-and-forward relaying
  using pilot-aided maximum likelihood estimation,''
\newblock {\em IEEE Transactions on Communications}, vol. 62, no. 10, pp.
  3462--3475, 2014.

\bibitem{wang2009super}
Gongpu Wang and Chintha Tellambura,
\newblock ``Super-imposed pilot-aided channel estimation and power allocation
  for relay systems,''
\newblock in {\em 2009 IEEE Wireless Communications and Networking Conference}.
  IEEE, 2009, pp. 1--6.

\bibitem{cai2004error}
Xiaodong Cai and Georgios~B Giannakis,
\newblock ``Error probability minimizing pilots for ofdm with m-psk modulation
  over rayleigh-fading channels,''
\newblock {\em IEEE transactions on vehicular technology}, vol. 53, no. 1, pp.
  146--155, 2004.

\bibitem{liu2016pilot}
Pei Liu, Shi Jin, Tao Jiang, Qi~Zhang, and Michail Matthaiou,
\newblock ``Pilot power allocation through user grouping in multi-cell massive
  mimo systems,''
\newblock {\em Ieee transactions on communications}, vol. 65, no. 4, pp.
  1561--1574, 2016.

\bibitem{zhao2017game}
Peiyue Zhao, G{\'a}bor Fodor, Gy{\"o}rgy D{\'a}n, and Mikl{\'o}s Telek,
\newblock ``A game theoretic approach to setting the pilot power ratio in
  multi-user mimo systems,''
\newblock {\em IEEE Transactions on Communications}, vol. 66, no. 3, pp.
  999--1012, 2017.

\bibitem{van2018joint}
Trinh Van~Chien, Emil Bj{\"o}rnson, and Erik~G Larsson,
\newblock ``Joint pilot design and uplink power allocation in multi-cell
  massive mimo systems,''
\newblock {\em IEEE Transactions on Wireless Communications}, vol. 17, no. 3,
  pp. 2000--2015, 2018.

\bibitem{zhou2010two}
Xiangyun Zhou, Tharaka~A Lamahewa, Parastoo Sadeghi, and Salman Durrani,
\newblock ``Two-way training: Optimal power allocation for pilot and data
  transmission,''
\newblock {\em IEEE Transactions on Wireless Communications}, vol. 9, no. 2,
  pp. 564--569, 2010.

\bibitem{rottenberg2016generalized}
Fran{\c{c}}ois Rottenberg, Fran{\c{c}}ois Horlin, Eleftherios Kofidis, and
  J{\'e}r{\^o}me Louveaux,
\newblock ``Generalized optimal pilot allocation for channel estimation in
  multicarrier systems,''
\newblock in {\em 2016 IEEE 17th International Workshop on Signal Processing
  Advances in Wireless Communications (SPAWC)}. IEEE, 2016, pp. 1--5.

\bibitem{xu2012energy}
Zhikun Xu, Geoffrey~Ye Li, Chenyang Yang, Shunqing Zhang, Yan Chen, and Shugong
  Xu,
\newblock ``Energy-efficient power allocation for pilots in training-based
  downlink ofdma systems,''
\newblock {\em IEEE Transactions on Communications}, vol. 60, no. 10, pp.
  3047--3058, 2012.

\bibitem{ma2003optimal}
Xiaoli Ma, Georgios~B Giannakis, and Shuichi Ohno,
\newblock ``Optimal training for block transmissions over doubly selective
  wireless fading channels,''
\newblock {\em IEEE Transactions on Signal Processing}, vol. 51, no. 5, pp.
  1351--1366, 2003.

\bibitem{khosravi2014joint}
Mahdi Khosravi and Saeed Mashhadi,
\newblock ``Joint pilot power \& pattern design for compressive ofdm channel
  estimation,''
\newblock {\em IEEE Communications Letters}, vol. 19, no. 1, pp. 50--53, 2014.

\bibitem{chen2003pilot}
Jiming Chen, Youxi Tang, and Shaoqian Li,
\newblock ``Pilot power allocation for ofdm systems,''
\newblock in {\em The 57th IEEE Semiannual Vehicular Technology Conference,
  2003. VTC 2003-Spring.} IEEE, 2003, vol.~2, pp. 1283--1287.

\bibitem{montalban2013power}
Rafael Montalban, Jos{\'e}~A L{\'o}pez-Salcedo, Gonzalo Seco-Granados, and
  A~Lee Swindlehurst,
\newblock ``Power allocation method based on the channel statistics for
  combined positioning and communications ofdm systems,''
\newblock in {\em 2013 IEEE International Conference on Acoustics, Speech and
  Signal Processing}. IEEE, 2013, pp. 4384--4388.

\bibitem{yu2016optimal}
Heejung Yu,
\newblock ``Optimal primary pilot power allocation and secondary channel
  sensing in cognitive radios,''
\newblock {\em IET Communications}, vol. 10, no. 5, pp. 487--494, 2016.

\bibitem{cheng2016optimal}
Hei~Victor Cheng, Emil Bj{\"o}rnson, and Erik~G Larsson,
\newblock ``Optimal pilot and payload power control in single-cell massive mimo
  systems,''
\newblock {\em IEEE Transactions on Signal Processing}, vol. 65, no. 9, pp.
  2363--2378, 2016.

\bibitem{dao2018pilot}
Hieu~Trong Dao and Sunghwan Kim,
\newblock ``Pilot power allocation for maximising the sum rate in massive mimo
  systems,''
\newblock {\em IET Communications}, vol. 12, no. 11, pp. 1367--1372, 2018.

\bibitem{ngo2014massive}
Hien~Quoc Ngo, Michail Matthaiou, and Erik~G Larsson,
\newblock ``Massive mimo with optimal power and training duration allocation,''
\newblock {\em IEEE Wireless Communications Letters}, vol. 3, no. 6, pp.
  605--608, 2014.

\bibitem{guo2015energy}
Kaifeng Guo, Yan Guo, and Gerd Ascheid,
\newblock ``Energy-efficient uplink power allocation in multi-cell
  mu-massive-mimo systems,''
\newblock in {\em Proceedings of European Wireless 2015; 21th European Wireless
  Conference}. VDE, 2015, pp. 1--5.

\bibitem{atitallah2017ber}
Ismail~Ben Atitallah, Christos Thrampoulidis, Abla Kammoun, Tareq~Y
  Al-Naffouri, Babak Hassibi, and Mohamed-Slim Alouini,
\newblock ``Ber analysis of regularized least squares for bpsk recovery,''
\newblock in {\em 2017 IEEE International Conference on Acoustics, Speech and
  Signal Processing (ICASSP)}. IEEE, 2017, pp. 4262--4266.

\bibitem{thrampoulidis2016ber}
Christos Thrampoulidis, Ehsan Abbasi, Weiyu Xu, and Babak Hassibi,
\newblock ``Ber analysis of the box relaxation for bpsk signal recovery,''
\newblock in {\em 2016 IEEE International Conference on Acoustics, Speech and
  Signal Processing (ICASSP)}. IEEE, 2016, pp. 3776--3780.

\bibitem{tan2001constrained}
Peng~Hui Tan, Lars~K Rasmussen, and Teng~J Lim,
\newblock ``Constrained maximum-likelihood detection in cdma,''
\newblock {\em IEEE Transactions on Communications}, vol. 49, no. 1, pp.
  142--153, 2001.

\bibitem{tikhonov1963solution}
AN~TIKHONOV,
\newblock ``Solution of incorrectly formaulated problems and the regularization
  method,''
\newblock in {\em Dokl. Akad. Nauk.}, 1963, vol. 151, pp. 1035--1038.

\bibitem{thrampoulidis2018precise}
Christos Thrampoulidis, Ehsan Abbasi, and Babak Hassibi,
\newblock ``Precise error analysis of regularized m-estimators in
  high-dimensions,''
\newblock {\em IEEE Transactions on Information Theory}, 2018.

\bibitem{alrashdi2017precise}
Ayed~M Alrashdi, Ismail~Ben Atitallah, Tareq~Y Al-Naffouri, and Mohamed-Slim
  Alouini,
\newblock ``Precise performance analysis of the lasso under matrix
  uncertainties,''
\newblock in {\em 2017 IEEE Global Conference on Signal and Information
  Processing (GlobalSIP)}. IEEE, 2017, pp. 1290--1294.

\bibitem{alrashdi2019precise}
Ayed~M Alrashdi, Ismail~Ben Atitallah, and Tareq~Y Al-Naffouri,
\newblock ``Precise performance analysis of the box-elastic net under matrix
  uncertainties,''
\newblock {\em IEEE Signal Processing Letters}, vol. 26, no. 5, pp. 655--659,
  2019.

\bibitem{kay1993fundamentals}
Steven~M Kay,
\newblock {\em Fundamentals of statistical signal processing},
\newblock Prentice Hall PTR, 1993.

\bibitem{practical_guide}
T.~Brown, Elisabeth {D. C}, and P.~Kyritsi,
\newblock {\em Practical Guide to the MIMO Radio Channel with MATLAB Examples:
  with MATLAB Examples},
\newblock Wiley, 1 edition, 2012.

\bibitem{yener2002cdma}
Aylin Yener, Roy~D Yates, and Sennur Ulukus,
\newblock ``Cdma multiuser detection: A nonlinear programming approach,''
\newblock {\em IEEE Transactions on Communications}, vol. 50, no. 6, pp.
  1016--1024, 2002.

\bibitem{ma2002quasi}
Wing-Kin Ma, Timothy~N Davidson, Kon~Max Wong, Zhi-Quan Luo, and Pak-Chung
  Ching,
\newblock ``Quasi-maximum-likelihood multiuser detection using semi-definite
  relaxation with application to synchronous cdma,''
\newblock {\em IEEE transactions on signal processing}, vol. 50, no. 4, pp.
  912--922, 2002.

\bibitem{grote2008ieee}
Alex Grote, Walter Grote, and Rodolfo Feick,
\newblock ``Ieee 802.11 goodput analysis for mixed real time and data
  traffic,''
\newblock in {\em Home Networking}, pp. 15--28. Springer, 2008.

\bibitem{hassibi2005sphere}
Babak Hassibi and Haris Vikalo,
\newblock ``On the sphere-decoding algorithm i. expected complexity,''
\newblock {\em IEEE transactions on signal processing}, vol. 53, no. 8, pp.
  2806--2818, 2005.

\bibitem{wu2012optimal}
Yihong Wu and Sergio Verd{\'u},
\newblock ``Optimal phase transitions in compressed sensing,''
\newblock {\em IEEE Transactions on Information Theory}, vol. 58, no. 10, pp.
  6241--6263, 2012.

\bibitem{kammoun_christos}
Z.~Deng, A.~Kammoun, and C.~Thrampoulidis,
\newblock ``High-dimensional binary linear classification,''
\newblock {\em submitted to information and inference: A journal of the IMA},
  2020.

\bibitem{rockafellar2015convex}
Ralph~Tyrell Rockafellar,
\newblock {\em Convex Analysis},
\newblock Princeton University Press, 2015.

\bibitem{svm_abla}
Abla Kammoun and Mohamed-Slim Alouini,
\newblock ``On the precise error analysis of support vector machines,''
\newblock {\em arXiv preprint arXiv:2003.12972}, 2020.

\bibitem{cox}
R.~D. Gill and P.~K. Andersen,
\newblock ``Cox's regression model for counting processes : a large sample
  study,''
\newblock {\em Ann. Statist.}, vol. 10, pp. 1100--1120, 1982.

\end{thebibliography}
\end{document}